\documentclass[twoside,leqno,twocolumn]{article}

\usepackage[letterpaper]{geometry}

\usepackage{siamproceedings}

\usepackage[switch]{lineno}

\usepackage[T1]{fontenc}
\usepackage{amsfonts}
\usepackage{graphicx}
\usepackage{epstopdf}
\usepackage{enumitem}
\usepackage{algorithmic}
\usepackage{colortbl}
\usepackage{makecell}
\usepackage{subcaption}
\usepackage{numprint}
\usepackage{siunitx}
\usepackage{listings}
\usepackage{multirow}
\usepackage{booktabs}
\usepackage{thmtools}
\usepackage{thm-restate}
\usepackage{nicefrac}
\usepackage{float}

\newcommand{\downuparrows}{\mathbin\downarrow\hspace{-.55em}\uparrow}

\usepackage{tikz}
\usetikzlibrary{calc,arrows.meta,positioning,fit,backgrounds}

\ifpdf
  \DeclareGraphicsExtensions{.eps,.pdf,.png,.jpg}
\else
  \DeclareGraphicsExtensions{.eps}
\fi

\newsiamremark{remark}{Remark}
\newsiamremark{hypothesis}{Hypothesis}
\crefname{hypothesis}{Hypothesis}{Hypotheses}
\newsiamthm{claim}{Claim}

\usepackage{amsopn}

\newcommand{\hhh}{\texttt{HHH}}
\newcommand{\egst}{\texttt{EGST}}

\newcommand{\hhho}{\texttt{HHH}$^\texttt{*}$}
\newcommand{\hhhoe}{\texttt{HHH}$^\texttt{*}$\texttt{-A}}
\newcommand{\hhhoh}{\texttt{HHH}$^\texttt{*}$\texttt{-H}}
\newcommand{\hhhohe}{\texttt{HHH}$^\texttt{*}$\texttt{-HA}}

\newcommand{\egsto}{\texttt{EGST}$^\texttt{*}$}
\newcommand{\hegsto}{\texttt{EGST}$^\texttt{*}$\texttt{-H}}
\newcommand{\degsto}{\texttt{EGST}$^\texttt{*}$\texttt{-D}}
\newcommand{\hdegsto}{\texttt{EGST}$^\texttt{*}$\texttt{-HD}}
\newcommand{\hegst}{\hegsto}
\newcommand{\degst}{\degsto}
\newcommand{\hdegst}{\hdegsto}

\newcommand{\Sc}[1]{S\textsubscript{#1}}

\newcommand{\epsT}{$\varepsilon$-Partition}
\newcommand{\RebalanceFactor}{\ensuremath{\rho}}
\newcommand{\hInd}{$h$-Index}

\newcommand{\hindexgr}[1]{\texttt{h(#1)}}
\newcommand{\hindexgrBest}{\hindexgr{$17/16$}}

\newcommand{\origepstab}[1]{\texttt{orig#1}} 
\newcommand{\softepstab}[1]{\texttt{soft#1}}
\newcommand{\lazyepstab}[1]{\texttt{lazy#1}}
\newcommand{\lateepstab}[1]{\texttt{late#1}}

\newcommand{\bigO}{\mathcal{O}}
\newcommand{\High}{\mathcal{H}}
\newcommand{\Low}{\mathcal{L}}
\newcommand{\AllV}{\mathcal{V}}

\newcommand{\DegreeNA}{\mathrm{deg}}
\newcommand{\Degree}[1]{\DegreeNA(#1)}
\newcommand{\DegreeL}[1]{\DegreeNA_{\Low}(#1)}
\newcommand{\DegreeH}[1]{\DegreeNA_{\High}(#1)}
\newcommand{\DegreeN}[2]{\DegreeNA_{\overline{\vphantom{\mathbf{x}}#1}}(#2)}

\newcommand{\NeighSym}{\mathrm{N}}
\newcommand{\Neigh}[1]{\NeighSym (#1)}
\newcommand{\NeighL}[1]{\NeighSym_{\Low}(#1)}
\newcommand{\NeighH}[1]{\NeighSym_{\High}(#1)}
\newcommand{\NeighN}[2]{\NeighSym_{\overline{\vphantom{\mathbf{x}}#1}}(#2)}
\newcommand{\NeighNH}[2]{\NeighSym_{\High,\overline{\vphantom{\mathbf{x}}#1}}(#2)}
\newcommand{\NeighNL}[2]{\NeighSym_{\Low,\overline{\vphantom{\mathbf{x}}#1}}(#2)}

\newcommand{\AuxvLV}{\ensuremath{\mathrm{vLV}}}
\newcommand{\Auxt}{\ensuremath{\mathrm{t}}}
\newcommand{\AuxuLv}{\ensuremath{\mathrm{uLv}}}
\newcommand{\AuxuLLv}{\ensuremath{\mathrm{uLLv}}}
\newcommand{\AuxcLV}{\ensuremath{\mathrm{cLV}}}
\newcommand{\AuxpLL}{\ensuremath{\mathrm{pLL}}}
\newcommand{\AuxuHv}{\ensuremath{\mathrm{uHv}}}
\newcommand{\AuxcL}{\ensuremath{\mathrm{cL}}}

\newcommand{\Partition}{\ensuremath{\mathcal{P}}}

\newcommand{\MaxDegree}{D}

\newcommand{\deglow}{\MaxDegree_\mathcal{L}}

\newcommand{\changesperupdate}{\#_{\downuparrows}}

\newcommand{\ie}{i.e.}
\newcommand{\eg}{e.g.}
\newcommand{\wlogtext}{w.l.o.g.}

\definecolor{hlc}{HTML}{D81B60}
\newcommand{\bgcolortext}[2]{\tikz[baseline=(n.base)] 
\node (n) [fill=#1!30,rounded corners=2pt,inner sep=2pt,outer sep=0pt] {#2};}

\newcommand{\FullSet}{\textsf{FS}}
\newcommand{\EngSet}{\textsf{ES}}

\newcommand{\DSTMP}{\textsc{TMP}}
\newcommand{\DSWiki}{\textsc{Wiki}}
\newcommand{\DScKron}{\textsc{KRONc}}
\newcommand{\DSgKron}{\textsc{KRONg}}

\newcommand{\baselineGraphflow}{\texttt{Graphflow}}
\newcommand{\baselineTurboFlux}{\texttt{TurboFlux}}
\newcommand{\baselineIEDyn}{\texttt{IEDyn}}
\newcommand{\baselineSymBi}{\texttt{SymBi}}
\newcommand{\baselineDNA}{\texttt{DNA}}
\newcommand{\baselineRapidFlow}{\texttt{RapidFlow}}

\makeatletter
\newcommand{\popt}[3][]{\par\medskip
  \noindent\textit{(#2) #3. }\if\relax\detokenize{#1}\relax\else
    \protected@edef\@currentlabel{#2}\label{popt:#1}\fi
  \ignorespaces
}
\makeatother
\newcommand{\poptref}[1]{(\ref{popt:#1})}

\newcounter{hhhoptimization}
\renewcommand{\thehhhoptimization}{$O_A$-\Roman{hhhoptimization}}

\newcounter{epsoptimization}
\renewcommand{\theepsoptimization}{$O_P$-\roman{epsoptimization}}

\newcounter{hindoptimization}

\newcounter{eoptimization}

\newcommand{\hhhopt}[2][]{\refstepcounter{hhhoptimization}\par\medskip
  \noindent\textit{(\thehhhoptimization) #2. }\if\relax\detokenize{#1}\relax\else\label{#1}\fi
  \ignorespaces
}

\newcommand{\epsopt}[2][]{\refstepcounter{epsoptimization}\par\medskip
  \noindent\textit{(\theepsoptimization) #2. }\if\relax\detokenize{#1}\relax\else\label{#1}\fi
  \ignorespaces
}

\newcommand{\egstopt}[2][]{\refstepcounter{hhhoptimization}\par\medskip
  \noindent\textit{(\thehhhoptimization) #2. }\if\relax\detokenize{#1}\relax\else\label{#1}\fi
  \ignorespaces
}

\newcommand{\hindopt}[2][]{\refstepcounter{epsoptimization}\par\medskip
  \noindent\textit{(\theepsoptimization) #2. }\if\relax\detokenize{#1}\relax\else\label{#1}\fi
  \ignorespaces
}

\newif\iflongversion
\longversiontrue

\newcommand{\restatethelemma}[1]{\begingroup
  \renewcommand{\thetheorem}{\ref{lemma:#1}}\csname #1\endcsname*
  \endgroup
}

\tikzset{
  smallmarkarrow/.style={red,-{Stealth[length=0.8mm]},line width=0.5},
  label/.style={font=\tiny},
}

\newcommand{\DrawvLV}{\begin{tikzpicture}[baseline=0.05ex, scale=0.2]
  \coordinate (A) at (0,0);
  \coordinate (B) at (1,0);
  \coordinate (C) at (0.5,0.866); 
  \draw (A) -- (C) -- (B);

  \filldraw[fill=white, draw=black] (A) circle (7pt);
  \draw[smallmarkarrow] (A) ++(0,-0.5) -- ++(0,0.25);
  \filldraw[fill=gray!70, draw=black] (B) circle (5.5pt);
  \filldraw[fill=black, draw=black] (C) circle (4pt);
\end{tikzpicture}}

\newcommand{\DrawuLv}{\begin{tikzpicture}[baseline=0.05ex, scale=0.2,transform shape]
  \coordinate (A) at (0,0);
  \coordinate (B) at (1,0);
  \coordinate (C) at (0.5,0.866); 
  \draw (A) -- (C) -- (B) ;

  \filldraw[fill=white, draw=black] (A) circle (7pt);
  \draw[smallmarkarrow] (A) ++(0,-0.5) -- ++(0,0.25);
  \filldraw[fill=white, draw=black] (B) circle (7pt);
  \draw[smallmarkarrow] (B) ++(0,-0.5) -- ++(0,0.25);
  \filldraw[fill=black, draw=black] (C) circle (4pt);
  
\end{tikzpicture}}

\newcommand{\DrawuHv}{\begin{tikzpicture}[baseline=0.05ex, scale=0.2,transform shape]
  \coordinate (A) at (0,0);
  \coordinate (B) at (1,0);
  \coordinate (C) at (0.5,0.866); 
  \draw (A) -- (C) -- (B) ;

  \filldraw[fill=white, draw=black] (A) circle (7pt);
  \draw[smallmarkarrow] (A) ++(0,-0.5) -- ++(0,0.25);
  \filldraw[fill=white, draw=black] (B) circle (7pt);
  \draw[smallmarkarrow] (B) ++(0,-0.5) -- ++(0,0.25);
  \filldraw[fill=white, draw=black] (C) circle (7pt);
  
\end{tikzpicture}}

\newcommand{\Drawt}{\begin{tikzpicture}[baseline=0.05ex, scale=0.2]
  \coordinate (A) at (0,0);
  \coordinate (B) at (1,0);
  \coordinate (C) at (0.5,0.866); 
  \draw (A) -- (C) -- (B) -- (A);

  \filldraw[fill=white, draw=black] (A) circle (7pt);
  \draw[smallmarkarrow] (A) ++(0,-0.5) -- ++(0,0.25);
  \filldraw[fill=gray!70, draw=black] (B) circle (5.5pt);
  \filldraw[fill=gray!70, draw=black] (C) circle (5.5pt);
\end{tikzpicture}}

\newcommand{\DrawuLLv}{\begin{tikzpicture}[baseline=0.05ex, scale=0.2,transform shape]
  \coordinate (A) at (0,0);
  \coordinate (B) at (0,1);
  \coordinate (C) at (1,1);
  \coordinate (D) at (1,0);

  \draw (A) -- (B) -- (C) -- (D);

  \filldraw[fill=white, draw=black] (A) circle (7pt);
  \draw[smallmarkarrow] (A) ++(0,-0.5) -- ++(0,0.25);
  \filldraw[fill=white, draw=black] (D) circle (7pt);
  \draw[smallmarkarrow] (D) ++(0,-0.5) -- ++(0,0.25);
  \filldraw[fill=black, draw=black] (B) circle (4pt);
  \filldraw[fill=black, draw=black] (C) circle (4pt);
  
\end{tikzpicture}}

\newcommand{\DrawcLV}{\begin{tikzpicture}[baseline=0.05ex, scale=0.2,transform shape]
  \coordinate (A) at (0,0);
  \coordinate (B) at (0,1);
  \coordinate (C) at (1,1);
  \coordinate (D) at (1,0);

  \draw (A) -- (B) -- (C);
  \draw (B) -- (D);

  \filldraw[fill=white, draw=black] (A) circle (7pt);
  \draw[smallmarkarrow] (A) ++(0,-0.5) -- ++(0,0.25);
  \filldraw[fill=white, draw=black] (D) circle (7pt);
  \draw[smallmarkarrow] (D) ++(0,-0.5) -- ++(0,0.25);
  \filldraw[fill=black, draw=black] (B) circle (4pt);
  \filldraw[fill=gray!70, draw=black] (C) circle (5.5pt);
  
\end{tikzpicture}}

\newcommand{\DrawpLL}{\begin{tikzpicture}[baseline=0.05ex, scale=0.2,transform shape]
  \coordinate (A) at (0,0);
  \coordinate (B) at (0,1);
  \coordinate (C) at (1,1);
  \coordinate (D) at (1,0);

  \draw (A) -- (B) -- (C);
  \draw (B) -- (D) -- (A);

  \filldraw[fill=white, draw=black] (A) circle (7pt);
  \draw[smallmarkarrow] (A) ++(0,-0.5) -- ++(0,0.25);
  \filldraw[fill=white, draw=black] (C) circle (7pt);
  \draw[smallmarkarrow] (C) ++(0,-0.5) -- ++(0,0.25);
  \filldraw[fill=black, draw=black] (B) circle (4pt);
  \filldraw[fill=black, draw=black] (D) circle (4pt);
  
\end{tikzpicture}}

\newcommand{\DrawcL}{\begin{tikzpicture}[baseline=0.05ex, scale=0.2,transform shape]
  \coordinate (A) at (0,0);
  \coordinate (B) at (0,1);
  \coordinate (C) at (1,1);
  \coordinate (D) at (1,0);

  \draw (A) -- (B) -- (C);
  \draw (B) -- (D);

  \filldraw[fill=white, draw=black] (A) circle (7pt);
  \draw[smallmarkarrow] (A) ++(0,-0.5) -- ++(0,0.25);
  \filldraw[fill=white, draw=black] (D) circle (7pt);
  \draw[smallmarkarrow] (D) ++(0,-0.5) -- ++(0,0.25);
  \filldraw[fill=white, draw=black] (C) circle (7pt);
  \draw[smallmarkarrow] (C) ++(0,-0.5) -- ++(0,0.25);
  \filldraw[fill=black, draw=black] (B) circle (4pt);
  
\end{tikzpicture}}

\newcommand{\dia}{\begin{tikzpicture}[baseline=0.05ex]
\draw (0,0) rectangle (0.15,0.15);
\draw (0,0) -- (0.15,0.15);
\end{tikzpicture}}

\newcommand{\fClique}{\begin{tikzpicture}[baseline=0.05ex]
\draw (0,0) rectangle (0.15,0.15);
\draw (0,0) -- (0.15,0.15);
\draw (0,0.15) -- (0.15,0);
\end{tikzpicture}}

\newcommand{\fCycle}{\begin{tikzpicture}[baseline=0.05ex]
\draw (0,0) rectangle (0.15,0.15);
\end{tikzpicture}}

\newcommand{\paw}{\begin{tikzpicture}[baseline=0.05ex]
\draw (0,0) -- (0,0.15);
\draw (0,0) -- (0.15,0);
\draw (0,0.15) -- (0.15,0);
\draw (0,0.15) -- (0.15,0.15);
\end{tikzpicture}}
\newcommand{\claw}{\begin{tikzpicture}[baseline=0.05ex]
\draw (0,0) -- (0,0.15);
\draw (0,0.15) -- (0.15,0);
\draw (0,0.15) -- (0.15,0.15);
\end{tikzpicture}}
\newcommand{\tPath}{\begin{tikzpicture}[baseline=0.05ex]
\draw (0,0) -- (0,0.15);
\draw (0,0.15) -- (0.15,0.15);
\draw (0.15,0.15) -- (0.15,0);
\end{tikzpicture}}

\newcommand{\tri}{\begin{tikzpicture}[baseline=0.05ex]
\draw (0,0) -- (0.075,0.15);
\draw (0,0) -- (0.15,0);
\draw (0.075,0.15) -- (0.15,0);
\end{tikzpicture}}

\newcommand{\DrawSZero}{\begin{tikzpicture}[baseline=0.05ex, scale=0.2]
  \coordinate (A) at (0,0);
  \coordinate (B) at (1,0);
  \coordinate (C) at (0.5,0.866); 
  \draw (A) -- (C) -- (B);

  \filldraw[fill=white, draw=black] (A) circle (7pt);
  \draw[smallmarkarrow] (A) ++(0,-0.5) -- ++(0,0.25);
  \filldraw[fill=black, draw=black] (B) circle (5.5pt);
  \filldraw[fill=black, draw=black] (C) circle (4pt);
\end{tikzpicture}}

\newcommand{\DrawSOne}{\begin{tikzpicture}[baseline=0.05ex, scale=0.2,transform shape]
  \coordinate (A) at (0,0);
  \coordinate (B) at (1,0);
  \coordinate (C) at (0.5,0.866); 
  \draw (A) -- (C) -- (B) -- (A);

  \filldraw[fill=white, draw=black] (A) circle (7pt);
  \draw[smallmarkarrow] (A) ++(0,-0.5) -- ++(0,0.25);
  \filldraw[fill=black, draw=black] (B) circle (5.5pt);
  \filldraw[fill=black, draw=black] (C) circle (4pt);
  
\end{tikzpicture}}

\newcommand{\DrawSTwo}{\DrawuLv
}

\newcommand{\DrawSThreeH}{\DrawuLLv
}
\newcommand{\DrawSThreeAny}{\begin{tikzpicture}[baseline=0.05ex, scale=0.2,transform shape]
  \coordinate (A) at (0,0);
  \coordinate (B) at (0,1);
  \coordinate (C) at (1,1);
  \coordinate (D) at (1,0);

  \draw (A) -- (B) -- (C) -- (D);

  \filldraw[fill=white, draw=black] (A) circle (7pt);
  \draw[smallmarkarrow] (A) ++(0,-0.5) -- ++(0,0.25);
  \filldraw[fill=white, draw=black] (D) circle (7pt);
  \draw[smallmarkarrow] (D) ++(0,-0.5) -- ++(0,0.25);
  \filldraw[fill=gray!70, draw=black] (B) circle (5.5pt);
  \filldraw[fill=gray!70, draw=black] (C) circle (5.5pt);
  
\end{tikzpicture}}

\newcommand{\DrawSFour}{\begin{tikzpicture}[baseline=0.05ex, scale=0.2,transform shape]
  \coordinate (A) at (0,0);
  \coordinate (B) at (0,1);
  \coordinate (C) at (1,1);
  \coordinate (D) at (1,0);

  \draw (A) -- (B) -- (C);
  \draw (B) -- (D);

  \filldraw[fill=white, draw=black] (A) circle (7pt);
  \draw[smallmarkarrow] (A) ++(0,-0.5) -- ++(0,0.25);
  \filldraw[fill=white, draw=black] (D) circle (7pt);
  \draw[smallmarkarrow] (D) ++(0,-0.5) -- ++(0,0.25);
  \filldraw[fill=black, draw=black] (B) circle (4pt);
  \filldraw[fill=black, draw=black] (C) circle (5.5pt);
  
\end{tikzpicture}}

\newcommand{\DrawSFive}{\DrawpLL
}

\newcommand{\DrawSSix}{\begin{tikzpicture}[baseline=0.05ex, scale=0.2,transform shape]
  \coordinate (A) at (0,0);
  \coordinate (B) at (0,1);
  \coordinate (C) at (1,1);
  \coordinate (D) at (1,0);

  \draw (A) -- (B) -- (C);
  \draw (B) -- (D) -- (A);
  \draw (C) -- (D);

  \filldraw[fill=white, draw=black] (A) circle (7pt);
  \draw[smallmarkarrow] (A) ++(0,-0.5) -- ++(0,0.25);
  \filldraw[fill=white, draw=black] (C) circle (7pt);
  \draw[smallmarkarrow] (C) ++(0,-0.5) -- ++(0,0.25);
  \filldraw[fill=black, draw=black] (B) circle (4pt);
  \filldraw[fill=black, draw=black] (D) circle (4pt);
  
\end{tikzpicture}}

\newcommand{\DrawSSeven}{\DrawcL
}

\begin{document}

\title{
Fully Dynamic Triangle and 4-Vertex Subgraph Counting: \\From Theory to
Practice and Back Again
}
\author{Kathrin Hanauer\thanks{Faculty of Computer Science, University of Vienna, Vienna, Austria; Email: \email{kathrin.hanauer@univie.ac.at}, \email{sophia.heck@univie.ac.at},
\email{leonhard.paul.sidl@univie.ac.at}
}
\thanks{Research Network Data Science, University of Vienna, Vienna, Austria}
\and
Sophia Heck\footnotemark[1]
\thanks{Doctoral School Computer Science DoCS, University of Vienna, Vienna, Austria}
\and
Monika Henzinger\thanks{Institute of Science and Technology Austria (ISTA), Klosterneuburg,
Austria; Email: \email{monika.henzinger@ista.ac.at}.}
\and
Leonhard Sidl\footnotemark[1]
\thanks{Department of Theoretical Chemistry, University of Vienna, Vienna, Austria; and Vienna Doctoral School in Chemistry, University of Vienna, Vienna, Austria }
}
\date{}

\maketitle

\begingroup
\renewcommand{\thefootnote}{}
\footnotetext{Extended author version of the ALENEX 2027 proceedings paper~\cite{HanauerHHS27:short} (to appear).}
\endgroup

\fancyfoot[R]{\scriptsize{Copyright \textcopyright\ 2027\\
Copyright for this paper is retained by authors}}

\begin{abstract}

Counting 
small patterns is a core primitive in large-scale graph analytics, with numerous applications in data management and network analysis.
As modern graphs are massive and continuously evolving, maintaining such counts efficiently under updates is essential.
Although specialized dynamic algorithms for patterns with up to four vertices exist, they remain largely theoretical; practitioners instead rely on general-purpose or approximate methods.

We present an in-depth study of two dynamic approaches by Eppstein et al.\ and Hanauer et al.\ for maintaining counts of triangles and all connected four-vertex subgraph patterns, both based on variants of heavy-light vertex partitioning.
We introduce theoretical and practical optimizations and develop hybrid variants that combine complementary strengths of both approaches.

Experiments on \num{56} large real-world and synthetic dynamic graphs show that our optimized algorithms outperform the original methods and strong static and dynamic baselines on nearly all instances, both in total running time and worst-case update time, with speedups of up to several orders of magnitude. We further provide practical guidelines for algorithm selection and tuning. To the best of our knowledge, this is the first comprehensive empirical study of heavy-light partitioning for dynamic subgraph counting, and our findings may inform other algorithms based on this technique.

\end{abstract}

\section{Introduction}
Subgraph counting is a fundamental problem in network analysis.
The frequency with which a (smaller) \emph{pattern graph} occurs as a subgraph in a larger \emph{data graph} captures essential information about local structures and serves as a fundamental graph statistic with numerous applications:
For example, the clustering coefficient, which relates the number of triangles in a graph to the number of length-2 paths, measures transitivity~\cite{WS98}.
In databases, the number of embeddings of a fixed pattern graph corresponds to the result size of the respective conjunctive join query.
\emph{Motifs}, i.e., (small) patterns that occur significantly more often than expected under an appropriate random model, are at the core of
social network analysis~\cite{DBLP:journals/internet/KusenS21},
recommendation systems~\cite{DBLP:journals/pvldb/GuptaSGGZLL14},
recognition of hidden structures in complex networks~\cite{milo2002network},
bot identification~\cite{DBLP:journals/ipm/KusenS20},
fraud detection in financial networks~\cite{DBLP:journals/pvldb/QiuCQPZLZ18},
and
cyber security~\cite{DBLP:conf/edbt/ChoudhuryHCAF15}.

In practice, data very often is not static.
Instead, it is updated continuously, giving rise to \emph{dynamic algorithms} that maintain subgraph counts under edge insertions or deletions.
Dynamic algorithms with strong theoretical guarantees are known for triangles~\cite{DBLP:conf/icdt/KaraNNOZ19,DBLP:journals/jgaa/EppsteinS12} as well as for all connected four-vertex patterns~\cite{DBLP:conf/sand/HanauerHH22,DBLP:journals/tcs/EppsteinGST12};
all of them are based on variants of \emph{heavy-light partitioning}, which exploits that the degree distribution in many real-world graphs is highly skewed.

Despite their relevance, maintaining the exact frequencies of triangles and four-vertex patterns in a dynamically evolving graph with both edge insertions and deletions has not received much attention in empirical research so far.
Previous work has either focused on approximate (often triangle) counting~\cite{DBLP:journals/algorithmica/BulteauFKP16,DBLP:journals/tkdd/StefaniERU17,DBLP:conf/asunam/HanS17,DBLP:conf/kdd/LimK15,DBLP:journals/pvldb/AhmedDWR17,DBLP:journals/bmcsb/MukherjeeHBK18},
considers a more general approach for counting or enumerating arbitrary patterns~\cite{DBLP:journals/pvldb/AhmedDWR17,DBLP:journals/bmcsb/MukherjeeHBK18,10.1093/nar/gkg340,fan2013incremental,kim2018turboflux,DBLP:journals/pvldb/MinPPGIH21}, and/or is restricted to cliques~\cite{DBLP:conf/hipc/MakkarBG17,DBLP:conf/apocs/DhulipalaLSY21}.
To the best of our knowledge, only Schiller et al.~\cite{DBLP:conf/alcob/SchillerJHS15} empirically study four-vertex subgraph counting in the dynamic setting, however without considering sophisticated approaches with strong theoretical guarantees, such as those based on vertex partitioning.
Furthermore, they only consider \emph{induced} subgraph counts,
where also non-edges must be preserved in the data graph.
\emph{Non-induced} counts, in contrast, allow additional edges and admit faster algorithms.
The two are mutually convertible
via formulas, which however require several other pattern counts, including the computationally more expensive cliques~\cite{DBLP:conf/sand/HanauerHH22}.

In this paper, we close the gap between theory and practice.
We implement and further optimize the non-induced four-vertex subgraph counting algorithms by Hanauer et al.~\cite{DBLP:conf/sand/HanauerHH22}
and Eppstein et al.~\cite{DBLP:journals/tcs/EppsteinGST12,DBLP:journals/jgaa/EppsteinS12}, along with the triangle counting algorithms based on their respective frameworks,
which, in case of Hanauer et al.~\cite{DBLP:conf/sand/HanauerHH22}, is conceptually similar to the algorithm of Kara et al.~\cite{DBLP:conf/icdt/KaraNNOZ19}.
As both can count 4-cliques, they can also report induced counts.
We consider triangles as well as all connected four-vertex pattern graphs. 
Specifically, we make the following \emph{contributions}:
\begin{enumerate}[itemsep=-1ex,partopsep=1ex,parsep=1ex]
\item
We identify and correct errors and omissions in the description of these algorithms.
\item 
We optimize auxiliary data structures used internally by the algorithms to achieve speedups in practice, while retaining or even improving the theoretical guarantees.
\item We describe and experiment with new ways to maintain a vertex partitioning under edge updates.
\item We cross-combine the algorithm by Eppstein et al.\ with the partitioning scheme by Hanauer et al.\ and vice-versa and analyze the resulting running times theoretically.
\item In an extensive experimental evaluation
on \num{56} large instances,
our optimized algorithms achieve mean speedups
of up to four orders of magnitude over static recomputation,
up to three orders over the only dynamic counting baseline,
and up to $26\times$ over state-of-the-art continuous matching algorithms.
\end{enumerate}

\section{Related Work}\label{sect:relwork_short}
The high relevance and ubiquity of subgraph counting and matching in diverse areas 
have given rise to a large body of work on general and specialized algorithms
as well as theoretical results.
Here, we restrict ourselves to results closest to ours and refer the reader 
to \autoref{sect:relwork_long} as well as
the detailed surveys on static~\cite{DBLP:journals/csur/RibeiroPSAS21} and dynamic~\cite{DBLP:journals/jea/HanauerHS22} subgraph counting for a broader overview. 

In the static setting, counting for \emph{arbitrary} patterns generalizes the $\mathcal{NP}$-hard subgraph isomorphism problem.
For patterns on up to five vertices, efficient specialized algorithms exist~\cite{DBLP:journals/ipl/KloksKM00,DBLP:journals/bioinformatics/HocevarD14,DBLP:journals/ans/OrtmannB17,DBLP:conf/www/PinarSV17},
alongside a large body of work on the
identification of motifs~\cite{milo2002network,itzhack2007optimal,DBLP:conf/wabi/Wernicke05,li2012netmode,DBLP:journals/tcbb/WangHWP12,DBLP:conf/sac/RibeiroS10},
and on approximate and streaming algorithms, especially for triangles~\cite{DBLP:conf/cocoon/JowhariG05,DBLP:conf/pods/BuriolFLMS06,DBLP:conf/kdd/LimK15,DBLP:journals/tkdd/LimJK18,DBLP:conf/pods/McGregorVV16,DBLP:journals/siamsc/KoldaPPST14,DBLP:journals/pvldb/PavanTTW13}.

In the dynamic setting,
Eppstein and Spiro~\cite{DBLP:journals/jgaa/EppsteinS12} maintain the triangle count in $\bigO(h)$ amortized update time, where the \hInd{} $h$ is the largest integer such that the graph has $h$ vertices of degree at least $h$. Note that $h \in \bigO(\sqrt{m})$.
Eppstein et al.~\cite{DBLP:journals/tcs/EppsteinGST12} extend this approach to non-induced subgraphs of size four with an amortized update time of $\bigO(h^2)$.
We refer to this algorithm as \egst{} and provide a more detailed description in \autoref{sect:egst}.
Kara et al.~\cite{DBLP:conf/icdt/KaraNNOZ19} show how to count triangles in $\bigO(\sqrt{m})$ amortized update time based on a heavy-light vertex partitioning.
Using a similar partitioning, Hanauer et al.~\cite{DBLP:conf/sand/HanauerHH22} give improved algorithms to count non-induced 4-vertex patterns in amortized time $\bigO(m^\frac{2}{3})$ except for the clique and the path, which take $\bigO(m)$ and $\bigO(\sqrt{m})$ update time, respectively.
We later refer to this algorithm as \hhh{} and also include their triangle counting algorithm, which is conceptually very similar to that by Kara et al.~\cite{DBLP:conf/icdt/KaraNNOZ19} and has the same guarantees.
Hanauer et al.\ further show that,
conditioned on the combinatorial k-clique conjecture~\cite{DBLP:journals/siamcomp/AbboudBW18} and the OMv conjecture~\cite{DBLP:conf/stoc/HenzingerKNS15}, respectively, substantially faster fully dynamic algorithms for induced and non-induced connected 4-vertex patterns 
do not exist.
Further theoretical results tighten these bounds for triangles~\cite{DBLP:conf/icdt/LuT21} and 4-cycles~\cite{DBLP:journals/pacmmod/AssadiS25}.
The only prior empirical study of dynamic four-vertex subgraph counting is by Schiller et al.~\cite{DBLP:conf/alcob/SchillerJHS15}, who maintain \emph{induced} counts via neighborhood scanning in $\bigO(\MaxDegree^2)$ update time,
where $\MaxDegree$ denotes the maximum vertex degree.
Finally, there also exists a large body of empirical work on continuous subgraph matching (CSM) for arbitrary patterns~\cite{fan2013incremental,kim2018turboflux,DBLP:journals/pvldb/MinPPGIH21,DBLP:conf/sigmod/KankanamgeSMCS17,DBLP:conf/sigmod/IdrisUV17,DBLP:journals/vldb/IdrisUVVL20,DBLP:journals/pvldb/SunSHL22}, from which frequencies can be extracted straightforwardly.
Sun et al.~\cite{DBLP:journals/pvldb/SunSLH22} provide an extensive survey and experimental study of these algorithms.

\section{Preliminaries}
Let $G=(V,E)$ be a simple undirected graph with $n = |V|$ and $m= |E|$.
The \emph{degree} $\Degree{v}$ of a vertex $v$ is the number of incident edges of $v$, and $\MaxDegree := \max_{v \in V} \Degree{v}$.
Let $N(v) := \{u \mid \{u, v\} \in E\}$ denote the neighborhood of $v$.
We define $N_{\overline{u}}(v) := N(v) \setminus \{u\}$ and $\DegreeN{u}{v} := |N_{\overline{u}}(v)|$.
For a set of vertices $X \subseteq V$ and some vertex $v \in V$,
we use $N_X(v) := N(v) \cap X$, 
$\DegreeNA_X(v) := |N_X(v)|$, $N_{X, \overline{u}}(v) := N_X(v) \setminus \{u\}$, and
$\MaxDegree_X := \max_{v \in X} \Degree{v}$.

A graph $G' = (V', E')$ is a (non-induced) \emph{subgraph} of $G$ if $V' \subseteq V$ and $E' \subseteq E$, and an \emph{induced subgraph} if $E' = \{ \{u,v\} \in E \mid u, v \in V'\}$.
A \emph{pattern} is a (small) graph that we are interested in finding in a larger graph.
Given a pattern $P$, its induced or non-induced count equals the number of induced or non-induced subgraphs in $G$, respectively, that are isomorphic to $P$, divided by the number of automorphisms of $P$.
We consider the following patterns (see also \autoref{fig:hhh_structures_new}):
triangle, 3-path, claw, paw, 4-cycle, diamond, and 4-clique.

A \emph{dynamic graph} is a (possibly empty) initial graph along with a sequence of edge insertions and deletions, also called \emph{updates}, which are applied to the graph in order.
These updates are unknown to a dynamic graph algorithm beforehand.
When considering a dynamic graph, $n$ and $m$ always refer to the current number of vertices and edges, respectively, \ie, some updates may already have been applied.
For the sake of simplicity, we do not explicitly consider node updates.
They can be handled straightforwardly.

A (heavy-light) \emph{vertex partition} $\Partition = (\High, \Low)$ partitions the vertex set $V$ into two disjoint subsets, where vertices with high degree (for some notion of high degree) are in $\High$ and all others are in $\Low$.
All partitions in the following can be maintained in amortized constant time per update.
We denote by $\changesperupdate$ the amortized number of vertices changing partition per update.

\section{Triangle and 4-Vertex Subgraph Counting}
We provide an overview of the only two algorithms we are aware of for counting all connected non-induced four-vertex subgraphs in dynamic graphs,
\egst{}~\cite{DBLP:journals/tcs/EppsteinGST12} and
\hhh~\cite{DBLP:conf/sand/HanauerHH22}.
Both can also be extended to count triangles~\cite{DBLP:journals/jgaa/EppsteinS12}.
We introduce several variants and optimizations targeted at improving their practical performance.
As claws can be counted in worst-case constant time~\cite{DBLP:conf/sand/HanauerHH22}, we do not discuss them further in this section.

\subsection{Common Overall Architecture}
Despite significant differences in their detailed strategies, \egst{} and \hhh{} share the same overarching design:

\begin{itemize}[label={}, leftmargin=0em, itemsep=0.2em, parsep=-1pt,topsep=-.5pt]
    \item 
\textbf{Vertex Partitioning} maintains disjoint high- and low-degree vertex sets, $\High$ and $\Low$.
    \item 
\textbf{Auxiliary Counts}
record occurrences of restricted patterns indexed by up to three anchor vertices. Restrictions specify whether vertices belong to $\High$ or $\Low$.
    \item 
\textbf{Final Pattern Counts} are obtained from a combination of auxiliary counts and neighborhood scanning.
\end{itemize}

Upon an edge insertion both algorithms perform the three steps in the order stated above.
After an edge deletion, the order of auxiliary structure updates and final pattern counts is swapped.
Below, we provide a brief overview of each element and refer to \cite{DBLP:conf/sand/HanauerHH22} and \cite{DBLP:journals/tcs/EppsteinGST12} for a complete description.
We assume that auxiliary counts are stored using fast hash maps
with expected $\bigO(1)$ access and modification time.

\subsection{\hhh{} Algorithm and \epsT{}}
We briefly review the \epsT{} as well as the \hhh{} algorithm~\cite{DBLP:conf/sand/HanauerHH22}
and describe a correction.
We then discuss several practical optimization strategies.
\subsubsection{$\varepsilon$-Partition}
\hhh{} uses an \epsT{} to define $\High$ and $\Low = \AllV \setminus \High$.
The partitioning scheme depends on a parameter $\varepsilon \in [0,1]$, which is
chosen such that the desired objective (running time or space) is minimized,
and the partition is re-initialized whenever the graph has grown or shrunken (w.r.t.\ $m$) substantially.
We refer to the time between two re-initializations as an \emph{epoch}.
More precisely,
for the $i$th epoch,
we set a degree threshold
$\theta_i := (2m_i)^\varepsilon$, where $m_i$ is the number of edges in the graph at the beginning of the epoch, and assign a vertex $v$ to $\High$ if $\Degree{v} \ge \theta_i$ and to $\Low$ otherwise.
Afterwards, during the same epoch, $v$ changes partitions only if $\Degree{v} \ge \frac{3}{2} \theta_i$ or $\Degree{v} < \frac{1}{2} \theta_i$.
We call such a change a \emph{vertex partition update (VPU)} and refer to \cite{DBLP:conf/sand/HanauerHH22} for details. 
The next epoch $i+1$ starts as soon as the current number of edges $m$ no longer satisfies $\lfloor\frac{m_i}{2}\rfloor \leq m < 2m_i$.
The transition between two epochs requires a \emph{major rebalancing (MR)} step,
where all vertices are re-classified according to the new threshold~$\theta_{i+1}$.
There are $\Omega(m)$ graph updates (insertions or deletions) in each epoch and at most two VPUs per graph update,
\ie, 
the number of changes per graph update
$\changesperupdate$ is in $\bigO(m^{-\varepsilon})$.
The \epsT{} further guarantees that $|\High| \in\bigO (m^{1-\varepsilon})$ and $\deglow \in \bigO (m^\varepsilon)$ at all times~\cite{DBLP:conf/sand/HanauerHH22}.

\subsubsection{\hhh{} Algorithm}\label{sec:hhh}
\hhh{} explicitly maintains the counts of triangles and 4-vertex subgraphs after each edge update.
Depending on the pattern(s) to count, it selects a set of auxiliary structures from the top row of \autoref{fig:hhh_structures_new} and maintains their counts.
As for the patterns, the auxiliary counts refer to non-induced subgraph counts, \ie, additional edges may exist.
With a slight abuse of notation we use the structure name to denote its count, i.e., $\AuxvLV[v]$ denotes the number of \AuxvLV{} structures with anchor vertex $v$.
Whenever an auxiliary structure appears or disappears, \ie, after
a VPU
or after an edge change in the graph, \hhh{} updates the corresponding auxiliary counts for all affected anchor vertices.
During an MR step, all auxiliary structures are recomputed from scratch by simulating edge insertions for all existing edges.
After a change in the graph, \hhh{} additionally updates the pattern counts based on auxiliary counts.
This results in an amortized update time of
$\bigO (m^{1/2})$ for triangles and 3-paths
and $\bigO(m^{2/3})$ for all remaining connected four-vertex non-induced subgraphs except the 4-clique.
The latter\footnote{This also works for the other four-vertex patterns, but with worse amortized update times.} can be directly maintained in $\bigO(m)$ worst-case update time without any auxiliary counts by iterating over all edges and checking adjacencies~\cite{DBLP:conf/sand/HanauerHH22}.

\subsubsection{Corrections}
\label{hhh:impl_notes}
For 3-paths containing a fixed edge $e=\{u,v\}$, the original paper~\cite{DBLP:conf/sand/HanauerHH22} incorrectly states the count for the subcase where $e$ is a terminal edge and the first adjacent inner vertex has low degree as $\mathrm{vLV}[v] - \DegreeN{v}{u} + \mathrm{vLV}[u] - \DegreeN{u}{v}$. The correct count is obtained by applying the following lemma
(proof in \autoref{sec:hhh:corrections:proof})
twice, once with the roles of $u$ and $v$ swapped.

\begin{restatable}{lemma}{hhhthreepath}
\label{lemma:hhhthreepath}
    Assume that the auxiliary structures $\AuxvLV$ and $\AuxuLv$ are maintained with anchors of any degree. Let $\{v, u\}$ be an edge.
    Then, 
    the number of 3-paths of the form $v\text{-}u\text{-}x\text{-}y$ , where 
    $x\in\Low$ is $\AuxvLV[u]- A_v - \AuxuLv[u,v]$, where $A_v= \DegreeN{u}{v}$ if $v \in \Low$ and 0 otherwise.
\end{restatable}

\subsubsection{Optimizations}\label{sect:hhh-opt}
Whereas algorithm design choices are often optimized for simplicity of presentation or proofs and focus on worst-case or amortized analysis, additional criteria become relevant for practical performance. 
In the following, we describe several modifications of \hhh{} and the \epsT{} that focus strongly on the latter.

\popt[prio-low]{O-LP}{Lower-Degree Vertex Prioritization}
Whenever we can choose whether we iterate over the neighbors of one end vertex of an edge or the other, we prioritize the lower-degree vertex.

\popt[arc-border]{O-NP}{Neighborhood Partitioning}
An often-used strategy of \hhh{} is to iterate over all high-degree vertices 
to check adjacency to some vertex $x$, taking $\bigO(|\High|)$ time. We replace this by iterating over the high-degree neighbors of $x$ in $\bigO(\DegreeH{x})$ time, requiring maintenance of each vertex's neighborhood partitioned into high- and low-degree neighbors, which is implementable with constant overhead.

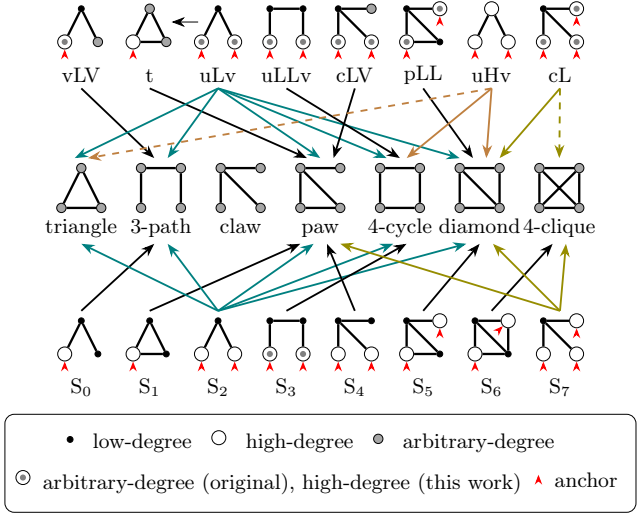
\begin{figure}
    \setlength{\belowcaptionskip}{-10pt}
    \hspace{-2cm}\resizebox{1.23\linewidth}{!}{\def\scaleTop{0.25}
\def\scaleBot{0.3}
\def\markLength{0.2}

\def\xskip{1.03}
\def\xskipBot{1.2}
\def\yskip{2.5}
\def\yskipbot{-2.2}

\tikzset{
  gp/.style={circle,draw,fill=gray!70,minimum size=4pt,inner sep=0pt},
  wp/.style={circle,draw,fill=white,minimum size=6pt,inner sep=0pt},
  wgp/.style={circle,draw,fill=white,minimum size=6pt,inner sep=0pt,
      path picture={
        \fill[gray] (path picture bounding box.center) circle[radius=1.3pt];
      }
  },
  bp/.style={circle,draw,fill=black,minimum size=2.5pt,inner sep=0pt},
  e/.style={line width=1.0pt},
  connectarrow/.style={-{Stealth[length=2.0mm]},line width=0.8pt},
  markarrow/.style={red,-{Stealth[length=1.5mm]},line width=0.7},
  label/.style={font=\small},
}

\ifdefined\markb\else
\newcommand{\markb}[1]{
  \draw[markarrow] ($(#1)+(0,-\markLength)$) -- (#1);
}
\fi

\ifdefined\markt\else
\newcommand{\markt}[1]{
  \draw[markarrow] ($(#1)+(0,\markLength)$) -- (#1);
}
\fi

\ifdefined\markl\else
\newcommand{\markl}[1]{
  \draw[markarrow] ($(#1)+(-\markLength,0)$) -- (#1);
}
\fi

\ifdefined\markr\else
\newcommand{\markr}[1]{
  \draw[markarrow] ($(#1)+(\markLength,0)$) -- (#1);
}
\fi

\ifdefined\markbl\else
\newcommand{\markbl}[1]{
  \draw[markarrow] ($(#1)+(-\markLength/1.414,-\markLength/1.414)$) -- (#1);
}
\fi

\ifdefined\labelbelow\else
\newcommand{\labelbelow}[2][0.5]{
  \node[label] at (0,-#1) {#2};
}
\fi

\ifdefined\nodeTop\else
\newcommand{\nodeTop}[3]{
  \node[#1] (#2) at ($\scaleTop*#3$) {};
}
\fi

\ifdefined\nodeBot\else
\newcommand{\nodeBot}[3]{
  \node[#1] (#2) at ($\scaleBot*#3$) {};
}
\fi

\ifdefined\conntectTop\else
\newcommand{\conntectTop}[4]{
\draw[connectarrow, #4] (#1*\xskip,\yskip-0.7) -- ($(#2*\xskipBot,2.5*\scaleBot)+(#3,0)$) {};
}
\fi

\ifdefined\conntectBot\else
\newcommand{\conntectBot}[4]{
\draw[connectarrow, #4] (#1*\xskip,\yskipbot+0.65) -- ($(#2*\xskipBot,-1.8*\scaleBot)+(#3,0)$) {};
}
\fi

\ifdefined\drawTriangle\else
\newcommand{\drawTriangle}[4]{
  #1{#2}{p1}{(-1,0)};
  #1{#3}{p2}{(0,2)};
  #1{#4}{p3}{(1,0)};
}
\fi

\ifdefined\drawSquare\else
\newcommand{\drawSquare}[5]{
  #1{#2}{p1}{(-1,0)};
  #1{#3}{p2}{(-1,2)};
  #1{#4}{p3}{(1,2)};
  #1{#5}{p4}{(1,0)};
}
\fi

\begin{tikzpicture}[scale=1]

\begin{scope}[shift={(0*\xskip,\yskip)}]
  \drawTriangle{\nodeTop}{wgp}{bp}{gp}
  \draw[e] (p1)--(p2)--(p3);
  \markb{p1};
  \labelbelow{\AuxvLV};
\end{scope}

\begin{scope}[shift={(1*\xskip,\yskip)}]
  \drawTriangle{\nodeTop}{wp}{gp}{gp}
  \draw[e] (p1)--(p2)--(p3)--(p1);;
  \markb{p1};
  \labelbelow{\Auxt};
\end{scope}

\begin{scope}[shift={(2*\xskip,\yskip)}]
  \drawTriangle{\nodeTop}{wgp}{bp}{wgp}
  \draw[e] (p1)--(p2)--(p3);
  \markb{p1};
  \markb{p3};
  \labelbelow{\AuxuLv};
  \draw[-{Stealth[length=2.0mm]}] (-0.3,0.3) --  (-0.7,0.3);
\end{scope}

\begin{scope}[shift={(3*\xskip,\yskip)}]
  \drawSquare{\nodeTop}{wgp}{bp}{bp}{wgp}
  \draw[e] (p1)--(p2)--(p3)--(p4);
  \markb{p1};
  \markb{p4};
  \labelbelow{\AuxuLLv};
\end{scope}

\begin{scope}[shift={(4*\xskip,\yskip)}]
  \drawSquare{\nodeTop}{wgp}{bp}{gp}{wgp}
  \draw[e] (p1)--(p2)--(p3) (p4)--(p2);
  \markb{p1};
  \markb{p4};
  \labelbelow{\AuxcLV};
\end{scope}

\begin{scope}[shift={(5*\xskip,\yskip)}]
  \drawSquare{\nodeTop}{wgp}{bp}{wgp}{bp}
  \draw[e] (p2)--(p4)--(p1)--(p2)--(p3);
  \markb{p1};
  \markb{p3};
  \labelbelow{\AuxpLL};
\end{scope}

\begin{scope}[shift={(6*\xskip,\yskip)}]
  \drawTriangle{\nodeTop}{wp}{wp}{wp}
  \draw[e] (p1)--(p2)--(p3);
  \markb{p1};
  \markb{p3};
  \labelbelow{\AuxuHv};
\end{scope}

\begin{scope}[shift={(7*\xskip,\yskip)}]
  \drawSquare{\nodeTop}{wgp}{bp}{wgp}{wgp}
  \draw[e] (p1)--(p2)--(p3) (p2)--(p4);
  \markb{p1};
  \markb{p3};
  \markb{p4};
  \labelbelow{\AuxcL};
\end{scope}

\begin{scope}[shift={(0*\xskipBot,0)}]
  \drawTriangle{\nodeBot}{gp}{gp}{gp}
  \draw[e] (p1)--(p2)--(p3)--(p1);
  \labelbelow[0.3]{triangle\vphantom{lp}};
\end{scope}

\begin{scope}[shift={(1*\xskipBot,0)}]
  \drawSquare{\nodeBot}{gp}{gp}{gp}{gp}
  \draw[e] (p1)--(p2)--(p3)--(p4);
  \labelbelow[0.3]{3-path\vphantom{lp}};
\end{scope}

\begin{scope}[shift={(2*\xskipBot,0)}]
  \drawSquare{\nodeBot}{gp}{gp}{gp}{gp}
  \draw[e] (p1)--(p2)--(p3) (p2)--(p4);
  \labelbelow[0.3]{claw\vphantom{lp}};
\end{scope}

\begin{scope}[shift={(3*\xskipBot,0)}]
  \drawSquare{\nodeBot}{gp}{gp}{gp}{gp}
  \draw[e] (p2)--(p4)--(p1)--(p2)--(p3);
  \labelbelow[0.3]{paw\vphantom{lp}};
\end{scope}

\begin{scope}[shift={(4*\xskipBot,0)}]
  \drawSquare{\nodeBot}{gp}{gp}{gp}{gp}
  \draw[e] (p1)--(p2)--(p3)--(p4)--(p1);
  \labelbelow[0.3]{4-cycle\vphantom{lp}};
\end{scope}

\begin{scope}[shift={(5*\xskipBot,0)}]
  \drawSquare{\nodeBot}{gp}{gp}{gp}{gp}
  \draw[e] (p2)--(p3)--(p4)--(p1)--(p2)--(p4);
  \labelbelow[0.3]{diamond\vphantom{lp}};
\end{scope}

\begin{scope}[shift={(6*\xskipBot,0)}]
  \drawSquare{\nodeBot}{gp}{gp}{gp}{gp}
  \draw[e] (p2)--(p3)--(p4)--(p1)--(p2)--(p4) (p1)--(p3);
  \labelbelow[0.3]{4-clique\vphantom{lp}};
\end{scope}

\begin{scope}[shift={(0*\xskip,\yskipbot)}]
  \drawTriangle{\nodeTop}{wp}{bp}{bp}
  \draw[e] (p1)--(p2)--(p3);
  \markb{p1};
  \labelbelow{\Sc{0}};
\end{scope}

\begin{scope}[shift={(1*\xskip,\yskipbot)}]
  \drawTriangle{\nodeTop}{wp}{bp}{bp}
  \draw[e] (p1)--(p2)--(p3)--(p1);;
  \markb{p1};
  \labelbelow{\Sc{1}};
\end{scope}

\begin{scope}[shift={(2*\xskip,\yskipbot)}]
  \drawTriangle{\nodeTop}{wp}{bp}{wp}
  \draw[e] (p1)--(p2)--(p3);
  \markb{p1};
  \markb{p3};
  \labelbelow{\Sc{2}};
\end{scope}

\begin{scope}[shift={(3*\xskip,\yskipbot)}]
  \drawSquare{\nodeTop}{wgp}{bp}{bp}{wgp}
  \draw[e] (p1)--(p2)--(p3)--(p4);
  \markb{p1};
  \markb{p4};
  \labelbelow{\Sc{3}};
\end{scope}

\begin{scope}[shift={(4*\xskip,\yskipbot)}]
  \drawSquare{\nodeTop}{wp}{bp}{bp}{wp}
  \draw[e] (p1)--(p2)--(p3) (p4)--(p2);
  \markb{p1};
  \markb{p4};
  \labelbelow{\Sc{4}};
\end{scope}

\begin{scope}[shift={(5*\xskip,\yskipbot)}]
  \drawSquare{\nodeTop}{wp}{bp}{wp}{bp}
  \draw[e] (p1)--(p2)--(p3) (p1)--(p4)--(p2);
  \markb{p1};
  \markb{p3};
  \labelbelow{\Sc{5}};
\end{scope}

\begin{scope}[shift={(6*\xskip,\yskipbot)}]
  \drawSquare{\nodeTop}{wp}{bp}{wp}{bp}
  \draw[e] (p1)--(p2)--(p3)--(p4)--(p1) (p4)--(p2);
  \markb{p1};
  \markbl{p3};
  \labelbelow{\Sc{6}};
\end{scope}

\begin{scope}[shift={(7*\xskip,\yskipbot)}]
  \drawSquare{\nodeTop}{wp}{bp}{wp}{wp}
  \draw[e] (p1)--(p2)--(p3) (p4)--(p2);
  \markb{p1};
  \markb{p4};
  \markb{p3};
  \labelbelow{\Sc{7}};
\end{scope}

\begin{scope}[overlay]
\conntectTop{0}{1}{-0.1}{black}

\conntectTop{1}{3}{-0.2}{black}

\conntectTop{2}{0}{-0.1}{teal}
\conntectTop{2}{1}{0.1}{teal}
\conntectTop{2}{3}{0.0}{teal}
\conntectTop{2}{4}{-0.2}{teal}
\conntectTop{2}{5}{-0.3}{teal}

\conntectTop{3}{4}{0}{black}

\conntectTop{4}{3}{0.2}{black}

\conntectTop{5}{5}{-0.1}{black}

\conntectTop{6}{4}{0.1}{brown}
\conntectTop{6}{5}{0.1}{brown}
\conntectTop{6}{0}{0.1}{brown,dashed}

\conntectTop{7}{5}{0.3}{olive}
\conntectTop{7}{6}{0}{olive,dashed}

\conntectBot{0}{1}{-0.1}{black}

\conntectBot{1}{3}{-0.3}{black}

\conntectBot{2}{0}{0}{teal}
\conntectBot{2}{1}{0.1}{teal}
\conntectBot{2}{3}{-0.1}{teal}
\conntectBot{2}{4}{-0.1}{teal}
\conntectBot{2}{5}{-0.2}{teal}

\conntectBot{3}{4}{0.1}{black}

\conntectBot{4}{3}{0.1}{black}

\conntectBot{5}{5}{0}{black}

\conntectBot{6}{6}{-0.1}{black}

\conntectBot{7}{3}{0.3}{olive}
\conntectBot{7}{5}{0.2}{olive}
\conntectBot{7}{6}{0.1}{olive}
\end{scope}

\path (current bounding box.south) coordinate (bbS);
\matrix (mA) [ampersand replacement=\&,
        anchor=north,
        column sep=5pt,
        nodes={font=\small},
        row sep=0pt] at ($(bbS)+(0,-2mm)$) {
  \node[bp, outer sep=5.0pt] {}; \& \node{low-degree}; \&
  \node[wp, outer sep=2.5pt] {}; \& \node{high-degree}; \&
  \node[gp, outer sep=4.0pt] {}; \& \node{arbitrary-degree}; \& \\
};
\matrix (mB) [ampersand replacement=\&,
        anchor=north,
        column sep=2pt,
        nodes={font=\small},
        row sep=0pt] at ($(bbS)+(0,-8mm)$) {
  \node[wgp, outer sep=2.5pt] {}; \& \node{arbitrary-degree (original), high-degree (this work)}; \&
  \draw[markarrow](0,-0.25) -- (0,-0.15); \& \node{anchor}; \\
};

\begin{scope}[]
  \node[draw=black, rounded corners, inner sep=1pt, fit=(mA) (mB)] {};
\end{scope}

\end{tikzpicture}}
    \caption{
    Dependencies between the patterns (center) and auxiliary structures in \hhh{} (top) and \egst{} (bottom).
    For each pattern, arcs show the structure(s) used to maintain it, where solid arcs follow the original descriptions of \hhh{}~\cite{DBLP:conf/sand/HanauerHH22} and \egst{}~\cite{DBLP:journals/tcs/EppsteinGST12} and dashed ones represent dependencies added in this work.
    The arcs are colored to improve readability.
    }
    \label{fig:hhh_structures_new}
\end{figure}

\popt[high-anchors-hhh]{O-H$_{\hhh}$}{High-Degree Anchors Only}
Allowing low-degree anchors can produce large auxiliary maps, increasing storage and maintenance costs. We therefore restrict all anchors to high-degree vertices (\autoref{fig:hhh_structures_new}). We can accomplish this without worsening asymptotic update bounds.
See \autoref{sec:hhh:high_anchors} for details.

\popt[hhh-extra-aux]{O-A}{Extra Use of Auxiliary Structures}
We can use extra auxiliary structures (dashed arcs in \autoref{fig:hhh_structures_new}) for triangle and 4-clique counting without changing asymptotic running time~(see \autoref{sec:hhh:extra_aux:details}).

\popt[eps-tuning]{O-$\varepsilon\RebalanceFactor$}{Tuning $\varepsilon$ and the MR Frequency~\RebalanceFactor}

Recall that $\varepsilon$ controls the size of $\High$, with $|\High| \in \bigO(m^{1-\varepsilon})$. Larger $\varepsilon$ reduces the number of high-degree vertices, yielding fewer anchors for auxiliary structures but more low-degree vertices as non-anchors, and shifts the trade-off between traversing $\High$ and processing neighbors of low-degree vertices. While a theoretically optimal $\varepsilon$ exists for each pattern, tuning it in practice yields significant runtime improvements.

An MR is triggered when $m$ leaves the range $\lfloor\frac{m_i}{2}\rfloor \le m < 2m_i$. This is generalized with a rebalance factor $\RebalanceFactor>1$: an MR occurs only when $m$ exits $\lfloor\frac{m_i}{\RebalanceFactor}\rfloor \le m < \RebalanceFactor m_i$, setting $m_{i+1}=m$ and $\theta_{i+1}=(\RebalanceFactor m_{i+1})^\varepsilon$. Larger $\RebalanceFactor$ yields longer epochs and fewer MRs without affecting asymptotic running times for constant $\RebalanceFactor$.

\popt[eps-sll]{O-SLL}{Soft, Lazy, or Late Rebalancing}
Even though we can reduce the frequency of MR events (cf.~\poptref{eps-tuning}), each of them still 
triggers expensive update steps when it occurs,
as all auxiliary structures are recomputed from scratch by simulating edge insertions (cf.\ section~\ref{sec:hhh}).
We here suggest three alternatives and refer to the original one as \origepstab{}.

Let $T_\text{aux}$ be the time needed to update all auxiliary structures for a single (simulated) edge insertion or deletion.
In \origepstab{}, an MR results in $m$ edge updates to the auxiliary structures and hence incurs a total cost of $\bigO(m\cdot T_\text{aux})$ and
an amortized cost of
$\bigO(m \cdot T_\text{aux} \cdot m^{-1})=\bigO(T_\text{aux})$ for each update since the last MR.
The time for a single VPU, \ie, for performing the partition change and updating all affected auxiliary structures, is
$T_\text{VPU} \in \bigO(\deglow \cdot T_\text{aux})$~\cite{DBLP:conf/sand/HanauerHH22}.
The proof is in \autoref{sec:optimizing_epstab:proofs:reclassified_bound}.

\begin{restatable}{lemma}{n_wrong_vertices}
\label{lemma:n_wrong_vertices}
The number of vertices that are reclassified by an MR is in
$\bigO(m^{1-\varepsilon})$.
\end{restatable}

\paragraph{\softepstab{}}
A straightforward approach is to run the VPU procedure for each vertex that is incorrectly partitioned after the MR rather than deleting all auxiliary structure counts and recomputing them from scratch. 
In practice, this performs
well if few vertices are affected.
We refer to this strategy as \softepstab.
See \autoref{sec:optimizing_epstab:proofs:soft} for the proof.

\begin{restatable}{lemma}{soft}
\label{lemma:soft}
MR via individual VPUs (\softepstab{} rebalancing) has the same total cost of $\bigO(m\cdot T_\text{aux})$ as the original version.
\end{restatable}

\paragraph{\lazyepstab{}.}
To entirely avoid a large and time-consuming operation in a single step, we consider a \lazyepstab{} variation to distribute the partition updates of an MR over multiple future edge updates, aiming for
more stable update times.
Here, we define $\theta_{i+1}$ and a new range for $m$ as usual, but don't perform any partition or auxiliary structure changes except, if necessary, for the end vertices of the edge that triggered the MR.
Every other vertex $v$ that is now wrongly partitioned w.r.t.\ $\theta_{i+1}$, \ie, $\Degree{v} < \frac{1}{2}\theta_{i+1}$
if $v \in \High$ or $\Degree{v} \geq \frac{3}{2}\theta_{i+1}$ if $v \in \Low$, is only reassigned via a VPU when one of its incident edges is part of an update.
As a consequence, if the graph grows, vertices without any incident edge updates may become stale and get stuck in $\High$, 
and symmetrically for vertices in $\Low$ if the graph shrinks.
This may violate the bound on $\deglow$ 
(proof in \autoref{sec:optimizing_epstab:proofs:lazy}) 
and hence waive the amortized update time guarantee under an adversarial update sequence.

\begin{restatable}{lemma}{lazy}
\label{lemma:lazy}
 Rebalancing with the \lazyepstab{} strategy
 may result in
 $\deglow \in \Theta(m)$.
\end{restatable}

\paragraph{\lateepstab{}.}

To spread rebalancing costs without sacrificing guarantees,
\lateepstab{} combines \softepstab{}'s selective updates with
gradual processing. Unlike \lazyepstab{}, it completes
reclassification before the next MR, preserving the degree bounds.
Recall that an epoch has length $\Omega(m)$, in particular it covers at least
$(\RebalanceFactor-1)m_i/\RebalanceFactor$ updates.
Checking both endpoints of
$\lceil\RebalanceFactor/(\RebalanceFactor-1)\rceil$ edges per update
therefore processes all $m_i$ initial edges within that interval.
This distribution of work also has a positive effect on the worst-case time for maintaining the \epsT{}~(proof in \autoref{sec:optimizing_epstab:proofs:late}):
\begin{restatable}{lemma}{late}
\label{lemma:late}
MR via individual VPUs over
the first $(\RebalanceFactor -1)\RebalanceFactor^{-1}m_i$ updates of the new epoch
(\lateepstab{}{} rebalancing)
(1) has the same total cost of $\bigO(m\cdot T_\text{aux})$ for the MR as the original version
and (2) a constant number of VPUs per update and only on vertices of degree in $\bigO(\deglow) = \bigO(m^\varepsilon)$.
\end{restatable}

\paragraph{Worst-Case Update Times.}
We finally analyze the \emph{worst-case} update times for each of the four variants.
See \autoref{sec:worst-case-times:proofs} for the derivations.

\begin{restatable}{lemma}{epst-worstcase}
\label{lemma:epst-worstcase}
The \emph{worst-case} update times for \hhh{} using different variants of the \epsT{}
are as shown in \autoref{tab:worstcase-times}.
In particular, \lateepstab{} has the same amortized update times as \origepstab{}, but improves the worst-case update times by factors in $[\sqrt{m}, m]$.
\end{restatable}

\begin{table}[t!]
\centering\small
\caption{Worst-case update times (in $\bigO$) to count subgraphs with \hhh{} using different \epsT{}s and the theoretically optimal $\varepsilon^*$ (for the amortized time).
Recall that direct scanning has $\bigO(m)$ worst-case update time.
}
\label{tab:worstcase-times}
\resizebox{\linewidth}{!}{\setlength{\tabcolsep}{2pt}
\begin{tabular}{@{}rc|ccc|cc@{}}
\toprule
\multirow{2}{4em}{Subgraph}
&
\multirow{2}{.5em}{$\varepsilon$}
& \multicolumn{3}{c}{worst-case}
& \multicolumn{2}{c}{amortized} \\
 &  & \origepstab{}/\softepstab{} & \lazyepstab{} &\lateepstab{} & \origepstab{}/\softepstab{}/\lateepstab{} & \lazyepstab{}\\
\midrule
 triangle \tri
 & $\nicefrac{1}{2}$
 & $m^{3/2}$ & $m$ & $m^{1/2}$ 
 & $m^{1/2}$ & $m$ \\
 3-paths \tPath
 & $\nicefrac12$
 & $m^{3/2}$  & $m^2$ & $m$ 
 & $m^{1/2}$ & $m$\\
 paw \paw
 & $\nicefrac13$
  & $m^{5/3}$   & $m^2$ & $m^{2/3}$
  & $m^{2/3}$ & $m^2$\\
 4-cycle \fCycle
 & $\nicefrac13$
 & $m^{5/3}$  & $m^3$ & $m$
 & $m^{2/3}$ & $m^2$\\
diamond \dia
 & $\nicefrac13$
 & $m^{5/3}$  & $m^3$  & $m$
 & $m^{2/3}$ & $m^2$\\
  \bottomrule
\end{tabular}}
\end{table}

\subsection{\egst{} Algorithm and \hInd}\label{sect:egst}
We next introduce the \hInd{} as well as the \egst{}~\cite{DBLP:journals/tcs/EppsteinGST12} algorithm
along with some clarifications that do not appear in the original work.
The section ends again with practical optimization strategies.

\subsubsection{\hInd{} Partition}  
The \hInd{} vertex partition is based on the Hirsch index (\hInd{}) concept from bibliometrics~\cite{hirsch_index_2005} and transferred to graphs.
The \hInd{} of a graph is the largest integer $h$ such that there are $h$ vertices with degree $\ge h$.  On an empty graph, $\Low=V$. A vertex $v$ swaps from $\Low\rightarrow\High$ once it reaches $\Degree{v}\ge2h$ and from $\High\rightarrow\Low$ if $\Degree{v} < h$.
There are no major rebalances.
With the \hInd{} partition, $|\High|\in\bigO(h)$, and $\deglow\in \bigO(h)$, with $\changesperupdate\in \bigO(1/h)$ amortized changes to the partition per update~\cite{DBLP:journals/jgaa/EppsteinS12}.

\subsubsection{\egst{} algorithm}

The \egst{} algorithm uses the vertex partition ($\High$, $\Low$) given by the \hInd. 
Similar to \hhh, \egst~also maintains per-vertex auxiliary counts that record in how many auxiliary structures a vertex participates as an anchor.
Many structures tracked by \egst~(see \autoref{fig:hhh_structures_new}) are similar to those of \hhh, with one key difference being that \egst~requires all anchor vertices to be high-degree, except for \Sc{3}.

As in \hhh{}, after every update all auxiliary and pattern counts are adjusted, however with other dependencies (see \autoref{fig:hhh_structures_new}).
In contrast to the \epsT{}, where $|\High|\in\bigO(m^{1-\varepsilon})$ and $\deglow\in\bigO(m^\varepsilon)$,
both values are bounded by $\bigO(h)$ here.
This results in an amortized update time of $\bigO(h^2)$~\cite{DBLP:journals/tcs/EppsteinGST12}.

\subsubsection{Corrections and Completions}
Certain details regarding count updates were either omitted or stated ambiguously in \cite{DBLP:journals/tcs/EppsteinGST12}, which we address in the following.
Let $e=\{u,v\}$ be an updated edge.

(1) Update rules for \Sc{2} \DrawSTwo{}/\Sc{3} \DrawSThreeAny{}:
No explicit update rules are given when $u\in\Low$, $v\in\High$ (or swapped), we add all missing cases. \Sc{3} was defined with arbitrary anchors for $u,v\in\Low$ updates, and we applied this to the missing cases as well. \Sc{2} is used with mixed-degree anchors in \Sc{5}/\Sc{6} updates but exclusively with high-degree ones in subgraph counting. Since high-degree anchors suffice for counting, we only maintain those, reducing memory.

(2) Update rules for \Sc{5} \DrawSFive{}/\Sc{6} \DrawSSix{}: 
The update rules are ambiguous when $u\in\Low$, $v\in\High$ (or swapped). The paper states that \Sc{2} \DrawSTwo{} is used in \Sc{5}/\Sc{6} edge updates. However, for \Sc{6} this is not possible, as the low-degree vertices must be adjacent, which cannot be tested using the \Sc{2} count. For \Sc{5}, this would require maintaining \Sc{2} with mixed-degree anchors which we decided not to do (see above). To preserve asymptotic bounds, we propose: When $u,v\in\Low$, increment \Sc{5}$[h_1,h_2]$ for each $h_1\in\NeighH{u}$, $h_2\in\NeighH{v}$, $h_1\ne h_2$, where $\{h_1,v\}\in E$ or $\{h_2,u\}\in E$. If $v\in\High$, increment \Sc{5}$[v,h]$ for each $h\in\NeighNH{v}{u}$, $\ell\in\NeighL{u}$, where $\{\ell,h\}\in E$ or $\{\ell,v\}\in E$. \Sc{6} updates are identical but require both edges present.

(3) Update rules for 3-paths \tPath{} :
\cite{DBLP:journals/tcs/EppsteinGST12} omits several cases for 3-path updates depending on the position of $e$ in the path. When $e$ is the \emph{central} edge, the count is adjusted by $(\Degree{u}-1)\cdot(\Degree{v}-1)$, which overcounts by including triangles containing $e$: a vertex $x$ adjacent to both $u$ and $v$ does not form a valid 3-path. We correct this by finding all such $x$ using \Sc{2} for $x\in\Low$ and adjacency tests for $x\in\High$ when $u,v\in\High$, or iterating over the low-degree endpoint's neighbors otherwise. When $e$ is \emph{non-central}, say the 3-path is $v\text{-}u\text{-}x\text{-}y$ (symmetrical for swapped roles of $u$ and $v$), 
\cite{DBLP:journals/tcs/EppsteinGST12} covers $u\in\Low$ with $x\in\High\cup\Low$, and $u,x\in\High$, but omits $u\in\High$ with $x\in\Low$.
We fill this gap in Lemma~\ref{lemma:egstthreepath}
(proof in \autoref{sec:egst:3path_proof}),
which counts such paths using \Sc{0} and \Sc{2} with triangle corrections.

\subsubsection{Optimizations}\label{egst:engineering}

This subsection presents variants designed with practical improvements to enhance the efficiency of \egst{} and the \hInd{}.
We use \poptref{prio-low} and \poptref{arc-border} as described in \autoref{sect:hhh-opt}.

\popt[high-anchors-egst]{O-H$_{\egst}$}{High-Degree Anchors Only}
For the same reasons as in \poptref{high-anchors-hhh}, having high-degree anchors only
is desirable also for \egst{}.
Here, this only requires a modification of $\Sc{3}$ (all other auxiliary structures have high anchors already):
For edge insertions/deletions with $e=\{u,v\}$, if $u,v\in\Low$ we update $\Sc{3}[h_1,h_2]$ by one for all distinct pairs $\{h_1,h_2\}$ with $h_1\in \NeighH{u}$ and $h_2\in \NeighH{v}$ in $\bigO(h^2)$ time. If, w.l.o.g., $v\in\High$  and  $u\in\Low$, we update $\Sc{3}[y,v]$ by one for all $y \in \NeighH{x}$ with $x\in \NeighL{u}$ in $\bigO(h^2)$ time.
In consequence, we update 4-cycle counts by iterating over low-degree neighbors instead of using \Sc{3}.
The new procedure takes time $\bigO(h^2)$ and does not affect the asymptotic running time.

\popt[direct-swap]{O-D}{Direct Partition Swap}
When a vertex changes partition, the auxiliary counts need to be updated accordingly. 
This can be done by examining all edges adjacent to the updated vertex one by one and following the procedure for edge updates, described in~\cite{DBLP:journals/tcs/EppsteinGST12} and \poptref{high-anchors-egst}, twice:
Once to subtract counts of now-invalid structures and once to add newly formed structures.
However, the partition swap can also be performed directly by 
scanning the 2-hop neighborhood in the same running time bounds~(see \autoref{sec:egst:direct:details}).

\popt[gradual]{O-GF}{Tuning the Gradual Factor of the \hInd{} Partition}
In \cite{DBLP:journals/jgaa/EppsteinS12}, a vertex $v$ is placed in $\High$ only when $\deg(v)\geq 2h$, which in practice often yields very few high-degree vertices. We introduce a gradual factor $c$, placing $v$ in $\High$ when $\deg(v)\geq c\cdot h$.
We prove that the running time bounds remain the same in \autoref{sec:h_index:c:proof}.

\subsection{Separating Algorithms and Partitioning}

Both \hhh{} and \egst{} operate on any vertex partition. We therefore decouple the algorithms from their specific partitions, derive update times for an abstract partition \Partition{}, and analyze cross-combinations.
\begin{table}[t]
\renewcommand{\arraystretch}{1.1}
\addtolength{\tabcolsep}{-2pt}
	\centering
    \footnotesize
	\caption{
     Amortized update times (in $\bigO$) for subgraph counts using \hhh{} or
     \egst{}
     combined with
     an abstract partition $\Partition$,
     the \epsT{}, or the \hInd{}.
     Shaded terms only apply to the high-anchors-only optimization.
    }\label{tab:generalP}
\resizebox{\linewidth}{!}{
	\begin{tabular}{@{}lcr@{$\ge$}lc@{}}
    \toprule
    Name          & \Partition & \multicolumn{2}{c}{\epsT} & \hInd  \\
    \midrule
     \multicolumn{4}{l}{\hhh{} \bgcolortext{gray}{with \poptref{high-anchors-hhh}}} \\
     \midrule
     triangle & $|\High| +$ \bgcolortext{gray}{$\deglow$} $+ \deglow^2\cdot\changesperupdate$ &$m^{\max(1-\varepsilon,\varepsilon)}$  & $\sqrt{m}$ & $h$   \\
     3-path & $|\High|+$ \bgcolortext{gray}{$\deglow$} $+ \deglow^2\cdot\changesperupdate$ &$m^{\max(1-\varepsilon,\varepsilon)}$ & $\sqrt{m}$  & $h$   \\
    paw &$|\High| + \deglow^2 +\deglow^2\cdot\changesperupdate$ &$m^{\max(1-\varepsilon,2\varepsilon)}$ & $m^{2/3}$   & $h^2$  \\
    4-cycle & $|\High| +\deglow^2 + (\deglow^3 + \deglow |\High |)\changesperupdate$  &$m^{\max(1-\varepsilon,2\varepsilon)}$ & $m^{2/3}$  & $h^2$  \\
    diamond& $|\High| + \deglow^2 + 
    (\deglow^3 + \deglow |\High |)\changesperupdate $ &$m^{\max(1-\varepsilon,2\varepsilon)}$ & $m^{2/3}$  & $h^2$   \\
    4-clique & $m$  & \multicolumn{2}{c}{$m$} & $m$   \\
    \midrule
     \multicolumn{4}{l}{\egst{} \bgcolortext{gray}{with \poptref{high-anchors-egst}}} \\
    \midrule
    triangle & $|\High| + \deglow + \deglow^2\cdot\changesperupdate$&$m^{\max(1-\varepsilon,\varepsilon)}$ & $\sqrt{m}$ & $h$ \\
    3-path & $|\High| + \deglow  + \deglow^2\cdot\changesperupdate$ &$m^{\max(1-\varepsilon,\varepsilon)}$ & $\sqrt{m}$ & $h$ \\
    paw & $|\High|^2 + \deglow^2 + \deglow^3\cdot\changesperupdate $ &$m^{\max(2-2\varepsilon,2\varepsilon)}$ & $m$ & $h^2$ \\
    4-cycle & $|\High| \cdot \deglow +|\High|^2 +$ \bgcolortext{gray}{$\deglow^2$} $+  \deglow^3\cdot\changesperupdate $ &$m^{\max(1-2\varepsilon,2-2\varepsilon,2\varepsilon)}$& $m$  & $h^2$\\
    diamond & $|\High|^2 + \deglow^2  +\deglow^3\cdot\changesperupdate$ &$m^{\max(2-2\varepsilon,2\varepsilon)}$&  $m$  & $h^2$\\
    4-clique & $|\High|^2 + \deglow^2 + \deglow^3\cdot\changesperupdate $&$m^{\max(2-2\varepsilon,2\varepsilon)}$& $m$ & $h^2$\\
    \bottomrule
	\end{tabular}
    }
\end{table}

Table~\ref{tab:generalP} states the amortized update times for both algorithms under \Partition{}.
For derivations, see \autoref{sec:abstract_p:egst:deviation} for \egst{}, \cite{hhh_full_version} for \hhh{}, and \autoref{sec:hhh:high_anchors} for \poptref{high-anchors-hhh}. 
The table also shows the cross-combination bounds, which (in all cases) preserve or worsen the amortized update time, especially for \egst{} with \epsT{}. Notably, \hhh{} with \hInd{} matches \egst{} with \hInd{} for all patterns except the 4-clique: \hhh{} uses a naive $\bigO(m)$ routine, while \egst{} achieves $\bigO(h^2)$ via \Sc{6} and \Sc{7}, exploiting $|\High|, \deglow \in \bigO(h)$ under \hInd{}. With \epsT{}, where $\bigO(|\High|) \ne \bigO(\deglow)$, \egst{}'s 4-clique time becomes $\bigO(m^{\max(2-2\varepsilon, 2\varepsilon)}) = \bigO(m)$ at $\varepsilon = 0.5$, matching the trivial bound without partitioning.

Our analysis also reveals that \hhh{} and \egst{} have identical complexity for triangles and 3-paths. They differ for paw, 4-cycle, and diamond: \egst{}'s update time scales quadratically with $|\High|$ for these patterns, as its auxiliary structures only involve low-degree vertices besides the anchors.
Patterns where both non-anchor vertices are in $\High$ require enumerating all pairs in $\High$.

\section{Experimental Evaluation}
We experimentally evaluate the algorithms, their proposed optimizations, and cross-combinations against several baselines on a diverse set of graphs.

\subsection{Methodology}
We implemented our algorithms~\cite{code:dynSC}
in C++20 (\texttt{g++ -O3}). For anchor-to-count mappings we used \texttt{boost::unordered\_flat\_map}~\cite{code:boost-unordered}, which outperformed alternatives in preliminary tests.
We always apply \poptref{prio-low} and \poptref{arc-border}, as they consistently improve performance and denote these base versions \hhho{} and \egsto{}. Further variants layer additional optimizations: \hhhoh{} adds high-degree anchors \poptref{high-anchors-hhh}, \hhhoe{} adds auxiliary counting structures \poptref{hhh-extra-aux}, and \hhhohe{} combines both. Analogously, \hegst{} uses high-degree anchors \poptref{high-anchors-egst}, \degst{} uses direct partition swaps \poptref{direct-swap}, and 
\hdegst{} combines both.
We evaluate \epsT{} variants \origepstab, \softepstab{}, \lazyepstab{}, and \lateepstab{},
where the $\varepsilon$ value is indicated as a suffix, e.g., \softepstab{.5} for $\varepsilon=0.5$.
The \hInd{} partition with gradual factor $c$ is written \hindexgr{c}. We denote algorithm \texttt{A} with partition \texttt{P} as \texttt{A|P}, e.g., \texttt{\egsto|\hindexgr{2}} or \texttt{\hhhoe{}|\softepstab{.5}}.

We compare against all available state-of-the-art algorithms capable of counting both triangles and four-vertex subgraphs.
The only fully-dynamic subgraph counter we are aware of is \baselineDNA{}~\cite{DBLP:conf/alcob/SchillerJHS15}, which we extended to provide a runner and result extraction~\cite{code:dna}.
We evaluate against the continuous subgraph matching (CSM) algorithms \baselineRapidFlow{}~\cite{DBLP:journals/pvldb/SunSHL22,code:csm}, \baselineGraphflow{}~\cite{DBLP:conf/sigmod/KankanamgeSMCS17}, \baselineTurboFlux{}~\cite{kim2018turboflux}, \baselineIEDyn{}~\cite{DBLP:conf/sigmod/IdrisUV17,DBLP:journals/vldb/IdrisUVVL20}, and \baselineSymBi{}~\cite{DBLP:journals/pvldb/MinPPGIH21}. The latter four are implemented and benchmarked by Sun et al.~\cite{DBLP:journals/pvldb/SunSLH22,code:csm}. Following their recommendations, we use \baselineGraphflow{} and \baselineSymBi{} for all patterns, and additionally include \baselineTurboFlux{} and \baselineIEDyn{} for tree-structured (3-path) queries.
We also compare against recomputing counts from scratch after each update using two static algorithms: \texttt{OB}~\cite{DBLP:journals/ans/OrtmannB17,code:oaqc} and \texttt{ESCAPE}~\cite{DBLP:conf/www/PinarSV17,code:escape}. For the latter we applied a memory leak patch~\cite{code:dynSC}.

Experiments ran on Ubuntu 24.04.1 LTS with a dual-socket AMD EPYC~9534 (128 cores, \SI{3.72}{\giga\hertz} max) and \SI{1.48}{\tera\byte} RAM, each experiment pinned to a single core without overlapping caches. We report medians over three runs and use geometric means across instances. A \emph{speedup} of $q\times$ for $A$ over $B$ means $q = T(B)/T(A)$ where $T(\cdot)$ can refer to \emph{total runtime} (sum of update times), \emph{maximum update time}, and \emph{memory} usage. Unless stated otherwise, all seven patterns are maintained simultaneously.
We validated our implementations by verifying the counts manually on small instances and by cross-checking them with the counts reported by the static baselines \texttt{OB} and \texttt{ESCAPE} after each update. 

\begin{table}[t!]
\caption{Full Data Set (\FullSet{}) with \#updates ($\delta$), and max.\ \#edges ($\widehat{m}$). Shaded cells indicate membership to the engineering dataset \EngSet{}. 
Full statistics in \autoref{tab:graphs}.
}
\label{tab:graphs_overview}
\centering
\scriptsize
\setlength{\tabcolsep}{2pt}
\renewcommand{\arraystretch}{0.85}
\newcommand{\gba}{\cellcolor{white!70!gray}}
\resizebox{\linewidth}{!}{
\begin{tabular}{cl rr|cl rr}
\toprule
\multicolumn{4}{c|}{\DSWiki{}~\cite{konect,konect_wiki_second_ref}} & \multicolumn{4}{c}{\DSTMP{}} \\
ID & Graph & $\delta$ & $\widehat{m}$ & ID & Graph & $\delta$ & $\widehat{m}$ \\
\midrule
\gba A &\gba simple\_wiki &\gba 1.6M &\gba 690.7K & A & reality-call~\cite{DBLP:conf/aaai/RossiA15,eagle2006reality} & 7.2K & 3.6K \\
\gba B &\gba nl\_wiki &\gba 18.9M &\gba 9.6M & B & wiki-elec~\cite{DBLP:conf/aaai/RossiA15,leskovec2010signed} & 201.5K & 100.8K \\
\gba C &\gba pl\_wiki &\gba 23.6M &\gba 11.4M & C & tech-as-topology~\cite{DBLP:conf/aaai/RossiA15} & 215.4K & 107.7K \\
\gba D &\gba it\_wiki &\gba 33.1M &\gba 15.7M & D & reddit-body~\cite{snapnets,kumar2018community} & 433.5K & 4K \\
\gba E &\gba fr\_wiki &\gba 56.3M &\gba 22.3M & E & mathoverflow~\cite{snapnets,ParanjapeBL17,stackexchange_data_dump} & 458K & 2.5K \\
\gba F &\gba de\_wiki &\gba 82M &\gba 28.1M &\gba F &\gba enron~\cite{DBLP:conf/aaai/RossiA15,cohen2005enron} &\gba 594.9K &\gba 297.5K \\
  &  &  &  & G & reddit-title~\cite{snapnets,kumar2018community} & 824.7K & 7.8K \\
  &  &  &  &\gba H &\gba askubuntu~\cite{snapnets,ParanjapeBL17,stackexchange_data_dump} &\gba 952.4K &\gba 6.1K \\
  &  &  &  &\gba I &\gba superuser~\cite{snapnets,ParanjapeBL17,stackexchange_data_dump} &\gba 1.5M &\gba 7.3K \\
  &  &  &  &\gba J &\gba wiki-talk~\cite{snapnets,ParanjapeBL17,DBLP:conf/icwsm/LeskovecHK10} &\gba 6.3M &\gba 102K \\
  &  &  &  &\gba K &\gba youtube-growth~\cite{DBLP:conf/aaai/RossiA15,mislove-2007-socialnetworks} &\gba 18.8M &\gba 9.4M \\
  &  &  &  &\gba L &\gba stackoverflow~\cite{snapnets,ParanjapeBL17,stackexchange_data_dump} &\gba 59.4M &\gba 277.9K \\
\midrule
\multicolumn{4}{c|}{\DScKron{}~\cite{DBLP:conf/alenex/HanauerH020,dyreach}} & \multicolumn{4}{c}{\DSgKron{}~\cite{DBLP:conf/alenex/HanauerH020,dyreach}} \\
ID & Graph & $\delta$ & $\widehat{m}$ & ID & Graph & $\delta$ & $\widehat{m}$ \\
\midrule
\gba A &\gba c\_answers &\gba 2.6M &\gba 185.5K & A & g\_answers & 538.3K & 184.9K \\
\gba B &\gba c\_delicious &\gba 3.2M &\gba 225.9K & B & g\_delicious & 646K & 225.3K \\
 C & c\_epinions & 5.9M & 413.7K &\gba C &\gba g\_epinions &\gba 1.1M &\gba 413K \\
 D & c\_flickr & 4.3M & 305.5K &\gba D &\gba g\_flickr &\gba 853.4K &\gba 304.7K \\
 E & c\_blog-nat05-6m & 18.7M & 1.3M & E & g\_blog-nat05-6m & 3.3M & 1.3M \\
 F & c\_blog-nat06all & 22.3M & 1.6M &\gba F &\gba g\_blog-nat06all &\gba 3.9M &\gba 1.6M \\
 G & c\_cit-hep-ph & 22.2M & 1.6M &\gba G &\gba g\_cit-hep-ph &\gba 3.9M &\gba 1.6M \\
 H & c\_cit-hep-th & 29.8M & 1.6M & H & g\_cit-hep-th & 3.9M & 1.6M \\
\gba I &\gba c\_ca-dblp &\gba 8.5M &\gba 597.3K & I & g\_ca-dblp & 1.6M & 596K \\
 J & c\_ca-gr-qc & 7.9M & 559.2K & J & g\_ca-gr-qc & 1.5M & 558.4K \\
 K & c\_ca-hep-th & 5.3M & 374.1K & K & g\_ca-hep-ph & 5.3M & 2.1M \\
 L & c\_ca-hep-ph & 30.4M & 2.1M & L & g\_ca-hep-th & 1M & 373.2K \\
\gba M &\gba c\_web-notredame &\gba 4.1M &\gba 288.7K & M & g\_web-notredame & 811.6K & 288.1K \\
 N & c\_as-newman & 6.9M & 487.6K &\gba N &\gba g\_as-newman &\gba 1.3M &\gba 486K \\
\gba O &\gba c\_as-routeviews &\gba 7.9M &\gba 556.9K & O & g\_as-routeviews & 1.5M & 555.3K \\
 P & c\_gnutella-30 & 2.9M & 203.4K & P & g\_gnutella-25 & 722.1K & 253.2K \\
 Q & c\_gnutella-25 & 3.6M & 254.1K & Q & g\_gnutella-30 & 587.7K & 202.7K \\
\gba R &\gba c\_atp-gr-qc &\gba 1.3M &\gba 90.9K & R & g\_atp-gr-qc & 277.3K & 90.5K \\
 S & c\_bio-proteins & 9.5M & 667.5K &\gba S &\gba g\_bio-proteins &\gba 1.8M &\gba 666.3K \\
\bottomrule
\end{tabular}
}
\vspace{-0.5cm}
\end{table}
We use real-world and synthetic datasets from various domains (Table~\ref{tab:graphs_overview}), partitioned into four groups:
\DSWiki{}~\cite{konect}: Wikipedia hyperlink networks with real insertions and deletions.
\DSTMP{}~\cite{DBLP:journals/pvldb/ChenLHB22,DBLP:conf/icde/MinJPGIH24,DBLP:journals/vldb/AghasadeghiBS24,DBLP:journals/kbs/LiZLYW22}: timestamped real-world graphs~\cite{snapnets,DBLP:conf/aaai/RossiA15} where edges are inserted at their timestamp and removed two weeks later~\cite{DBLP:journals/pvldb/ChenLHB22}.
\DScKron{} and \DSgKron{}~\cite{DBLP:conf/alenex/HanauerH020}: synthetic dynamic Kronecker graphs~\cite{LeskovecCKFG10} derived from snapshot differences, with constant or growing vertex counts respectively.
All instances were preprocessed by removing multi-edges, self-loops, directions, and deletions of non-existing edges.
Beyond the full dataset (\FullSet{}), we use a balanced engineering subset (\EngSet{}) of the first two million updates from six graphs per group (highlighted in \autoref{tab:graphs_overview}; details in \autoref{tab:engeneering_set}).
We will finally use \FullSet{} to validate our parameter choices made on \EngSet{}.

\subsection{\hhh{} and \epsT{} Optimizations}

We first evaluate \hhho{} with optimizations on the original \epsT{}, followed by a parameter study for $\varepsilon$ and $\RebalanceFactor$.
Finally, we consider the best variant of \hhho{} in combination with different partitions, \ie{} \epsT{} variants \softepstab{}, \lazyepstab{} and \lateepstab{} as well as the \hInd{}.

\subsubsection{Proposed Optimizations}
We start by evaluating \poptref{high-anchors-hhh}, where all anchors have high degree, on the \EngSet{} graphs.
For this experiment, we set $\varepsilon = \num{.417}$ as a compromise between the optimum of $\frac{1}{2}$ for triangles and 3-paths and $\frac{1}{3}$ for all other patterns (except for 4-clique/claw, which don't use the partition), \ie,
we compare \texttt{\hhhoh{}|\origepstab{.417}} to \texttt{\hhho{}|\origepstab{.417}}.
As \hhhoh{} restricts what vertices can serve as anchors to those with high degree, fewer counts need to be stored at the expense of more iterations over neighborhood sets.
Despite this, \hhhoh{} achieves a $21.9\times$ speedup in total time, a 
$116.1\times$
speedup in the maximum update time, and reduces memory usage by a factor of
$49.0\times$
over \hhho{}.
These results showcase the impressive efficiency of \poptref{high-anchors-hhh} and demonstrate that the typical memory-runtime tradeoff, \ie, storing more to save computation time, can reach the point where finding, storing and accessing the auxiliary counts outweigh the benefits of reduced traversals in the subgraph counting step.

Another proposed optimization, \poptref{hhh-extra-aux}, uses
additional auxiliary structures where possible; however, preliminary experiments
showed this is only slightly beneficial if all patterns are counted.

\emph{We use \hhhohe{} in experiments where all patterns are counted and \hhhoh{} otherwise.}

\subsubsection{Analysis of $\varepsilon$ for the \epsT{}}
\label{exp:eps:study}
\begin{figure}[t]
    \setlength{\belowcaptionskip}{-15pt}
    \setlength{\abovecaptionskip}{-0pt}
    \includegraphics[width=\linewidth]{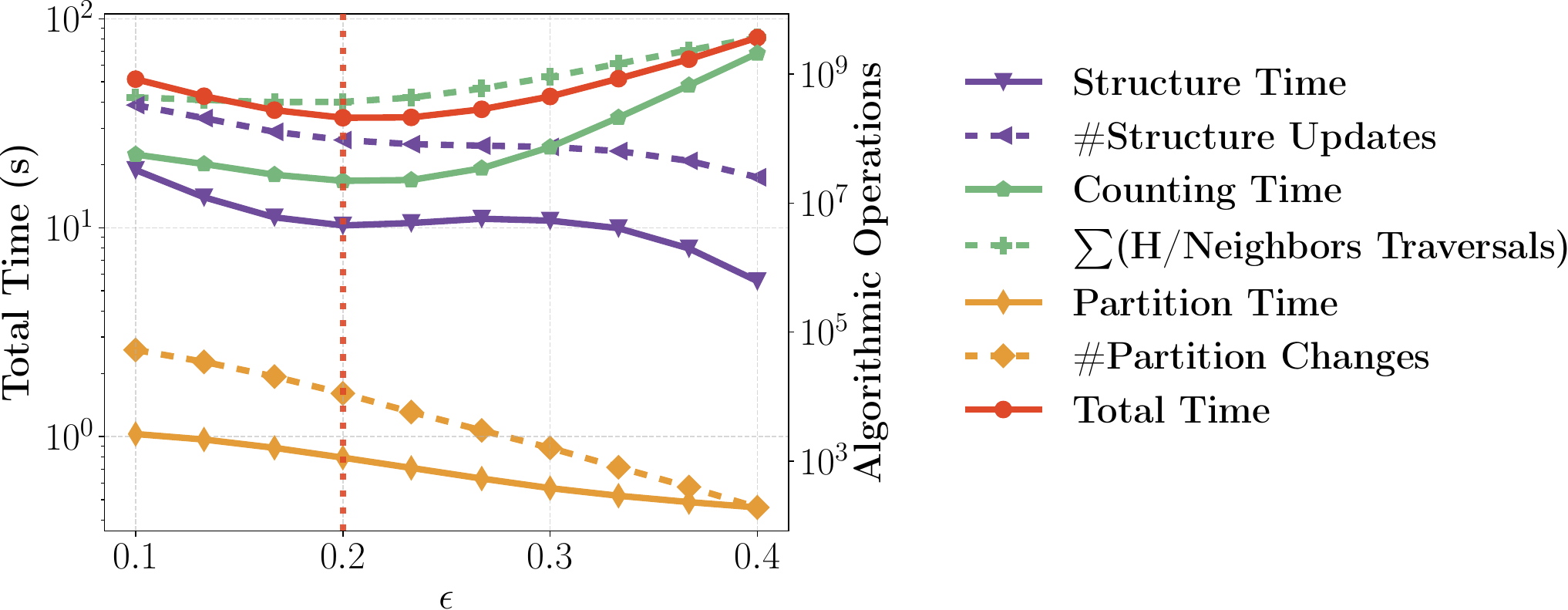}
    \caption{Effects of the $\varepsilon$ parameter on the runtime of \texttt{\hhhohe{}|\origepstab{}}\label{exp:lazy_eps}
    using dataset \EngSet{}
    (counting all patterns).
     Note the logarithmic y-axis.
    }
\end{figure}
In Figure~\ref{exp:lazy_eps} we decompose runtime into partition maintenance, auxiliary-structure maintenance, and pattern counting, which directly correlate with the number of partition changes, structure updates and traversals of the neighborhood or the $H$ set, respectively.
We use \origepstab{} on \EngSet{}; the effects are identical for other variants.
Increasing $\varepsilon$ reduces the high-degree set and generally lowers partition and structure time, but increases counting time, which dominates overall runtime. We observe that counting all patterns simultaneously is fastest at $\varepsilon=0.2$, below the theoretical optima of 1/3 or 1/2. Two effects contribute to this difference. First, neighborhood partitioning replaces full high-degree set scans with high-degree-neighbor scans, reducing the penalty for a larger $\High$. Second, simultaneous counting shares auxiliary structures across patterns, reducing their relative maintenance overhead. Thus, the optimal $\varepsilon$ for counting all patterns simultaneously cannot be derived by simply averaging the individual optima per pattern.
\emph{In practice, $\varepsilon$ should be chosen smaller than the theoretical optimum.}

\subsubsection{Best Partitions for \hhh{}}
\label{sec:exp:hhh:eps_vari}
\begin{figure}[t]

\setlength{\belowcaptionskip}{-15pt}
\setlength{\abovecaptionskip}{-0pt}

\includegraphics[width=\linewidth]{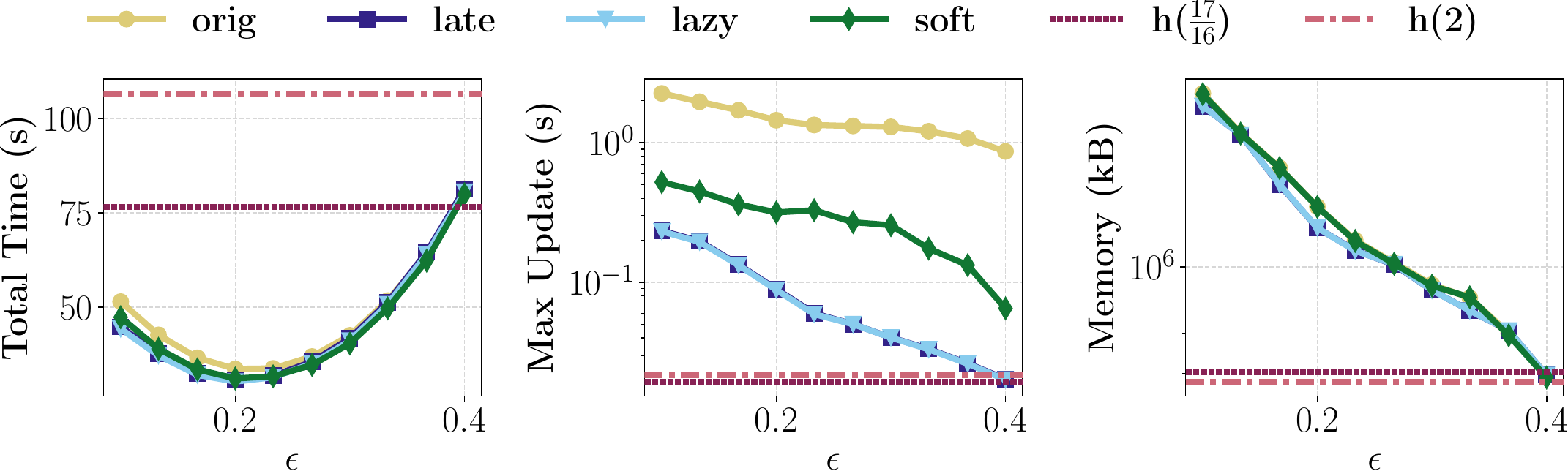}
\caption{
Performance
of \hhhohe{}
for different vertex partitions
w.r.t. $\varepsilon$ 
on \EngSet{} data
(counting all patterns).
} 
\label{exp:eps_variants}
\end{figure}
We compare all partitions studied in this work using \hhhohe{} on \EngSet{}. For the \hInd{}, we vary the gradual factor $c\in\{1+\frac{1}{2^x}\mid x=0,\dots,10\}$; for the \epsT{}, we compare \origepstab{}, \softepstab{}, \lateepstab{}, and \lazyepstab{} with $\varepsilon\in\{0.1+\frac{i}{30}\mid i=0,\dots,9\}$ rounded to three digits. Figure~\ref{exp:eps_variants} reports overall running time, maximum update time, and memory usage. To increase readability, we show only \hindexgr{$2$} (the original choice) and \hindexgr{$\frac{17}{16}$} (the best in our experiments).

\emph{Total runtime}. All \epsT{} variants perform best at $\varepsilon=0.2$.
Among variants preserving amortized guarantees,
\lateepstab{.2} is fastest, $1.10\times$ faster than \origepstab{.2} and $1.01\times$ than \softepstab{.2}. \lazyepstab{} is another \SI{1.2}{\percent} faster but lacks these guarantees. For \hInd{} partitions, \hindexgrBest{} achieves a $1.39\times$ speedup over \hindexgr{$2$}, though \lateepstab{.2} remains
$2.5\times$
faster.

\paragraph{Maximum update time.}
Maximum update time generally decreases with $\varepsilon$,
as the high-degree set shrinks and partition changes become fewer.
By replacing reconstruction with individual VPUs, \softepstab{}
improves on \origepstab{} by factors of $4.3$--$13.3$.
\lateepstab{} further limits VPUs per update, yielding improvements
of $9.6$--$42.6$, including $16.2\times$ at $\varepsilon=0.2$.
\lazyepstab{}, with at most two VPUs per update, is only
\SI{1.3}{\percent} faster than \lateepstab{}.
With a smaller high-degree set and fewer partition changes,
\hindexgrBest{} achieves $4.6\times$ faster maximum updates
than \lateepstab{.2}, at the expense of total runtime.

\paragraph{Memory.}
Fewer high-degree vertices mean fewer potential anchors and
generally smaller auxiliary maps.
\origepstab{} and \softepstab{} use similar memory because
their partitions are identical after each update.
At $\varepsilon=0.2$, \lateepstab{} and \lazyepstab{} use
\SI{93.1}{\percent} of \origepstab{}'s memory on average,
though differences vary across instances.
With an even smaller high-degree set, \hindexgrBest{} uses
only \SI{61.7}{\percent} of \lateepstab{.2}'s memory.

We studied the \emph{rebalance factor} $\RebalanceFactor$ (see \autoref{sec:exp:hhh:recomp_factor}) and found the optimal choice to be highly graph-dependent. We fix $\RebalanceFactor=2$ as a good compromise.

Finally, we also briefly evaluate \texttt{\hhhoh{}|\lateepstab{}} for single patterns instead of all simultaneously~(cf.~ \autoref{sec:single_patterns}).
Also here, $\varepsilon$ should be chosen smaller than its theoretically optimal value:
For diamonds, paws and 4-cycles, $\varepsilon=0.233$ yields the best results in practice with speedups of
$1.2\times$, $1.6\times$, $1.4\times$,
respectively, over the theoretically best $\varepsilon=\frac{1}{3}$.
The deviation in running time from $\varepsilon = 0.2$, the best choice for ``all'', is within \SI{10}{\percent} in the geometric mean. 
For counting all patterns simultaneously, a slightly smaller $\varepsilon$ is faster in practice since auxiliary structures can be used for multiple patterns without increasing structure time (see Section~\ref{exp:eps:study}).
For 3-paths, $\varepsilon=0.333$ is
$1.5\times$
faster than the theoretical optimum of $\varepsilon=\frac{1}{2}$, and almost twice as fast as $\varepsilon = 0.2$.
Clever neighborhood scanning without any auxiliary structures works best for triangles and 4-cliques and obliterates the need for a partition, yielding speedups of
$1.9\times$ and $1.5\times$
over \lateepstab{.2} respectively.

\emph{\lateepstab{.2} and \hindexgrBest{} are the fastest \epsT{} and \hInd{} variants, respectively, offering different tradeoffs: \lateepstab{} is
$2.5\times$
faster in total runtime, while \hindexgrBest{} has
$4.6\times$
faster maximum updates and uses $38.4\%$ less memory. \lazyepstab{.2} (without preserving theoretical bounds) is ${\approx}1\%$ faster than \lateepstab{.2}. $\RebalanceFactor=2$ is a good default for all \epsT{}s.}

\subsection{Optimizing \egst{} Performance}
We compare the four combinations of proposed \egst{} optimizations: high-degree anchors (\hegst{}), direct partition swaps (\degst{}), and their combination (\hdegst{}) on \EngSet{} using \hindexgr{2}. \hdegst{} achieves a
$4.3\times$
faster total runtime and
$132.9\times$
faster maximum update time over base \egst{}, with high-degree anchors contributing most of the gain; full results are in~\autoref{sec:egst:opt_versions:exp}.

\subsubsection{Different Partitions for \egst{}}
\begin{figure}[t]

\setlength{\belowcaptionskip}{-17pt}
\setlength{\abovecaptionskip}{-0pt}

\includegraphics[width=\linewidth]{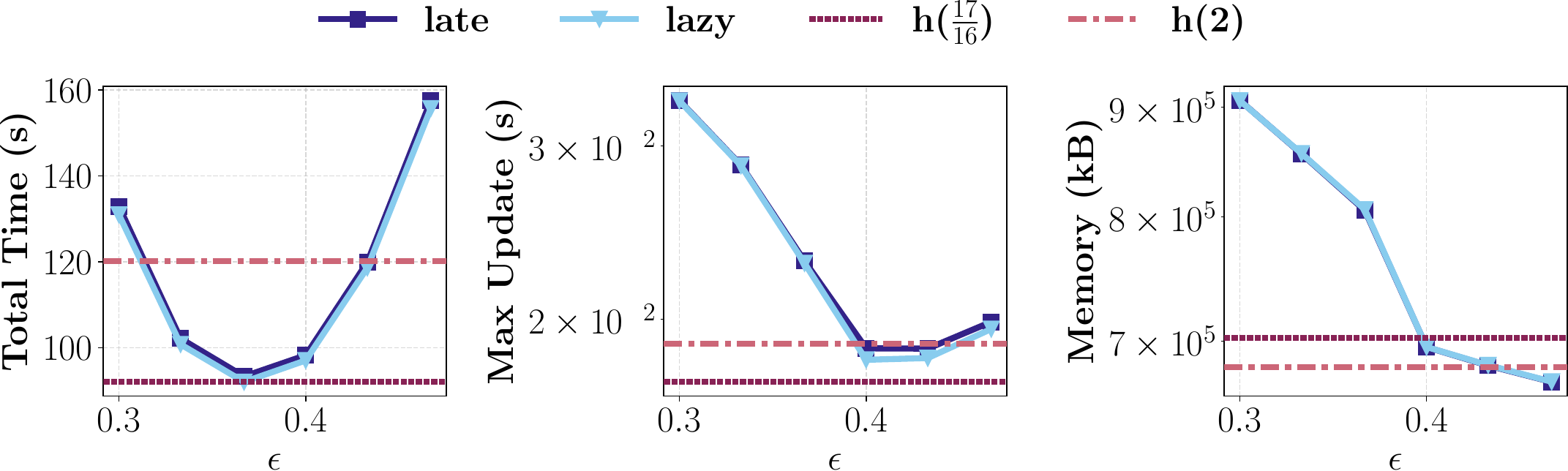}
\caption{
Performance of \hdegst{} 
for different vertex partitions
w.r.t. $\varepsilon$ 
on \EngSet{} data
(counting all patterns).
}
\label{exp:egst_eps}
\end{figure}
Using \hdegst{} as best-performing \egst{} variant, we evaluate vertex partitions on \EngSet{}. We vary $c\in\{1+\frac{1}{2^x}\mid x=0,\dots,10\}$ for \hInd{} and $\varepsilon\in\{0.3+\frac{i}{30}\mid i=0,\dots,5\}$ (rounded to three digits) for \lazyepstab{} and \lateepstab{}, the best \epsT{} variants from \autoref{sec:exp:hhh:eps_vari}.
Figure~\ref{exp:egst_eps} reports total runtime, maximum update time, and memory. As in \autoref{sec:exp:hhh:eps_vari}, we only include \hindexgr{$2$} and \hindexgrBest{}, where the latter achieves
$1.3\times$
faster total runtime,
$1.1\times$
faster updates, and only $4\%$ more memory.

The effects of \epsT{} in \autoref{exp:eps:study} also favor $\varepsilon$ below the theoretical
optimum of $1/2$ for \egst{}.
\lateepstab{.367}, \lazyepstab{.367}, and \hindexgrBest{}
have similar total runtime, but \hindexgrBest{}
achieves $1.32\times$ faster maximum updates and uses
\SI{14.8}{\percent} less memory.
These averages mask instance-level differences:
$\varepsilon=0.333$ wins on 17 of 24 graphs by up to
$1.3\times$, whereas \hindexgrBest{} wins on dense,
high-count instances such as \DSWiki{} by up to $3.3\times$,
consistent with \egst{}'s quadratic dependence on $|\High|$.

For single patterns, there is no clear winner~(cf.~\autoref{sec:single_patterns}).
We focus on \lateepstab{}/\hInd{} to preserve theoretical guarantees. While \hindexgrBest{} performs best for diamonds, \lateepstab{} is faster otherwise:
$1.1\times$
for 4-cycles and paws ($\varepsilon=.333$), and
$1.03\times$
for 3-paths ($\varepsilon=.433$). As with \hhh{}, triangles and 4-cliques are fastest without a partition (equivalently, \lazyepstab{1.0}).

\emph{
When counting all patterns with \hdegst{}, \hindexgrBest{} is the best partition with similar total runtime as \lateepstab{.367}, but with
$1.3\times$
faster maximum update time and $14.8\%$ less memory.
}

\subsection{Comparison to Static Baselines}
\begin{figure}[t]

\setlength{\abovecaptionskip}{-0pt}

\includegraphics[width=\linewidth]{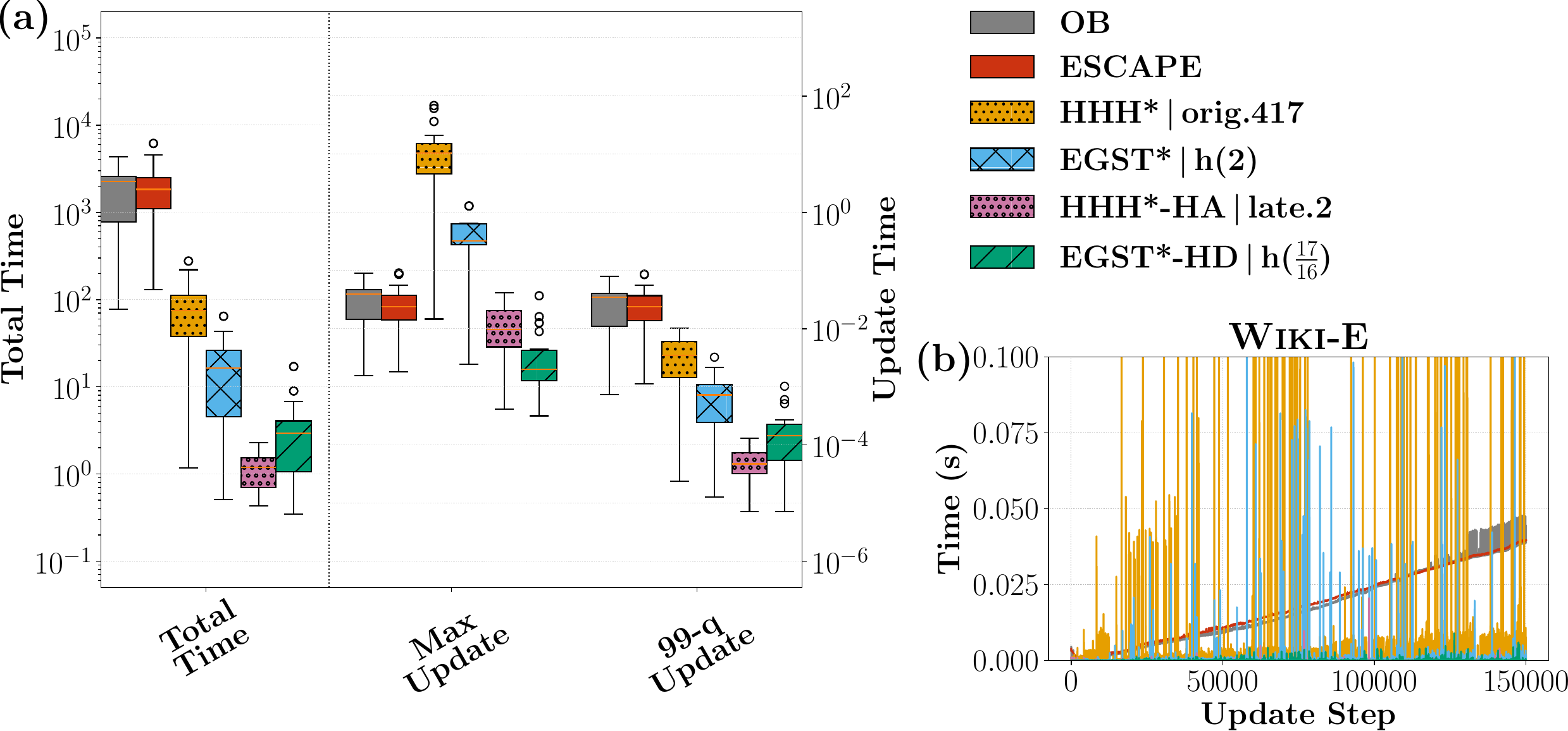}

\caption{
(a) Running times (total, maximum, and 99th-pctl) over \EngSet{} graphs limited to 150K updates and (b) update times on \DSWiki{}-E. Both counting all patterns.
} 
\label{exp:comp_static}
\end{figure}

We compare our dynamic algorithms against static recomputation via \texttt{OB} and \texttt{ESCAPE} after every update when counting all patterns simultaneously. We include both basic (\texttt{\hhho{}|\origepstab{.417}}, \texttt{\egsto{}|\hindexgr{2}}) and fully optimized versions (\texttt{\hhhohe{}|\lateepstab{}{.2}}, \texttt{\hdegst{}{}|\hindexgrBest{}}), evaluated on the first \num{150000} updates of each \EngSet{} instance to accommodate static runtimes.
In consequence, \DScKron{} instances are insertions-only here, while
\DSgKron{}, \DSWiki{}, and \DSTMP{} instances retain most characteristics.
Results (Figure~\ref{exp:comp_static}) show: 
Dynamic counting improves total runtime by 2--4 orders of
magnitude, but basic versions occasionally take longer than
static recomputation for individual updates, despite being
1--3 orders faster on 99\% of updates.
Figure~\ref{exp:comp_static} links these peaks to expensive
auxiliary-map maintenance and major rebalancing.
The introduced optimizations restrict anchors and avoid full reconstruction, improving total runtime and maximum update time.

\emph{Overall, \texttt{\hhhohe{}|\lateepstab{}{.2}} and \texttt{\hdegst{}{}|\hindexgrBest{}} clearly outperform static algorithms and non-optimized counterparts in both total and maximum update time.}

\subsection{Comparison to Dynamic Baselines}
\begin{figure}[t]

\setlength{\belowcaptionskip}{-17pt}
\setlength{\abovecaptionskip}{5pt}

\includegraphics[width=\linewidth]{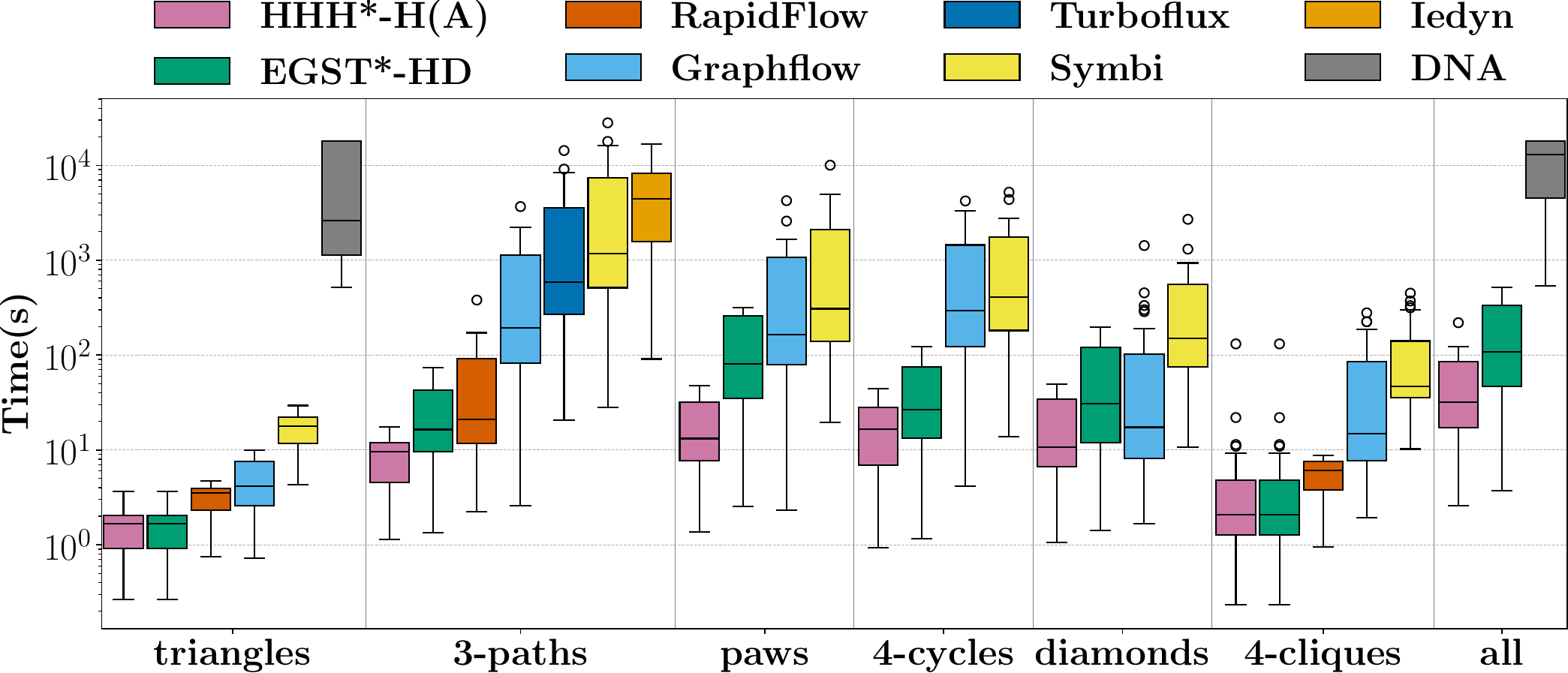}
\caption{
Running times of dynamic subgraph counting and CSM algorithms
using dataset \EngSet{}.
}
\label{exp:comp_to_dyn_baselines}
\end{figure}
\begin{figure*}[t]
  \includegraphics[width=\textwidth]{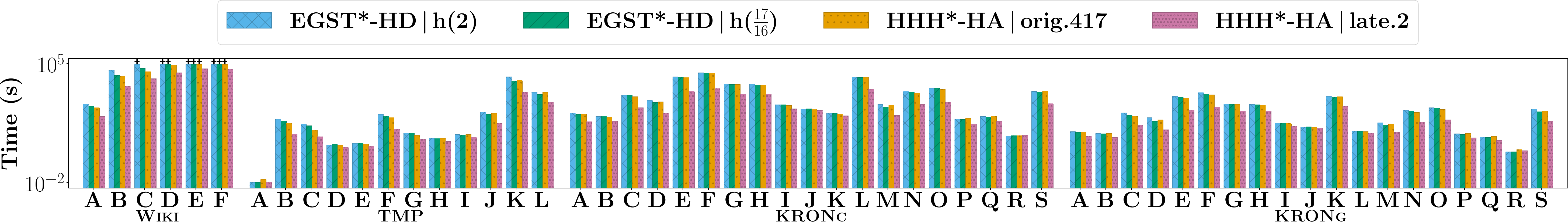}
  \caption{
  Running times of dynamic subgraph counting algorithms
using dataset \FullSet{} (counting all patterns).
  }
  \label{exp:full_data}
\end{figure*}

We compare our algorithms against \baselineDNA{} (dynamic subgraph counting) and state-of-the-art CSM algorithms~\cite{DBLP:journals/pvldb/SunSLH22,DBLP:journals/pvldb/SunSHL22} on \EngSet{}, using the best-performing configuration of our algorithms for each pattern. Specifically, \texttt{\hhhohe{}|\lateepstab{.2}} is used for all patterns simultaneously, \hhhoh{} with \lateepstab{.233} (diamonds, paws, 4-cycles), \lateepstab{.333} (3-paths), or \lateepstab{1.0} (triangles, 4-cliques); and \hdegst{} with \hindexgrBest{} (all, diamonds), \lateepstab{.333} (3-paths, 4-cycles), \lateepstab{.433} (paws), or \lateepstab{1.0} (triangles, 4-cliques). Note that the \lateepstab{1.0} variants are identical scanning algorithms without auxiliary structures. CSM algorithms count single patterns; \baselineDNA{} counts triangles and all patterns. All baselines have a \SI{5}{\hour} timeout.
Results are shown in Figure~\ref{exp:comp_to_dyn_baselines}. 
\baselineDNA{} timed out on half of the instances.
\baselineSymBi{} timed out for 3-paths on \DSWiki{}-D.
We use \SI{5}{\hour} as a lower bound for speedup computation.
\baselineRapidFlow{} crashed (segmentation faults) on up to \SI{25}{\percent} of the instances for paws, 4-cycles, and diamonds.
We therefore excluded it for these patterns.

Compared to \baselineDNA{}, \hhhohe{} is
$\geq 264.6\times$
faster for all patterns and \hhhoh{} is $\geq 2325\times$ faster on triangles.
For triangles, 3-paths, and 4-cliques, \baselineRapidFlow{} performed best among CSM baselines;
on paws, 4-cycles and diamonds, \baselineGraphflow{} did.
\hhhoh{} outperforms it across all patterns, with speedups between 
$1.8\times$
(4-cliques) and $26.2\times$ (4-cycles). On individual instances, gains can exceed two orders of magnitude (\eg,
$149.8\times$
on \DSWiki{}-D for 4-cycles). The exception can be graphs with very low pattern counts like \DScKron{}-R~(Table~\ref{tab:engeneering_set}), where \baselineGraphflow{} outperforms \hhhoh{} by up to
$1.7\times$
on 3-paths, paws, and diamonds, where auxiliary structure overhead does not pay off.

\emph{\hhhoh{} achieves speedups of
$1.8\times$–$26.2\times$
over \baselineGraphflow{}/\baselineRapidFlow{} and
$\geq 2325\times$
over \baselineDNA{} (triangles); \hhhohe{} is
$\geq 264.6\times$
faster than \baselineDNA{} (all patterns). Only on graphs with very low pattern frequency can \baselineGraphflow{} outperform \hhhoh{} by up to \SI{44}{\percent}.}

\subsection{Performance Evaluation on \FullSet{}}

We evaluate performance (Figure~\ref{exp:full_data}) by comparing the best combinations (\texttt{\hhhohe{}|\lateepstab{.2}}, \texttt{\hdegst{}|\hindexgrBest{}}) alongside their untuned counterparts (\texttt{\hhhohe{}|\origepstab{.417}}, \texttt{\hdegst{}|\hindexgr{2}}) on \FullSet{} to validate the parameter choices made on \EngSet{}~(Table~\ref{tab:graphs}).
We use a \SI{24}{\hour} timeout. For timed-out runs, we record a time of \SI{24}{\hour}, which under-estimates the reported speedups.
Partition tuning yields a
$2.6\times$
speedup 
($0.9$–$7.7\times$ per graph)
for \texttt{\hhhohe{}|\lateepstab{.2}} over \texttt{\hhhohe{}|\origepstab{.417}}, and
$1.1\times$ ($0.9$–$2.0\times$)
for \texttt{\hdegst{}|\hindexgrBest{}} over \texttt{\hdegst{}|\hindexgr{2}}. Comparing the best versions, \texttt{\hdegst{}|\hindexgrBest{}} timed out on \DSWiki{}-D/E/F after completing \SI{97.5}{\percent}/\SI{61.0}{\percent}/\SI{39.2}{\percent} of updates, while \texttt{\hhhohe{}|\lateepstab{.2}} finished in \SI{7.9}{\hour}/\SI{13.7}{\hour}/\SI{13.2}{\hour}. On all remaining instances, \texttt{\hhhohe{}|\lateepstab{.2}} achieves an overall
$2.7\times$ speedup ($0.9$–$8.3\times$),
improving runtime on every instance except those with very few pattern counts (\DScKron{}-R, \DSgKron{}-R), where scanning algorithms without auxiliary structures remain preferable. This effect is less severe for \texttt{\hdegst{}|\hindexgrBest{}}, which maintains simpler auxiliary structures. The speedup of \texttt{\hhhohe{}|\lateepstab{.2}} partly stems from storing more auxiliary structures, enabled by fewer structural restrictions and a larger $|\High|$, costing \SI{74.9}{\percent} more memory than \texttt{\hdegst{}|\hindexgrBest{}}; increasing $\varepsilon$ with \lateepstab{} reduces memory at the expense of runtime.

\emph{\texttt{\hhhohe{}|\lateepstab{.2}} outperforms \hdegst{} by multiple factors on nearly all instances, but uses more memory.}

\section{Conclusion and Practical Recommendations}

Our \hhhoh{}/\hhhohe{} algorithms are the most effective for dynamic subgraph counting, offering speedups of up to several orders of magnitude. We recommend the following configurations:
\begin{itemize}[itemsep=0ex,partopsep=1ex,parsep=1ex,labelwidth=1em,leftmargin=3em,topsep=1pt]
\item[\textbf{All}] (all patterns simultaneously): \texttt{\hhhohe{}|\lateepstab{.2}}.
\item[\dia{} \paw{} \fCycle{}] (diamond, paw, 4-cycle): \texttt{\hhhoh{}|\lateepstab{.233}}.
\item[\tPath{}] (3-path): \texttt{\hhhoh{}|\lateepstab{.333}}.
\item[\tri{} \fClique{} \claw{}] (triangle, 4-clique, claw): prefer efficient scanning/direct counting as in \hhho{}/\egsto{} without partition.
\end{itemize}
\textbf{Memory constrained:} Increase $\varepsilon$ in \lateepstab{} to trade runtime for memory, if $\varepsilon\ge0.4$ switch to \hindexgrBest{},  or to \texttt{\hdegst{}|\hindexgrBest{}} if memory remains insufficient.\\
\textbf{Practically fastest:} Use \lazyepstab{} instead of \lateepstab{}; it was $\approx \SI{1}{\percent}$ faster than \lateepstab{} in practice, but waives the amortized and worst-case update time guarantees.\\
\textbf{Small counts:} If it is known a priori that the subgraph counts will be small, auxiliary structures may create overhead; the scanning routines used by \hhhohe{}/\egsto{} for cliques or \baselineGraphflow{} can be preferable.

\section*{Acknowledgements}
OpenAI Codex assisted with language editing, reference formatting, and LaTeX/PDF checks. The authors assume responsibility for all content.

\clearpage

\appendix
\renewcommand\sectionautorefname{Appendix}
\section{Appendix}
\subsection{Extended Related Work}\label{sect:relwork_long}
Complementing the summary in \autoref{sect:relwork_short}, we give a more detailed literature overview here 
and refer the reader to detailed surveys on static~\cite{DBLP:journals/csur/RibeiroPSAS21} and dynamic~\cite{DBLP:journals/jea/HanauerHS22} subgraph counting for an even more comprehensive treatment. 
We omit results on subgraph \emph{detection}, which is only concerned with the question whether the pattern occurs at all or not, but not with its frequency.
We mention results on subgraph \emph{matching} (or \emph{enumeration}), as extracting frequencies is straightforward.

\paragraph{Static Setting}
Subgraph counting for \emph{arbitrary patterns} generalizes the $\mathcal{NP}$-hard subgraph isomorphism problem, thus ruling out efficient algorithms unless $\mathcal{P}=\mathcal{NP}$.
Much prior work is concerned with the identification of \emph{motifs}~\cite{milo2002network,itzhack2007optimal,DBLP:conf/wabi/Wernicke05,li2012netmode,DBLP:journals/tcbb/WangHWP12,DBLP:conf/sac/RibeiroS10}, \ie, patterns that occur more frequently than can be expected under a suitable random model.
Kloks et al.~\cite{DBLP:journals/ipl/KloksKM00} show how to count 4-cliques in $\bigO(m^{(\omega+1)/2})$ time and all other induced 4-vertex subgraphs in $\bigO(n^\omega + m^{1.69})$ time\footnote{$\omega$ denotes the exponent of fast matrix multiplication algorithms, $\omega < 2.373$~\cite{DBLP:conf/issac/Gall14a}}.
Ho\v{c}evar and Dem\v{s}ar~\cite{DBLP:journals/bioinformatics/HocevarD14} give an algorithm for counting induced subgraphs with up to five vertices that exploits symmetries in the pattern.
A similar approach was later employed by Ortmann and Brandes~\cite{DBLP:journals/ans/OrtmannB17} for non-induced subgraphs of size up to four.
Pinar et al.~\cite{DBLP:conf/www/PinarSV17} give an algorithm for patterns with up to five nodes.
Especially for triangles, there is also a large body of work on approximate and streaming algorithms~\cite{DBLP:conf/cocoon/JowhariG05,DBLP:conf/pods/BuriolFLMS06,DBLP:conf/kdd/LimK15,DBLP:journals/tkdd/LimJK18,DBLP:conf/pods/McGregorVV16,DBLP:journals/siamsc/KoldaPPST14,DBLP:journals/pvldb/PavanTTW13}.

\paragraph{Dynamic Setting}
Eppstein and Spiro~\cite{DBLP:journals/jgaa/EppsteinS12} describe an algorithm for maintaining the triangle count in $\bigO(h)$ amortized update time, where $h$ is again the h-index of the graph\footnote{Note that $h \in \bigO(\sqrt{m})$.}.
Eppstein et al.~\cite{DBLP:journals/tcs/EppsteinGST12} extend this approach to non-induced subgraphs of size four with an amortized update time of $\bigO(h^2)$ (\egst{}, see \autoref{sect:egst}).
Kara et al.~\cite{DBLP:conf/icdt/KaraNNOZ19} show how to count triangles in $\bigO(\sqrt{m})$ amortized update time based on a heavy-light vertex partitioning.
Based on their work, Lu and Tao~\cite{DBLP:conf/icdt/LuT21} study an optimal tradeoff between update time and approximation quality and obtain arboricity-sensitive update times.
Using a similar vertex partitioning strategy as Kara et al., Hanauer et al.~\cite{DBLP:conf/sand/HanauerHH22} give improved algorithms to count non-induced 4-vertex patterns in amortized time $\bigO(m^\frac{2}{3})$ except for the clique and the path, which take $\bigO(m)$ and $\bigO(\sqrt{m})$ update time, respectively (\hhh{}).
Their triangle counting algorithm is conceptually very similar to that by Kara et al.~\cite{DBLP:conf/icdt/KaraNNOZ19} and has the same guarantees.
Hanauer et al.\ further show that
conditioned on the combinatorial k-clique conjecture~\cite{DBLP:journals/siamcomp/AbboudBW18}, induced connected 4-vertex patterns cannot be maintained with amortized update time $\bigO(m^{1-\gamma})$ and query time $\bigO(m^{2-\gamma})$ for any $\gamma > 0$.
There also is no fully dynamic algorithm that maintains non-induced connected 4-vertex pattern counts in amortized $\bigO(m^{1-\delta})$ update time and $\bigO(m^{2-\delta})$ query time with an error probability of at most $\frac{1}{3}$ for any $\delta > 0$ unless the OMv conjecture~\cite{DBLP:conf/stoc/HenzingerKNS15} fails.
Assadi et al.~\cite{DBLP:journals/pacmmod/AssadiS25} show that 4-cycles can be maintained in $\bigO(m^{\frac{2}{3}-0.0099811})$ using fast matrix multiplication.
Dhulipala et al.~\cite{DBLP:conf/apocs/DhulipalaLSY21} consider a batch-dynamic setting, where each graph update may consist of multiple edge insertions and deletions,
to count cliques and give easily parallelizable algorithms for counting cliques with an amortized update time of $\bigO(B \sqrt{B+m})$, where $B$ is the size of the batch.

Schiller et al.~\cite{DBLP:conf/alcob/SchillerJHS15} describe a practical algorithm based on neighborhood scanning to count induced triangles and 4-vertex subgraphs in $\bigO(\MaxDegree^2)$ update time,
where $\MaxDegree$ denotes the maximum vertex degree.

\paragraph{CSM}
There also exists a large body of empirical work on continuous subgraph matching (CSM) for arbitrary subgraphs~\cite{fan2013incremental,kim2018turboflux,DBLP:journals/pvldb/MinPPGIH21,DBLP:conf/sigmod/KankanamgeSMCS17,DBLP:conf/sigmod/IdrisUV17,DBLP:journals/vldb/IdrisUVVL20,DBLP:journals/pvldb/SunSHL22}.
CSM algorithms often rely on an \textit{index}, a precomputed data structure that accelerates query processing by storing structural information about the data graph.
Edge updates then consist of index updates and matching enumeration. 
\baselineGraphflow{}~\cite{DBLP:conf/sigmod/KankanamgeSMCS17} computes incremental changes without relying on an explicit index.
\baselineTurboFlux{}~\cite{kim2018turboflux} and \baselineIEDyn{}~\cite{DBLP:conf/sigmod/IdrisUV17, DBLP:journals/vldb/IdrisUVVL20} construct indices with space and time complexity $\bigO(|E(G)| \cdot |V(Q)|)$ for data graph $G$ and query graph $Q$, while \baselineSymBi{}~\cite{DBLP:journals/pvldb/MinPPGIH21} builds an index with complexity $\bigO(|E(G)| \cdot |E(Q)|)$.
Sun et al.~\cite{DBLP:journals/pvldb/SunSLH22} provide an extensive survey and experimental study of these algorithms and give recommendations for different use cases. According to their study, there is no single dominant algorithm:
\baselineSymBi{} is particularly effective for sparse queries or long-running queries, \baselineTurboFlux{} and \baselineIEDyn{} are well-suited for tree queries, while \baselineGraphflow{} performs best on dense queries or when both query and data graphs are sparse.

\subsection{Proof of Correctness for \hhh{} 3-Path Counting}

\label{sec:hhh:corrections:proof}
\restatethelemma{hhhthreepath}
\begin{proof}
    Recall from \autoref{fig:hhh_structures_new} that $\AuxvLV[u]$ (\DrawvLV{}) counts the paths $u\text{-}x\text{-}y$ with $x\in\Low$.
    There are two cases where an extension of a path counted by $\AuxvLV[u]$ with edge $\{v,u\}$ does not result in a valid 3-path of the form $v-u-x-y$: (1) $x=v$ and (2) $y=v$.
    The number of invalid paths in (1) is $A_v= \DegreeN{u}{v}$ if $v \in \Low$ and 0 otherwise, and the number of invalid paths for (2) is $\AuxuLv[u,v]$, \ie, the number of 2-paths adjacent to both $u$ and $v$ with a low-degree joint.
\end{proof}

\subsection{The \hhh{} Algorithm with High Anchors}
\label{sec:hhh:high_anchors}

We describe the required techniques and changes to the structure and subgraph count maintenance.
We begin with the auxiliary structures. The structures \Auxt{} and \AuxuHv{} remain unchanged, since their anchors already have high degree in \cite{DBLP:conf/sand/HanauerHH22}.
For the remaining structures the change to high-degree anchors affects both edge updates and partition changes of vertices.

In case of edge updates for most structures, it is sufficient to include an additional requirement to the description in \cite{DBLP:conf/sand/HanauerHH22} for anchor vertices to be in $\High$. Note that this $\bigO(1)$-time check does not change the asymptotic runtime, when using \poptref{arc-border} this even improves the actual running time in  practice. 
The only exception is the \Auxt{} auxiliary structure, where an update of edge $e=\{u,v\}$ requires additional changes, since as before $\AuxuLv[u,v]$ was used in the update routine which is no longer possible when $u\in\Low$ or $v\in\Low$. In this case, assume \wlogtext{} $u\in\Low$ and adjust $\Auxt[h]$ by $|\{\{v,h\}\in E\}|$ for each $h \in \NeighH{u}$ and if $v\in\High$, $\Auxt[v]$ by $\big|\{x \in \Neigh{u} \mid \{v,x\}\in E\}\big|$.
This adjustment is an increment for insertions and a decrement for deletions and can be performed in time $\bigO(m^{\max(1-\varepsilon,\varepsilon)})$, which equals the previous update time bound of \Auxt{} in \cite{DBLP:conf/sand/HanauerHH22}.

We next describe how we update the count of each auxiliary structure when a vertex $v$ changes the partition. We use $v\rightarrow\Low$ ($v\rightarrow\High$) as a notation for $v$ changing to $\Low$ ($\High$). In case $v$ has a non-anchor position we extend \cite{DBLP:conf/sand/HanauerHH22} by the high anchors requirement. For $v$ in an anchor position we use the following routines:

(\AuxvLV{} \DrawvLV{}) If $v\rightarrow\Low$, we set $\AuxvLV[v] = 0$. If $v\rightarrow\High$, we compute newly formed $\AuxvLV[v] = \sum_{x \in \NeighL{v}} \DegreeN{v}{x}$, in $\bigO(m^\varepsilon)$.

(\AuxuLv{} \DrawuLv{}) If $v\rightarrow\Low$, we set $\AuxuLv[v, x] = 0$ for all $x \in \NeighNH{v}{\ell}$, where $\ell \in \NeighL{v}$. 
In practice, instead of finding $x$, we can also find all entries of \AuxuLv{} that contain $v$ and set them to zero, in case $|H|$ is smaller than $\Degree{v}\cdot \deglow$.
If $v\rightarrow\High$, we increase $uLv[v, h]$ by 1 for all $h \in \NeighNH{v}{\ell}$ where $\ell \in \NeighL{v}$. This runs in $\bigO(m^{2\varepsilon})$.

(\AuxuLLv{} \DrawuLLv{}) We find all $z \in \NeighH{y}$ for $y\in \NeighL{x}$ and $x\in \NeighL{v}$. If $v\rightarrow\Low$, we set $\AuxuLLv[v,z]=0$ and if $v\rightarrow\High$, we increase $\AuxuLLv[v,z]$ by $1$.
In practice, when $v\rightarrow\Low$ and $|\High|<\Degree{v} \cdot \deglow^2$, instead of finding $z$, it is faster to find all entries of \AuxuLLv{} that contain $v$ and set them to zero.
Hence, the time for a partition swap here is $\deglow^3\in\bigO(m^{3\varepsilon})$.

(\AuxcLV{} \DrawcLV{}) We find $ h \in \NeighH{l}$ with $l \in \NeighL{v}$.
If $v\rightarrow\Low$ set $\AuxcLV[v,h]=0$.
If $v\rightarrow\High$, increase $\AuxcLV[v,h]$ by $\Degree{l}-2$.
In practice, when $|H|<\Degree{v}\cdot \deglow$, it is faster to find all entries of \AuxcLV{} that contain $v$ and set them to zero instead of finding $h$ as described above.
The \AuxcLV{} vertex change takes $\bigO(m^{2\varepsilon})$ time.

(\AuxpLL{} \DrawpLL{}) The \AuxpLL{} has two positions for $v$ as anchor: in the triangle or in the arm. 
If $v\rightarrow\Low$, we decrease, and if $v\rightarrow\High$, we increase the count.
For $v$ as anchor in the triangle, we find $x, y \in \NeighL{v}$ with $\{x, y\} \in E$. Specifically, we adjust $\AuxpLL[v,h]$ by 1 for $h$ in $\NeighH{x}$ or $\NeighH{y}$, and by 2 if it is in both neighborhoods. 
For $v$ as anchor in the arm, we find $h\in \NeighNH{v}{l}$ and $x\in \NeighNL{v}{l}$ with $l\in \NeighL{v}$ such that $\{h, x\} \in E$ and adjust $\AuxpLL[v,h]$ by 1.
This computation takes $\bigO(m^{3\varepsilon})$ time.
In practice, in the case $v\rightarrow\Low$ and $|H|<\Degree{v} \cdot\deglow^2$, it is faster to find all entries of \AuxpLL{} that contain $v$ and set them to zero instead of finding $h$ as described above.

(\AuxcL{} \DrawcL{}) For $l \in \NeighL{v}$ and $x, y \in \NeighNH{v}{l}$, if $v\rightarrow\Low$, we decrease, and if $v\rightarrow\High$, we increase $\AuxcL[v, x, y]$ by 1.

We now describe the adjustment to the query algorithm when restricting all anchors for auxiliary structures to high-degree vertices. 
As for claws and 4-cliques no auxiliary structures are used they remain unchanged.
For the other subgraphs the technique is straightforward, we now iterate over the neighbors of low-degree vertices for which we do not store auxiliary counts anymore.
To update the subgraph count after an edge update $e=\{u,v\}$, we compute the subgraphs containing $e$ and adjust the total by adding or subtracting this count for edge insertion or deletion, respectively.

(\tri) Let $x$ be the third vertex of a triangle containing $e$. To obtain $T_\High$ the number of triangles with 
$x\in\High$, we follow \cite{DBLP:conf/sand/HanauerHH22} and count  $T_\High=|\{\{v,h\}\in E \mid h \in \NeighNH{v}{u}\}|$.
We adapt the way to obtain the number $T_\Low$ of triangles with $x\in\Low$.
Before, $T_\Low=\AuxuLv[u,v]$, but now we can only use this if $u,v \in \High$.
Otherwise, let w.l.o.g.\ $u\in\Low$. Then $T_\Low=|\{\{v,\ell\}\in E \mid \ell \in \NeighNL{v}{u}\}|$.

(\tPath) 
Counting 3-paths containing $e$ as a centerpiece is done as in \cite{DBLP:conf/sand/HanauerHH22} with $\DegreeN{v}{u}\cdot\DegreeN{u}{v}$ minus the number of triangles containing $e$.
We only need to adapt the way to obtain the number of such triangles with a third vertex $x\in\Low$ as described above.
We next consider a non-center edge $e = \{v, u\}$ in a 3-path and apply the following twice (with swapped roles of $u$ and $v$):
Let $e$ be the start of a path, \ie~$v-u-x-y$. If $x\in\High$, we follow the description in \cite{DBLP:conf/sand/HanauerHH22}.
For the case that $x\in\Low$,
we proceed as follows:
    (i) If $u,v\in\High$, directly use Lemma~\ref{lemma:hhhthreepath}.
    (ii) If $u\in\High$ and $v\in\Low$, proceed as in Lemma~\ref{lemma:hhhthreepath}, but compute \AuxuLv[$u,v$] directly as $\sum_{x\in \NeighL{v}}\mathbf{1}_{\{\{u,x\}\in E\}}$.
    (iii) If $u \in \Low$, the number of 3-paths is $\sum_{x \in \NeighNL{v}{u}}\left(\DegreeN{u}{x} - \mathbf{1}_{\{\{x,v\}\in E\}}\right)$, as we first count all paths $u-x-y$ for $x\ne v$ as $\DegreeN{u}{x}$ and then correct this count for the case where $y=v$ with $\mathbf{1}_{\{\{x,v\}\in E\}}$.

(\paw) To count paws containing $e$, we distinguish the position of $e$ to be (1) the arm of the paw or (2) adjacent to the arm and (3) the third triangle edge non-adjacent to the arm.
We need the count $t_{\{u,v\}}$ of triangles with an edge $e$ as obtained above and triangles $t_u$ ($t_v$) containing vertex $u$ ($v$) as computed by \cite[Lemma 32]{hhh_full_version}.
Then, similar to \cite{hhh_full_version}
we count (1) as $t_u+t_v-2t_{\{u,v\}}$ and (2) as $t_{\{u,v\}} \cdot (\DegreeN{v}{u}+\DegreeN{u}{v}-2)$.  In \cite{hhh_full_version}, case (3) was counted as $\AuxcLV[u,v]$ plus those where the central vertex has high degree. Now, we can only use $\AuxcLV$ if $u,v \in \High$. Otherwise \wlogtext{} $u\in\Low$, we count $\sum_{x\in \{\NeighN{v}{u} | \{x, v\} \in E\}} (\Degree{x}-2)$.

(\fCycle) To count 4-cycles containing $e$, when $u, v \in \High$ we apply \cite[Lemma 40]{hhh_full_version}. However, when at least one endpoint of $e$ is in $\Low$, we can not use the same auxiliary structures anymore. For the case only one endpoint is in $\Low$, let \wlogtext{} $u \in \Low$ and $v\in\High$. Then, we find $b \in \NeighN{v}{u}$. If $b \in \High$ we add $\AuxuLv[v, b] -1 + \AuxuHv[v, b]$ and if $b\in\Low$ we add $|\{\{a, v\} \in E| a \in N_{\overline{v}, \overline{u}}(b)\}|$ to the count.
If both $u, v \in \Low$, we simply count $|\{\{a, b\} \in E | b \in \NeighN{v}{u},a \in \NeighN{u}{v}\}|$.

(\dia) Counting diamonds containing $e$, when $u, v \in \High$ we apply \cite[Lemma 16]{hhh_full_version}.
If $u \in \Low$ and/or $v \in \Low$, we need to make adaptions. We distinguish the position of $e$ to be (1) the chord, (2) part of the 4-cycle.
For (1) we apply the exact description from the proof in \cite[Lemma 16]{hhh_full_version}.
For (2), if $u, v \in \Low$, we iterate over $x \in \NeighN{v}{u}$ and $y \in N_{\overline{u},\overline{x}}(v)$, checking whether $\{x, y\} \in E$. For each chord present from $\{\{x, v\}$, $\{y, u\}\}$, we increase the count by 1. Otherwise, \wlogtext, let $u \in \Low$ and $v \in \High$. 
First, we count all diamonds where $u$ is not adjacent to the chord. For each $x \in \NeighN{v}{u}$ with $\{x, v\} \in E$ as the chord, we count the number of triangles containing $\{x, v\}$ but not $u$. If $x \in \High$, the count is $\AuxuLv[x, v] - 1 + \mathrm{uHv}[x, v]$. Otherwise, we iterate over $y \in N_{\overline{u},\overline{v}}(x)$ and count each $y$ such that $\{y, v\} \in E$. 
Second, we count all diamonds where $u$ is adjacent to the chord, but $e$ is not the chord. For $x, y \in \NeighN{v}{u}$, we increase the count by 1 if $\{x, y\}, \{y, v\} \in E$, and by 1 if $\{x, y\}, \{x, v\} \in E$.
\subsection{Optimizing the \hhh{} Algorithm: \poptref{hhh-extra-aux} Extra Use of Auxiliary Structures}
\label{sec:hhh:extra_aux:details}

When counting triangles or 4-cliques, we can use additional auxiliary structures without changing asymptotic running times.
These dependencies are highlighted as dashed arcs in \autoref{fig:hhh_structures_new}.
For counting triangles that contain an edge $e=\{u,v\}$ and a third vertex $x$, where $u,v, x \in \High$, we look up $\mathrm{uHv}[u,v]$ instead of iterating over all vertices in $\High$.
As $\mathrm{uHv}$ can be maintained in $\bigO(m^{\max(1-\varepsilon,\varepsilon)})$ time~\cite{DBLP:conf/sand/HanauerHH22}, this does not affect the overall cost to maintain triangle counts, which also is $\bigO(m^{\max(1-\varepsilon,\varepsilon)})$ when only using \AuxuLv{}.
As mentioned earlier, 4-clique counts can be maintained in worst-case $\bigO(m)$ time without using any auxiliary counts.
If we use $\AuxcL$ with high-degree anchors, \ie, \poptref{high-anchors-hhh},
and an edge $\{u, v\}$ with $u, v \in \High$ is inserted or deleted, we use $\AuxcL$ to count 4-cliques $\{u, v, x, y\}$ where $x\in\Low$ and $y\in\High$ 
by adding up $\AuxcL[u,v,y]$ for each $y \in \High$.
This requires a single pass over $\High$ and thus takes $\bigO(m^{1-\varepsilon})$ time.
As $\AuxcL$ can be maintained in $\bigO(m^{2\varepsilon})$ and the update time for other cases remains $\bigO(m)$, we obtain the same (now amortized) update time as long as $\varepsilon \leq \frac{1}{2}$.

\paragraph{Experimental Evaluation.}
\label{sec:hhh:extra_aux:exp}
To evaluate the practical improvement we compare \texttt{\hhhohe{}|\origepstab{.417}} and \texttt{\hhhoh{}|\origepstab{.417}} on dataset \EngSet{}.

\begin{figure}[t]

\includegraphics[width=\linewidth]{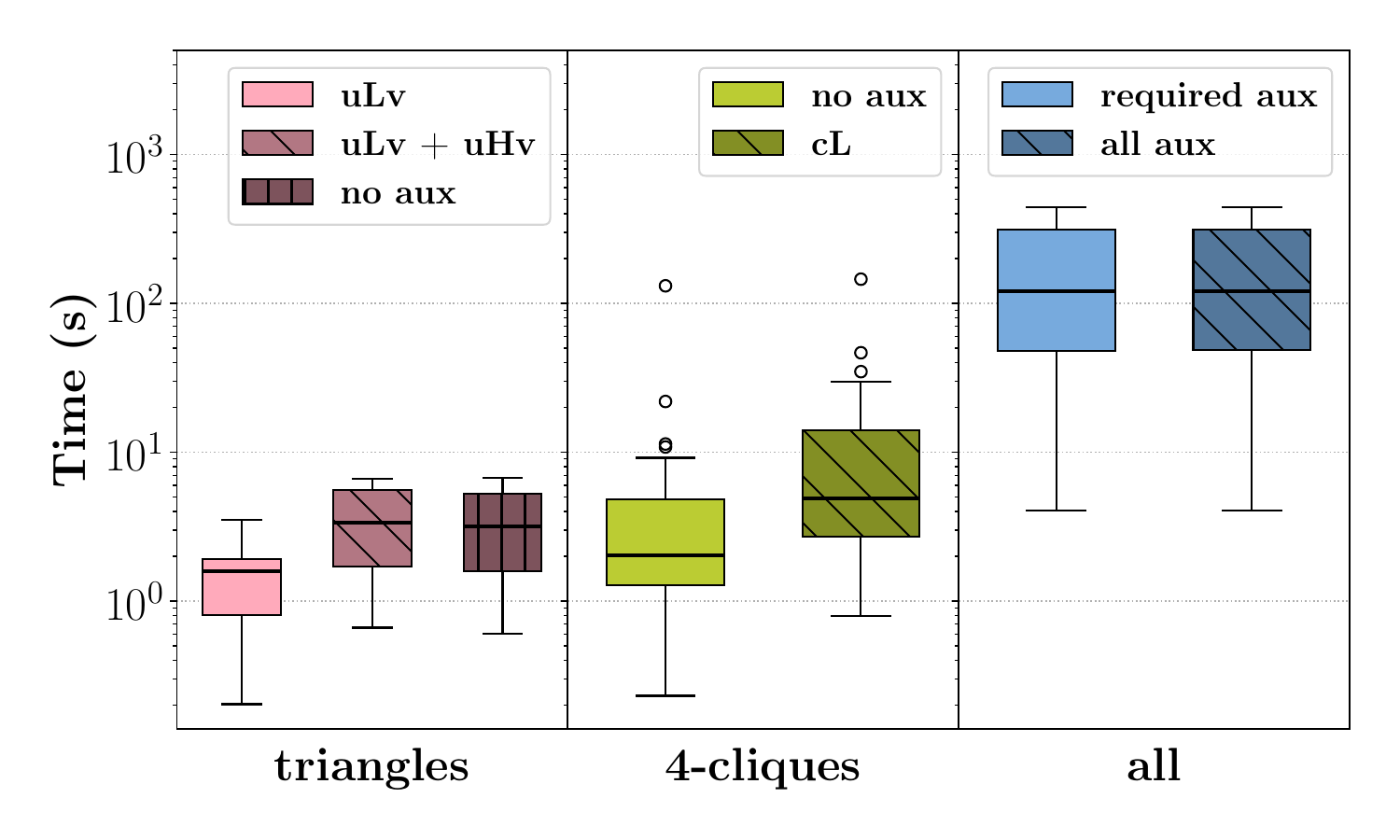}
\caption{Comparing the total update times for different patterns using more  auxiliary structures than required in the \hhh{}  algorithm using \origepstab{} \epsT{} with $\varepsilon=0.417$ on dataset \EngSet{}.} 
\label{exp:hhh_add_aux}
\end{figure}

If we are only interested in 4-cliques or triangles, not maintaining (and using) any auxiliary structures is $2.59\times$ (4-cliques) and $2.25\times$ (triangles) faster. 
However, when counting all patterns simultaneously, using all available auxiliary structures performed similarly (within $2\%$).
While \poptref{hhh-extra-aux} reduces neighbor traversals, increased lookup times and branch prediction misses slightly offset the gains, leading to marginally  faster running times in the geometric mean ($0.1\%$)
with \poptref{hhh-extra-aux}.
\subsection{Bounding Reclassifications during Major Rebalancing}
\label{sec:optimizing_epstab:proofs:reclassified_bound}
\restatethelemma{n_wrong_vertices}
\begin{proof}
Let $m_i$ and $m_{i+1} = m$ be the number of edges in the last and current MR step, respectively.
Let $\High$ and $\High'$ be the set of high-degree vertices immediately before and right after the MR, respectively.
Note that the number of current edges $m$ remains unchanged during an MR.
By definition and the properties of the \epsT, $m_i \in\Theta(m)$, so $|\High|\in\bigO(m_i^{1-\varepsilon})=\bigO(m^{1-\varepsilon})$, and $|\High'|\in\bigO(m^{1-\varepsilon})$.
In the worst case,
$\High \cap \High' = \emptyset$, so $|\High| + |\High'| = \bigO(m^{1-\varepsilon})$ vertices are classified differently after an MR than before.
\end{proof}

\subsection{Cost of \softepstab{} Major Rebalancing}
\label{sec:optimizing_epstab:proofs:soft}

\restatethelemma{soft}
\begin{proof}
Let $n^\prime$ be the number of vertices that need to change partition in an MR step, where $n' \in \bigO(m^{1-\varepsilon})$ by Lemma~\ref{lemma:n_wrong_vertices}.
Thus, an MR step with \softepstab{} requires
$\bigO(T_\text{VPU} \cdot n') = \bigO(m^\varepsilon \cdot T_\text{aux} \cdot m^{1-\varepsilon}) = \bigO(m \cdot T_\text{aux})$ time in total.

\end{proof}

\subsection{Degree Bound Degradation under \lazyepstab{}}
\label{sec:optimizing_epstab:proofs:lazy}

\restatethelemma{lazy}
\begin{proof}
If $\varepsilon= 0$, $\Low = \emptyset$, and if $\varepsilon = 1$, $\Degree{v} \in \Theta(m)$, so the statement trivially holds.
Otherwise, assume we are in epoch $i$.
Let $v \in \Low$ with $\Degree{v} = \lfloor\frac{1}{2}\theta_i\rfloor - 1$.
Delete all edges not incident to $v$.
Then, $\Degree{v} = m$, and, as no edge incident to $v$ was updated, $v$ remains in $\Low$.
\end{proof}

\subsection{Cost of \lateepstab{} Major Rebalancing}
\label{sec:optimizing_epstab:proofs:late}

\restatethelemma{late}
\begin{proof}
Since \lateepstab{} performs the same steps as \softepstab{}, it has the same total cost of $\bigO(m\cdot T_\text{aux})$ by Lemma~\ref{lemma:soft}
and (1) follows.
By construction, all vertex checks are done within the current epoch.
For (2), recall that there may be up to two VPUs per update in the original \epsT{}, which result from checking the end vertices of the updated edge.
With \lateepstab{}, we perform $\RebalanceFactor(\RebalanceFactor - 1)^{-1}$ VPUs per update in addition and hence $\bigO(1)$ in total per update.
Note that the threshold $\theta_i$ changes only by a constant factor between one epoch and the next.
Hence, each vertex that needs to change partition still has degree in $\bigO(m^{\varepsilon})$ when the change is performed.
\end{proof}

\subsection{Impact of \epsT{}s on Worst-case Update Times in \hhh{}}
\label{sec:worst-case-times:proofs}
We analyze the \emph{worst-case} update times for each of the four variants of the \epsT{}:

\restatethelemma{epst-worstcase}

After an edge update, the following steps are performed: the partition is updated in $\bigO(1)$; if triggered, an MR is executed at cost $T_{\text{MR}}$; if required, VPUs of the endpoints are performed at cost $T_{\text{VPU}}$; auxiliaries are updated at cost $T_{\text{aux}}$; and counts are updated using auxiliary structures at cost $T_{\text{count}}$. The worst-case update time is thus $\bigO(T_{\text{count}} + T_{\text{MR}} + T_{\text{VPU}} + T_{\text{aux}})$.

To derive these terms, we use $\tau_{\text{aux}}$ and $\tau_{\text{VPU}}$: the cost of an edge update and a VPU to a single auxiliary structure, respectively, derived from~\cite{hhh_full_version} (see \autoref{tab:P:hhh:aux:short}). The total cost $T_{\text{aux}}$ (resp.\ $T_{\text{VPU}}$) of a subgraph is the sum of $\tau_{\text{aux}}$ (resp.\ $\tau_{\text{VPU}}$) over all its required auxiliary structures. For example, $T_{\text{aux}}(\text{3-path})=\tau_{\text{aux}}(\AuxvLV)+\tau_{\text{aux}}(\AuxuLv)$.

\begin{corollary}\label{remark:orig-soft:worstcase}
Using \origepstab{}/\softepstab{}, the worst-case update time for subgraph counting is
$\bigO(m^{3/2})$ for triangles and 3-paths, and
$\bigO(m^{5/3})$ for paws, 4-cycles, and diamonds.
\end{corollary}

\begin{proof}

Under \origepstab{}, an MR recomputes all auxiliaries from scratch in a
single step, so $T_{\text{MR}} \in \bigO(m \cdot T_{\text{aux}})$, which
dominates all other terms. Using \softepstab{} an MR has the same time by Lemma~\ref{lemma:soft}.
Taking $T_{\text{aux}}$ per auxiliary from \cref{tab:P:hhh:aux:short},
summing over the auxiliaries required by each subgraph, and substituting
$\deglow \in \bigO(m^{\varepsilon})$, $|\High| \in \bigO(m^{1-\varepsilon})$
with the optimal $\varepsilon^*$ from~\cite{hhh_full_version}
yields the bounds listed in \cref{tab:worstcase}.
\end{proof}

\begin{table}[t!]
\centering\small
\caption{Per-auxiliary costs (in $\bigO$)~\cite{hhh_full_version}.
$\tau_{\text{aux}}$/$\tau_{\text{VPU}}$: cost to perform an edge update/a VPU to a single auxiliary structure.\\
}
\label{tab:P:hhh:aux:short}
\begin{tabular}{lcc}
\toprule
Auxiliary Structure & $\tau_{\text{aux}}$ & $\tau_{\text{VPU}}$ \\
\midrule
\DrawvLV{}  \AuxvLV  & $\deglow$              & $\deglow$ \\
\DrawuLv{}  \AuxuLv  & $\deglow$              & $\deglow^2$ \\
\Drawt{}    \Auxt    & $|\High|$    & $\deglow^2$
\\
\DrawuLLv{} \AuxuLLv & $\deglow^2$            & $\deglow^3$ \\
\DrawcLV{}  \AuxcLV  & $\deglow^2$            & $\deglow^2$ \\
\DrawpLL{}  \AuxpLL  & $\deglow^2$            & $\deglow^3$ \\
\DrawuHv{}  \AuxuHv  & $|\High|$              & $\deglow^2 + \deglow \cdot |\High|$ \\
\DrawcL{}   \AuxcL   & $\deglow^2$            & $\deglow^3$ \\
\bottomrule
\end{tabular}
\end{table}

\begin{table*}[t!]
\centering\small
\caption{Worst-case update times (in $\bigO$) to count subgraphs with \hhh{} using different partitions. Recall that for \origepstab{}, \softepstab{} and \lateepstab{} $\deglow \in\bigO(m^{\varepsilon})$, $|\High| \in\bigO(m^{1-\varepsilon})$ and for \lazyepstab{} $\deglow\in\bigO(m)$.\\
$T_{\text{aux}}$/$T_{\text{VPU}}$: cost for edge update/VPU to all required auxiliary structures (combined from \autoref{tab:P:hhh:aux:short})\\
$T_{\text{count}}$: counting cost given the auxiliaries (from~\cite{hhh_full_version} and \cref{sec:hhh:extra_aux:details}).\\
Worst-case update times using\\
\lateepstab{}/\lazyepstab{}: $\bigO(T_{\text{count}} + T_{\text{aux}} + T_{\text{VPU}})$, or\\
\origepstab{}/\softepstab{}: $\bigO(T_{\text{count}} + T_{\text{aux}} +T_{\text{VPU}} + m \cdot T_{\text{aux}})$.\\
Shaded terms apply only with \poptref{high-anchors-hhh}.
}
\label{tab:worstcase}
\resizebox{\linewidth}{!}{\begin{tabular}{r|l|c|c|c|c|c|c}
 Subgraph & Required auxiliaries & $T_{\text{aux}}$ & $T_{\text{VPU}}$ & $T_{\text{count}}$ & \lateepstab{} & \origepstab{}/\softepstab{} & \lazyepstab{}\\
\hline
triangle \tri
  & \AuxvLV
  & $\deglow$
  & $\deglow$
  & $|\High|$ \bgcolortext{gray}{${}+\deglow$}
  & $m^{1/2}$ & $m^{3/2}$ & $m$ \\
3-path \tPath
  & \AuxvLV, \AuxuLv
  & $\deglow$
  & $\deglow^2$
  & $|\High|$ \bgcolortext{gray}{${}+\deglow$}
  & $m$ & $m^{3/2}$ & $m^2$ \\
paw \paw
  & \Auxt, \AuxuLv, \AuxcLV
  & $|\High| + \deglow^2$
  & $\deglow^2$
  & $\deglow^2$
  & $m^{2/3}$ & $m^{5/3}$ & $m^2$ \\
4-cycles \fCycle
  & \AuxuLv, \AuxuLLv, \AuxuHv
  & $|\High| + \deglow^2$
  & $\deglow^3 + \deglow \cdot |\High|$
  & $|\High| + \deglow^2$
  & $m$ & $m^{5/3}$ & $m^3$\\
diamond \dia
  & \AuxuLv, \AuxpLL, \AuxcL
  & $\deglow^2$
  & $\deglow^3$
  & $|\High| + \deglow^2$
  & $m$ & $m^{5/3}$ & $m^3$\\
\end{tabular}}
\end{table*}

\begin{corollary}\label{remark:lazy:worstcase}
Using \lazyepstab{}, the worst-case update time for subgraph counting is
$\bigO(m)$ for triangles, $\bigO(m^2)$ for 3-paths and paws, and
$\bigO(m^3)$ for 4-cycles and diamonds.
\end{corollary}

\begin{proof}
Under \lazyepstab{}, an MR takes $O(1)$ time, but $\deglow\notin\bigO(m^\varepsilon)$ (Lemma~\ref{lemma:lazy}) but $\deglow\in\bigO(m)$ like any degree.
Hence, the worst-case update time is $\bigO(T_{\text{count}} + T_{\text{aux}} + T_{\text{VPU}})$.
The pattern costs $T_{\text{aux}}$ and $T_{\text{VPU}}$ are listed in
\cref{tab:worstcase}.

\end{proof}

\begin{corollary}\label{lemma:late:worstcase}
With \lateepstab{} rebalancing, the worst-case update time for subgraph
counting via \hhh{} and \epsT{} is:
$\bigO(m^{1/2})$ for triangles,
$\bigO(m^{2/3})$ for paws, and
$\bigO(m)$ for 3-paths, 4-cycles, and diamonds, while preserving
amortized bounds.
\end{corollary}

\begin{proof}
Under \lateepstab{}, the VPUs from an MR are spread over $\Omega(m)$
subsequent updates (Lemma~\ref{lemma:late}), so each update performs $\bigO(1)$
VPUs and no MR cost is incurred, giving a worst-case of
$\bigO(T_{\text{count}} + T_{\text{aux}} + T_{\text{VPU}})$.
The per-auxiliary costs $T_{\text{aux}}$ and $T_{\text{VPU}}$ are listed in
\cref{tab:P:hhh:aux:short}.
For each subgraph, we sum $T_{\text{aux}}$ and $T_{\text{VPU}}$ over all
required auxiliaries and substitute
$\deglow \in \bigO(m^{\varepsilon})$,
$|\High| \in \bigO(m^{1-\varepsilon})$ with the optimal $\varepsilon^*$
from~\cite{hhh_full_version}.
\cref{tab:worstcase} summarizes the results.
For example, triangles use only \AuxvLV{} with
$\varepsilon^* = \tfrac{1}{2}$:
$T_{\text{count}} + T_{\text{aux}} + T_{\text{VPU}}
= \bigO(m^{1/2} + m^{1/2} + m^{1/2}) = \bigO(m^{1/2})$.
Amortized bounds are preserved since the total rebalancing work over each
window remains unchanged.
\end{proof}

The amortized update times are computed as $\bigO(T_{\text{count}} + T_{\text{VPU}} \cdot \changesperupdate + T_{\text{aux}})$~\cite{DBLP:conf/sand/HanauerHH22}.
Hence, \softepstab{}/\lateepstab{} have the same amortized times as \origepstab{} in~\cite{DBLP:conf/sand/HanauerHH22}.
For \lazyepstab{} we can argue that $\changesperupdate \in \bigO(1/m)$ by a similar argument as Lemma~\ref{lemma:hindexc}. Hence, from~\autoref{tab:worstcase} the times in~\autoref{tab:worstcase-times} follow.
\subsection{Proof of Correctness for \egst{} 3-Path Counting}
\label{sec:egst:3path_proof}

\begin{restatable}{lemma}{egstthreepath}
\label{lemma:egstthreepath}
Assume $\Sc{0}$ and $\Sc{2}$ are maintained. Let $\{v,u\}$ be an edge with $u\in\High$. The number of 3-paths $v\text{-}u\text{-}x\text{-}y$ with $x\in\Low$ can be counted in $\bigO(h)$ time as $\Sc{0}[u] - C + \sum_{h\in\High\setminus\{v\}}\Sc{2}[u,h] - D$, where, if $v\in\Low$, $C = |\NeighL{v}| + |\{\{u,\ell\}\in E \mid \ell\in\NeighL{v}\}|$ and $D = |\{\{v,h\}\in E \mid h\in\High\}|$, and $C=D=0$ otherwise.
\end{restatable}

\begin{proof}
    To count paths with $y\in\Low$, recall that $\Sc{0}[u]$ counts paths $u\text{-}x\text{-}y$ with $x,y\in\Low$. Two cases counted in $\Sc{0}[u]$ do not yield valid 3-paths $v\text{-}u\text{-}x\text{-}y$: (1) $x=v$ and (2) $y=v$. Case (1) contributes $|\NeighL{v}|$ invalid paths if $v\in\Low$, and 0 otherwise. Case (2) contributes $|\{\ell\in\NeighL{v}\mid\{u,\ell\}\in E\}|$ if $v\in\Low$, and 0 otherwise. Hence, paths with $y\in\Low$ are counted by $\Sc{0}[u]-C$ where $C=|\NeighL{v}|+|\{\ell\in\NeighL{v}\mid\{u,\ell\}\in E\}|$ if $v\in\Low$ and 0 otherwise.
    
    To count paths with $y\in\High$, recall that $\Sc{2}[u,h]$ counts vertices $x\in\Low$ adjacent to both $u$ and $h\in\High$. Summing $\Sc{2}[u,y]$ over $y\in\High\setminus\{v\}$ counts all such paths, but includes invalid ones where $x=v$. The number of these is $|\{h\in\High\mid\{v,h\}\in E\}|$ if $v\in\Low$, and 0 otherwise. Hence, paths with $y\in\High$ are counted by $\sum_{h\in\High\setminus\{v\}}\Sc{2}[u,h]-D$ where $D=|\{h\in\High\mid\{v,h\}\in E\}|$ if $v\in\Low$ and 0 otherwise.
    
    The running time is $\bigO(|\High|+\deglow)=\bigO(h)$, as we iterate over all elements of $\High$ or neighbors of low-degree vertices.
\end{proof}

\subsection{Direct Partition Swaps of the \egst{} Algorithm}
\label{sec:egst:direct:details}

As an alternative to performing partition swaps via simulating repeated edge updates~\cite{DBLP:journals/tcs/EppsteinGST12}, we describe how to update auxiliary structures directly by scanning the 2-hop neighborhood of the moving vertex.
Consider a vertex $v$ moving from $\Low$ to $\High$ (the reverse case is symmetric). Note that upon a partition change $\Degree{v}\in\bigO(\deglow)$ and for the \hInd{} partition $\deglow\in\bigO(h)$. For each structure $\Sc{i}$, we first reduce counts where $v$ served as a low-degree non-anchor vertex, then increase counts where $v$ becomes an anchor. 
All partition swaps have the same running time as in \cite{DBLP:journals/tcs/EppsteinGST12} with $\bigO(h^2)$ for $\Sc{0}-\Sc{2}$ and $\bigO(h^3)$ for $\Sc{3}-\Sc{7}$.

\begin{enumerate}[label=$S_{\arabic*}$, start=0]

\item \DrawSZero\ (anchor $a\in\High$, path $a$-$b$-$c$ with $b,c\in\Low$).
\emph{Reduce:} for each $b\in\NeighL{v}$ and $h\in\High$, reduce $\Sc{0}[h]$ by $[\{b,h\}\in E]$ (case $v=c$) and by $[\{v,h\}\in E]$ (case $v=b$).
\emph{Increase:} for each $b\in\NeighL{v}$, increase $\Sc{0}[v]$ by $\DegreeL{b}$.
Time: $\bigO(\deglow\cdot|\High|)$.

\item \DrawSOne\ (anchor $a\in\High$, triangle $a$-$b$-$c$ with $b,c\in\Low$).
\emph{Reduce} ($v=b$ by symmetry): for each $c\in\NeighL{v}$ and $a\in\NeighH{v}$ with $\{c,a\}\in E$, reduce $\Sc{1}[a]$ by 1.
\emph{Increase} ($v=a$): for each $b\in\NeighL{v}$, increase $\Sc{1}[v]$ by $|\{c\in\NeighL{b}:\{c,v\}\in E\}|$.
Time: $\bigO(\deglow^2)$.

\item \DrawSTwo\ (anchors $a,c\in\High$, path $a$-$b$-$c$ with $b\in\Low$).
\emph{Reduce} ($v=b$): for each pair $a,c\in\NeighH{v}$, reduce $\Sc{2}[a,c]$ by 1.
\emph{Increase} ($v=a$): for each $b\in\NeighL{v}$ and $c\in\NeighH{b}$, increase $\Sc{2}[v,c]$ by 1.
Time: $\bigO(\deglow^2)$.

\item \DrawSThreeH\ (anchors $a,d\in\High$, path $a$-$b$-$c$-$d$ with $b,c\in\Low$).
\emph{Reduce} ($v=b$): for each $a\in\NeighH{v}$, $c\in\NeighL{v}$, $d\in\NeighH{c}$, reduce $\Sc{3}[a,d]$ by 1.
\emph{Reduce} ($v=c$): for each $b\in\NeighL{v}$, $a\in\NeighH{b}$, $d\in\NeighH{v}$, reduce $\Sc{3}[a,d]$ by 1.
\emph{Increase} ($v=a$): for each $b\in\NeighL{v}$, $c\in\NeighL{b}$, $d\in\NeighH{c}$, increase $\Sc{3}[v,d]$ by 1.
Time: $\bigO(\deglow^3)$.
When anchors may have any degree~\cite{DBLP:journals/tcs/EppsteinGST12}, replace $\NeighH{\cdot}$ by $\Neigh{\cdot}$ at anchor positions; the bound remains $\bigO(\deglow^3)$.

\item \DrawSFour\ (anchors $a,d\in\High$, non-anchors $b,c\in\Low$, $b$ central).
\emph{Reduce} ($v=b$): for each $a,d\in\NeighH{v}$ and $c\in\NeighL{v}$, reduce $\Sc{4}[a,d]$ by 1.
\emph{Reduce} ($v=c$): for each $b\in\NeighL{v}$ and $a,d\in\NeighH{b}$, reduce $\Sc{4}[a,d]$ by 1.
\emph{Increase} ($v=a$): for each $b\in\NeighL{v}$, $c\in\NeighL{b}$, $d\in\NeighH{b}$, increase $\Sc{4}[v,d]$ by 1.
Time: $\bigO(\deglow^3)$.

\item \DrawSFive\ (anchors $a,c\in\High$ with $a$ in triangle, non-anchors $b,d\in\Low$, $b$ central).
\emph{Reduce} ($v=b$): for each $a,c\in\NeighH{v}$ and $d\in\NeighL{v}$ with $\{a,d\}\in E$, reduce $\Sc{5}[a,c]$ by 1.
\emph{Reduce} ($v=d$): for each $b\in\NeighL{v}$ and $a,c\in\NeighH{b}$ with $\{a,v\}\in E$, reduce $\Sc{5}[a,c]$ by 1.
\emph{Increase} ($v=a$): for each $b,d\in\NeighL{v}$ with $\{b,d\}\in E$ and $c\in\NeighH{b}$, increase $\Sc{5}[v,c]$ by 1.
\emph{Increase} ($v=c$): for each $b\in\NeighL{v}$, $a\in\NeighH{b}$, $d\in\NeighL{b}$ with $\{a,d\}\in E$, increase $\Sc{5}[a,v]$ by 1.
Time: $\bigO(\deglow^3)$.

\item \DrawSSix\ (anchors $a,c\in\High$, non-anchors $b,d\in\Low$).
\emph{Reduce} ($v=b$ by symmetry): for each $d\in\NeighL{v}$ and $a,c\in\NeighH{v}$ with $\{a,d\}\in E$ and $\{c,d\}\in E$, reduce $\Sc{6}[a,c]$ by 1.
\emph{Increase} ($v=a$): for each $b,d\in\NeighL{v}$ with $\{b,d\}\in E$ and $c\in\NeighH{b}$ with $\{c,d\}\in E$, increase $\Sc{6}[v,c]$ by 1.
Time: $\bigO(\deglow^3)$.

\item \DrawSSeven\ (anchors $a,c,d\in\High$, non-anchor $b\in\Low$).
\emph{Reduce} ($v=b$): for each triple $a,c,d\in\NeighH{v}$, reduce $\Sc{7}[a,c,d]$ by 1.
\emph{Increase} ($v=a$): for each $b\in\NeighL{v}$ and $c,d\in\NeighH{b}$, increase $\Sc{7}[v,c,d]$ by 1.
Time: $\bigO(\deglow^3)$.

\end{enumerate}
\subsection{Proof of the Gradual Factor for the \hInd{} Partition}
\label{sec:h_index:c:proof}

\begin{restatable}{lemma}{hindexc}
\label{lemma:hindexc}

    Let $h$ be the \hInd{} after an update operation.
    If a vertex $v$ moves from $\Low\to\High$ when $\Degree{v}\ge c\cdot h$ for any $c>1$ and from $\High\to\Low$ when $\Degree{v}<h$, the resulting partition has $\changesperupdate\in\bigO(1/h)$ amortized changes per update.
\end{restatable}

\begin{proof}
    We follow the proof scheme in ~\cite{DBLP:journals/jgaa/EppsteinS12}.
    Note that~\cite[Lemma 1]{DBLP:journals/jgaa/EppsteinS12} states that a sequence of operations during which $|\High|$ changes from $h$ to $h^\prime>h$  includes at least $(h^\prime-h)^2$ insert operations.

    Hence, when the \hInd{} changes from $h$ to $h^\prime=ch$ this requires at least $(c-1)^2 h^2$ insert operations.

    We consider the case that a vertex $v$ is removed from $\High$. Let $h$ be the \hInd{} when $v$ was inserted into $\High$. Let $h^\prime = \max(h,h_i)$ for any \hInd{} $h_i$ in an operation between the insertion and removal from $\High$.

    The authors of~\cite{DBLP:journals/jgaa/EppsteinS12} provided a charging scheme to show that any vertex in $\High$ remains in $\High$ for a series of updates. We adapt this charging scheme: For every edge insert we give each element $x\in\High$ a credit of $\frac{1}{(c-1)^2 h_i^2}$ and for every edge removal $\frac{1}{(c-1)h_i}$ adjacent to an $x\in\High$ to that $x$, where $h_i$ is the \hInd{} after the operation.

    To show that any $v$ removed from $\High$ still must have accumulated $\Omega(1)$ credits ~\cite[Lemma 2]{DBLP:journals/jgaa/EppsteinS12}, we consider two cases: (1) the \hInd{} grew over $c h$, then $v$ has credit at least $\frac{(c-1)^2 h^2}{(c-1)^2 (c h)^2} = \Omega(1)$ for any constant $c>1$. Now assume that, (2) the \hInd{}  was always at most $c h$. Then $v$ received at least $\frac{(h^\prime - h)^2}{(c-1)^2 (c h)^2}$ credit for insertions and  $\frac{ch - h^\prime}{c h}$ credit for deletions.
    Note that by definition $h \le h' \le ch$, then $\frac{(h^\prime - h)^2}{(c-1)^2 (c h)^2} + \frac{ch - h^\prime}{c h} = \Omega(1)$ for any constant $c>1$.
    
    Given that
    any element removed from $\High$ has gained $\Omega(1)$ credits,
    the proof that the partition has $\changesperupdate\in \bigO(1/h)$ amortized changes to the partition per update follows directly from~\cite[Theorem 2]{DBLP:journals/jgaa/EppsteinS12}. 
\end{proof}

\subsection{The \egst{} Algorithm with an Abstract Partition}
\label{sec:abstract_p:egst:deviation}
We follow~\cite{DBLP:journals/tcs/EppsteinGST12} and provide a detailed description of every update step to derive the amortized update times for each pattern. To this end we start with the auxiliary structures and use these to derive the patterns.

\paragraph{Auxiliary structure updates.} The amortized update time consists of edge updates which happen every update and partition swaps that only happen at an amortized rate of $\changesperupdate$ changes per update.
An edge update $e=\{u,v\}$ affects auxiliary structures only if at least one endpoint is in $\Low$~\cite{DBLP:journals/tcs/EppsteinGST12}. When exactly one endpoint is low, we assume \wlogtext{} that $u\in\Low$ and $v\in\High$. 
Below we describe only edge insertions; deletions are symmetric.
\begin{enumerate}[label=$S_{\arabic*}$, start=0]
\item \DrawSZero: If $u,v\in\Low$, for each $h\in \High$, increase $\Sc{0}[h]$ by 2 if $h\in\NeighH{u}\cap\NeighH{v}$, or 1 if $h$ is in exactly one of $\NeighH{u},\NeighH{v}$.
If $v\in\High$, increase $\Sc{0}[v]$ by $\DegreeL{u}$. Time: $\bigO(|\High|)$.

\item \DrawSOne: If $u,v\in\Low$, for each $h\in \NeighH{u}$ with $\{h,v\}\in E$ increase $\Sc{1}[h]$ by 1. If $v\in\High$, for each $\ell\in \NeighL{u}$ with $\{\ell,v\}\in E$, increase $\Sc{1}[v]$ by 1. Time: $\bigO(\deglow)$.

\item \DrawSTwo: If $u,v\in\Low$, no update is needed. If $v\in\High$, increase $\Sc{2}[v,h]$ by 1 for each $h\in \NeighNH{v}{u}$. Time: $\bigO(\deglow)$.

\item \DrawSThreeH: We first consider the variant where anchors must be high-degree (Section~\ref{egst:engineering}). If $u,v\in\Low$, for each pair $h_u \in \NeighH{u}$, $h_v \in \NeighH{v}$ with $h_u\neq h_v$, increase $\Sc{3}[h_u,h_v]$ by 1. If $v\in\High$, for each $\ell\in\NeighL{u}$ and $h \in \NeighH{\ell}$, increase $\Sc{3}[v,h]$ by 1. Time: $\bigO(\deglow^2)$. In the original formulation~\cite{DBLP:journals/tcs/EppsteinGST12}, anchors may have any degree; the update is analogous with $\Neigh{\cdot}$ replacing $\NeighH{\cdot}$ at anchor positions, also taking $\bigO(\deglow^2)$ time.

\item \DrawSFour: If $u,v\in\Low$, for each pair $h_1,h_2 \in \NeighH{u}$, increase $\Sc{4}[h_1,h_2]$ by 1; symmetrically for $\NeighH{v}$. If $v\in\High$, for each $h\in \NeighNH{v}{u}$, increase $\Sc{4}[v,h]$ by $\DegreeL{u}$. Time: $\bigO(\deglow^2)$.

\item \DrawSFive: If $u,v\in\Low$, for each pair $h_1,h_2 \in \NeighH{u}$ with $\{h_1,v\}\in E$, increase $\Sc{5}[h_1,h_2]$ by 1; symmetrically for $\NeighH{v}$ with adjacency to $u$. If $v\in\High$, for each $h\in \NeighNH{v}{u}$ and $\ell \in \NeighL{u}$, increase $\Sc{5}[v,h]$ by $[\{\ell,v\}\in E] + [\{\ell,h\}\in E]$, accounting for $e$ being in the triangle or the arm, respectively. Time: $\bigO(\deglow^2)$.

\item \DrawSSix: If $u,v\in\Low$, for each pair $h_1,h_2 \in \NeighH{u}\cap\NeighH{v}$, increase $\Sc{6}[h_1,h_2]$ by 1. If $v\in\High$, for each $h\in \NeighNH{v}{u}$ and $\ell \in \NeighL{u}$ with $\{\ell,v\}\in E$ and $\{\ell,h\}\in E$, increase $\Sc{6}[v,h]$ by 1. Time: $\bigO(\deglow^2)$.

\item \DrawSSeven: If $u,v\in\Low$, no update is needed. If $v\in\High$, for each pair $h_1,h_2\in \NeighNH{v}{u}$, increase $\Sc{7}[v,h_1,h_2]$ by 1. Time: $\bigO(\deglow^2)$.
\end{enumerate}
Partition swaps are handled by simulating a deletion of all incident edges followed by reinsertion,
costing $\bigO(\deglow \cdot \mathcal{T})$ where $\mathcal{T}$ is the per-edge update time for each structure~\cite{DBLP:journals/tcs/EppsteinGST12} or by the practical optimization as given in \autoref{sec:egst:direct:details} resulting in the same time.
\begin{lemma}
The amortized update time for auxiliary structures under partition \Partition{} is
$\bigO(|\High| + \deglow|\High|\cdot \changesperupdate)$ for \Sc{0},
$\bigO(\deglow + \deglow^2\cdot \changesperupdate)$ for \Sc{1}-\Sc{2}, and
$\bigO(\deglow^2 +\deglow^3\cdot \changesperupdate)$ for \Sc{3}-\Sc{7}.
\end{lemma}

\paragraph{Subgraph count updates.}
We describe how updating edge $e=\{u,v\}$ affects each subgraph count. Below we describe only edge insertions resulting in increase of the count; deletions are symmetric and decrease the count. In case at least one endpoint is in $\Low$, we assume this is $u$. Let $\bigO(\Sc{i})$ refer to the time it takes to maintain the auxiliary structure \Sc{i}.
 
\begin{enumerate}[label=\textbf{\alph*)}]
\renewcommand{\labelenumi}{}

\item[\claw] (claw): Increase the count by $\binom{\Degree{u}-1}{2} + \binom{\Degree{v}-1}{2}$, since each claw containing $e$ has its center at $u$ or $v$. Time: $\bigO(1)$.

\item[\tri] (triangle):\label{label_triangle} Increase the count by the number of vertices $x$ adjacent to both $u$ and $v$. If $u,v\in\High$, use $\Sc{2}[u,v]$ to count $x\in\Low$, and additionally for each $x\in\NeighH{u}$ checking adjacency to $v$. If $u\in\Low$, for each $x\in\Neigh{u}$ and test $\{x,v\}\in E$. Time: $\bigO(|\High|+\deglow) + \bigO(\Sc{2})=\bigO(|\High|+\deglow+\deglow^2\cdot\changesperupdate)$.

\item[\tPath] (3-path): If $e$ is the center edge ($x$--$u$--$v$--$y$), increment by $(\Degree{u}-1)(\Degree{v}-1)$ and subtract triangles through $e$, which are counted as in \tri{} above. If $e$ is a leaf edge, consider both orientations (once with swapped roles of $u$ and $v$): Assuming $v$--$u$--$x$--$y$, for each $x\in\NeighN{v}{u}$ we need to increment by $\Degree{x}-1-[\{x,v\}\in E]$. If $u\in\Low$, iterate over all such $x$ directly. If $u\in\High$, handle $x\in\High$ by iteration over $\High$ and $x\in\Low$ via Lemma~\ref{lemma:egstthreepath} using $\Sc{0}$ and $\Sc{2}$. Time: $\bigO(|\High|+\deglow) + \bigO(\Sc{0},\Sc{2})=\bigO(|\High|+\deglow+\deglow^2\cdot\changesperupdate)$.

\item[\paw] (paw): If $e$ is the arm, consider both orientations, once with swapped roles for $u$ and $v$: Assume $v$--$u$--$x$--$y$ with $\{u,y\}\in E$ and count triangles containing $u$ but not $e$. If $u\in\Low$, for each pair $x,y\in\NeighN{v}{u}$ with $\{x,y\}\in E$, increment by 1.
If $u\in\High$, use $\Sc{1}[u]$ for triangles with two low vertices; since $\Sc{1}[u]$ counts all such triangles through $u$, when $v\in\Low$ we must subtract $|\{x\in\NeighL{v}:\{x,u\}\in E\}|$ to exclude triangles that contain $v$ (as those involve $e$). For triangles with exactly one low vertex, sum $\Sc{2}[u,h]$ over $h\in\NeighNH{v}{u}$, where each $\Sc{2}[u,h]$ counts low vertices adjacent to both $u$ and $h$; again when $v\in\Low$, subtract 1 for each such $h$ with $\{h,v\}\in E$ to exclude triangles through $e$. For all-high triangles, iterate over pairs in $\NeighNH{v}{u}$ checking mutual adjacency directly.

If $e$ is in the triangle ($u$--$v$--$x$--$u$ with arm to $y$), find common neighbors $x$ of $u,v$; each such triangle contributes $\Degree{u}+\Degree{v}+\Degree{x}-6$ paws, one for each choice of an arm edge from $\{u,v,x\}$ to a vertex outside the triangle.
If at least one endpoint of $e$ is in $\Low$, this is handled by neighbor iteration. If $u,v\in\High$, let $n_t$ count the triangles through $e$. We split the contribution $\Degree{x}-2$ (arms from $x$) from $\Degree{u}+\Degree{v}-4$ (arms from $u$ or $v$), accumulating the latter globally over all triangles. For $x\in\Low$: increase $n_t$ by $\Sc{2}[u,v]$, which counts low vertices adjacent to both $u$ and $v$; for the $\Degree{x}-2$ arms, use $\Sc{4}[u,v]$ to count those with $y\in\Low$ 
and sum $\Sc{7}[u,v,y]$ over $y\in\High$.
For $x\in\High$: iterate over $x\in\NeighH{u}$, and if $\{x,v\}\in E$, increase $n_t$ by 1 and the paw count by $\Degree{x}-2$. Finally, increase the paw count by $n_t\cdot(\Degree{u}+\Degree{v}-4)$, accounting for all arms originating at $u$ or $v$ across all $n_t$ triangles.

Time: $\bigO(|\High|^2+\deglow^2) + \bigO(\Sc{1},\Sc{2},\Sc{4},\Sc{7})=\bigO(|\High|^2 + \deglow^2 + \deglow^3\cdot\changesperupdate)$.

\item[\fCycle] (4-cycle): Let $x,y$ be the two vertices not in $e$, forming the path $u$--$x$--$y$--$v$. We case-split on how many of $x,y$ lie in $\High$.

If $x,y\in\Low$: Use $\Sc{3}[u,v]$ which directly counts such 2-paths among low vertices. When applying \poptref{high-anchors-egst}, $\Sc{3}$ is only available for high-degree anchors, so if $u\in\Low$ we instead iterate over 2-paths among low-degree neighbors of $u$, costing an additional $\bigO(\deglow^2)$.

If exactly one of $x,y$ is high, say $x\in\High$: iterate over $h\in\NeighH{u}$ as candidates for $x$. When $v\in\High$, use $\Sc{2}[h,v]$ to count low vertices $y$ adjacent to both $h$ and $v$, correcting by subtracting 1 when $u\in\Low$ and $\{h,v\}\in E$ to exclude paths through $u$. When $v\in\Low$, instead check $\{h,x\}\in E$ for each $x\in\NeighNL{u}{v}$.

If $x,y\in\High$: iterate over all pairs $h_1\in\NeighH{u}$, $h_2\in\NeighH{v}$ and increase by 1 whenever $\{h_1,h_2\}\in E$.

Time with \poptref{high-anchors-egst}: $\bigO(\deglow^2+|\High|\cdot\deglow+|\High|^2) + \bigO(\Sc{2},\Sc{3})=\bigO(|\High|\cdot\deglow + \deglow^2 +|\High|^2 + \deglow^3\cdot\changesperupdate)$.

\item[\dia] (diamond): Edge $e$ is either the diagonal or a side of the outer 4-cycle; we sum both.

\emph{$e$ is the diagonal.} We need pairs of common neighbors of $u,v$. If $u\in\Low$, iterate over pairs in $\NeighN{v}{u}$ checking adjacency to $v$. If $u,v\in\High$, compute the total common-neighbor count $cn$ via $\Sc{2}[u,v]$ (low common neighbors) plus scanning $\NeighH{u}$ (high common neighbors); we then increase the diamond count by $\binom{cn}{2}$.

\emph{$e$ is on the outer cycle.} Let $u$--$x$--$y$--$v$ form the outer cycle. Then either $u$--$y$ or $v$--$x$ form the diagonal.\\
\emph{$x,y\in\High$:} Iterate over pairs in $\NeighNH{v}{u}$ (and $\NeighNH{u}{v}$), checking mutual adjacency and adjacency to the opposite endpoint.\\
\emph{At least one of $x,y$ in $\Low$:}
If $u,v\in\High$: 
We further distinguish (1) to count exactly one of $x,y$ low.
For $x\in\High$, iterate over $x\in\NeighNH{v}{u}$: if $v$--$x$ is the diagonal, require $\{x,u\}\in E$ and use $\Sc{2}[x,v]$; if $u$--$y$ is the diagonal, use $\Sc{7}[u,v,x]$. Symmetrically for $y\in\High$, iterate over $y\in\NeighNH{u}{v}$: if $u$--$y$ is the diagonal, require $\{y,v\}\in E$ and use $\Sc{2}[u,y]$; if $v$--$x$ is the diagonal, use $\Sc{7}[v,u,y]$.
And (2) to count both $x,y\in\Low$, use $\Sc{5}[u,v]$.\\
If $u\in\Low$: For diamonds with diagonal $x$--$y$, iterate over $x\in\NeighN{v}{u}$ with $\{x,y\}\in E$, count common neighbors of $x,y$ excluding $u$: if $x,y\in\High$ use $\Sc{2}[x,y]-1$; otherwise enumerate neighbors of the low endpoint checking adjacency. For diamonds with diagonal $u$--$y$, iterate over pairs $x,y\in\NeighN{v}{u}$ with $\{x,y\},\{y,v\}\in E$.

Time: $\bigO(|\High|^2+\deglow^2) + \bigO(\Sc{2},\Sc{5},\Sc{7})=\bigO(|\High|^2 + \deglow^2 +\deglow^3\cdot\changesperupdate)$.

\item[\fClique] (4-clique): Find two vertices $x,y$ such that $\{u,v,x,y\}$ forms a complete graph. If $u\in\Low$, iterate over pairs of neighbors $x,y$ of $u$ and check $\{x,v\},\{y,v\},\{x,y\}\in E$. If $u,v\in\High$, use $\Sc{6}[u,v]$ for $x,y\in\Low$; for exactly one high, iterate over $h\in\NeighH{u}\cap\NeighH{v}$ and use $\Sc{7}[u,v,h]$; for $x,y\in\High$, iterate over pairs in $\NeighH{u}\cap\NeighH{v}$ checking mutual adjacency. Time: $\bigO(|\High|^2+\deglow^2) + \bigO(\Sc{6},\Sc{7}) = \bigO(|\High|^2 + \deglow^2 + \deglow^3\cdot\changesperupdate)$.

\end{enumerate}

\begin{lemma}
The amortized update times for counting undirected four-vertex subgraphs using \egst{} with partition \Partition{} are listed in Table~\ref{tab:generalP}.
\end{lemma}

\subsection{Full Dataset Statistics}

The full statistics of the data used can be found in \autoref{tab:graphs} and for the engineering set in \autoref{tab:engeneering_set}.

\begin{table*}[!h]

\caption{
Full Data Set (\FullSet{}) statistics:
\#updates ($\delta$),
insertion ratio ($\delta_+$);
maximum \#non-isolated vertices ($\widehat{n}$),
edges ($\widehat{m}$), and degree ($\widehat{\MaxDegree}$);
average h-index ($\overline{h}$),
$\overline{h^3/m}$ ratio, and global clustering coefficient ($\overline{CC}$); and network type, \ie, one of
social network (Soc), interaction graph (Int), hyperlink network (HyL), internet topology (Top) or miscellaneous (Misc).
Shaded rows indicate membership of the engineering dataset \EngSet{}.
}
\label{tab:graphs}
\centering
\footnotesize

\renewcommand{\arraystretch}{0.9}
\newcommand{\colorrow}{\rowcolor{white!70!gray}}
\begin{tabular}{cl|cc|cc|cc|c|c|c}
\toprule
\tiny
\small
ID &Graph & $\delta$ & $\delta_+$ & $\widehat{n}$ & $\widehat{m}$ & $\widehat{\MaxDegree}$ & $\overline{h}$ & $\overline{h^3/m}$ & $\overline{CC}$ & Type \\
\midrule
\multicolumn{11}{c}{Wikipedia Hyperlink Networks (\DSWiki{})~\cite{konect,konect_wiki_second_ref}}\\
\midrule
\colorrow
A&simple\_wiki & 1.6M & \numprint{72}\% & 99.5K & 690.7K & \numprint{10054} & 230 & 38.93 & 8.2e-3 & HyL \\
\colorrow
B&nl\_wiki & 18.9M & \numprint{75}\% & 1M & 9.6M & \numprint{68864} & 723 & 84.76 & 2.3e-3 & HyL \\
\colorrow
C&pl\_wiki & 23.6M & \numprint{74}\% & 1M & 11.4M & \numprint{102331} & 847 & 113.22 & 2.9e-3 & HyL \\
\colorrow
D&it\_wiki & 33.1M & \numprint{74}\% & 1.2M & 15.7M & \numprint{71489} & 948 & 109.59 & 1.4e-3 & HyL \\
\colorrow
E&fr\_wiki & 56.3M & \numprint{70}\% & 2.2M & 22.3M & \numprint{113317} & 1.1K & 132.39 & 1.8e-3 & HyL \\
\colorrow
F&de\_wiki & 82M & \numprint{67}\% & 2.2M & 28.1M & \numprint{75577} & 1.2K & 123.85 & 2.9e-3 & HyL \\
\midrule
\multicolumn{11}{c}{Temporal Graphs (\DSTMP{})}\\
\midrule
A&reality-call~\cite{DBLP:conf/aaai/RossiA15,eagle2006reality} & 7.2K & \numprint{50}\% & 3.6K & 3.6K & \numprint{3599} & 1 & 0 & 0 & Int \\
B&wiki-elec~\cite{DBLP:conf/aaai/RossiA15,leskovec2010signed} & 201.5K & \numprint{50}\% & 7.1K & 100.8K & \numprint{1065} & 113 & 32.96 & 2.6e-2 & Misc \\
C& tech-as-topology~\cite{DBLP:conf/aaai/RossiA15} & 215.4K & \numprint{50}\% & 34.8K & 107.7K & \numprint{2760} & 108 & 26.29 & 9.2e-3 & Top \\
D&reddit-body~\cite{snapnets,kumar2018community} & 433.5K & \numprint{50}\% & 3.3K & 4K & \numprint{373} & 21 & 3.3 & 2.1e-2 & HyL \\
E&mathoverflow~\cite{snapnets,ParanjapeBL17,stackexchange_data_dump} & 458K & \numprint{50}\% & 1K & 2.5K & \numprint{137} & 17 & 3.91 & 1.6e-2 & Int \\
\colorrow
F&enron~\cite{DBLP:conf/aaai/RossiA15,cohen2005enron} & 594.9K & \numprint{50}\% & 87K & 297.5K & \numprint{1726} & 166 & 34.57 & 1.4e-2 & Int \\
G&reddit-title~\cite{snapnets,kumar2018community} & 824.7K & \numprint{50}\% & 4.9K & 7.8K & \numprint{500} & 33 & 6.21 & 1.1e-2 & HyL \\
\colorrow
H&askubuntu~\cite{snapnets,ParanjapeBL17,stackexchange_data_dump} & 952.4K & \numprint{50}\% & 4.9K & 6.1K & \numprint{339} & 26 & 4.93 & 5.5e-3 & Int \\
\colorrow
I&superuser~\cite{snapnets,ParanjapeBL17,stackexchange_data_dump} & 1.5M & \numprint{50}\% & 5.5K & 7.3K & \numprint{364} & 28 & 4.9 & 5.9e-3 & Int\\
\colorrow
J&wiki-talk~\cite{snapnets,ParanjapeBL17,DBLP:conf/icwsm/LeskovecHK10} & 6.3M & \numprint{50}\% & 74.4K & 102K & \numprint{29540} & 87 & 13.26 & 1.0e-3 & Int \\
\colorrow
K&youtube-growth~\cite{DBLP:conf/aaai/RossiA15,mislove-2007-socialnetworks} & 18.8M & \numprint{50}\% & 3.2M & 9.4M & \numprint{91751} & 612 & 55.04 & 3.1e-4 & Soc \\
\colorrow
L&stackoverflow~\cite{snapnets,ParanjapeBL17,stackexchange_data_dump} & 59.4M & \numprint{50}\% & 158.3K & 277.9K & \numprint{825} & 117 & 9.14 & 4.3e-3 & Int \\
\midrule
\multicolumn{11}{c}{Constant-Size Kronecker Graphs (\DScKron{})~\cite{DBLP:conf/alenex/HanauerH020,dyreach}}\\
\midrule
\colorrow
A&c\_answers & 2.6M & \numprint{54}\% & 74.3K & 185.5K & \numprint{578} & 112 & 8.04 & 7.2e-4 & Soc\\
\colorrow
B&c\_delicious & 3.2M & \numprint{54}\% & 93.7K & 225.9K & \numprint{306} & 80 & 2.48 & 5.2e-4 & Soc\\
C&c\_epinions & 5.9M & \numprint{54}\% & 85.7K & 413.7K & \numprint{2140} & 175 & 13.94 & 1.3e-3 & Soc\\
D&c\_flickr & 4.3M & \numprint{54}\% & 78.2K & 305.5K & \numprint{1544} & 151 & 12.09 & 1.2e-3 & Soc\\
E&c\_blog-nat05-6m & 18.7M & \numprint{54}\% & 120.4K & 1.3M & \numprint{3036} & 350 & 34.74 & 1.5e-3 & Soc\\
F&c\_blog-nat06all & 22.3M & \numprint{54}\% & 122.6K & 1.6M & \numprint{3431} & 403 & 44.33 & 1.6e-3 & Soc\\
G&c\_cit-hep-ph & 22.2M & \numprint{54}\% & 130.7K & 1.6M & \numprint{661} & 182 & 4.31 & 5.7e-4 & Soc\\
H&c\_cit-hep-th & 29.8M & \numprint{53}\% & 130.8K & 1.6M & \numprint{579} & 178 & 3.84 & 5.3e-4 & Soc\\
\colorrow
I&c\_ca-dblp & 8.5M & \numprint{54}\% & 127.2K & 597.3K & \numprint{202} & 82 & 1.01 & 4.7e-4 & Soc\\
J&c\_ca-gr-qc & 7.9M & \numprint{54}\% & 129.7K & 559.2K & \numprint{92} & 49 & 0.24 & 5.8e-4 & Soc\\
K&c\_ca-hep-th & 5.3M & \numprint{54}\% & 121.8K & 374.1K & \numprint{129} & 55 & 0.5 & 4.6e-4 & Soc\\
L&c\_ca-hep-ph & 30.4M & \numprint{54}\% & 130.8K & 2.1M & \numprint{933} & 256 & 8.54 & 7.7e-4 & Soc\\
\colorrow
M&c\_web-notredame & 4.1M & \numprint{54}\% & 85.4K & 288.7K & \numprint{916} & 151 & 12.63 & 9.1e-4 & HyL\\
N&c\_as-newman & 6.9M & \numprint{54}\% & 83.1K & 487.6K & \numprint{3308} & 219 & 23 & 9.5e-4 & Top\\
\colorrow
O&c\_as-routeviews & 7.9M & \numprint{54}\% & 86.5K & 556.9K & \numprint{3603} & 244 & 27.59 & 1.4e-3 & Top\\
P&c\_gnutella-30 & 2.9M & \numprint{54}\% & 95.6K & 203.4K & \numprint{334} & 75 & 2.21 & 5.2e-5 & Top\\
Q&c\_gnutella-25 & 3.6M & \numprint{54}\% & 103.9K & 254.1K & \numprint{389} & 86 & 2.7 & 5.5e-5 & Top\\
\colorrow
R&c\_atp-gr-qc & 1.3M & \numprint{54}\% & 86.4K & 90.9K & \numprint{26} & 14 & 0.04 & 9.4e-5 & Misc\\
S&c\_bio-proteins & 9.5M & \numprint{54}\% & 107.2K & 667.5K & \numprint{1773} & 203 & 13.44 & 3.2e-4 & Int\\
\midrule
\multicolumn{11}{c}{Growing Kronecker Graphs (\DSgKron{})~\cite{DBLP:conf/alenex/HanauerH020,dyreach}}\\
\midrule
A&g\_answers & 538.3K & \numprint{67}\% & 74K & 184.9K & \numprint{563} & 80 & 6.27 & 1.4e-3 & Soc \\
B&g\_delicious & 646K & \numprint{67}\% & 93.4K & 225.3K & \numprint{282} & 59 & 2.12 & 1.0e-3 & Soc \\
\colorrow
C&g\_epinions & 1.1M & \numprint{68}\% & 85.6K & 413K & \numprint{2023} & 135 & 14.42 & 2.1e-3 & Soc \\
\colorrow
D&g\_flickr & 853.4K & \numprint{68}\% & 77.8K & 304.7K & \numprint{1393} & 124 & 14.83 & 2.0e-3 & Soc \\
E&g\_blog-nat05-6m & 3.3M & \numprint{70}\% & 120.3K & 1.3M & \numprint{2912} & 238 & 23.22 & 2.3e-3 & Soc \\
\colorrow
F&g\_blog-nat06all & 3.9M & \numprint{70}\% & 122.6K & 1.6M & \numprint{3286} & 273 & 29.24 & 2.4e-3 & Soc \\
\colorrow
G&g\_cit-hep-ph & 3.9M & \numprint{70}\% & 130.7K & 1.6M & \numprint{622} & 137 & 3.94 & 9.7e-4 & Soc \\
H&g\_cit-hep-th & 3.9M & \numprint{70}\% & 130.8K & 1.6M & \numprint{568} & 134 & 3.67 & 9.2e-4 & Soc\\
I&g\_ca-dblp & 1.6M & \numprint{69}\% & 127.1K & 596K & \numprint{181} & 63 & 0.98 & 8.4e-4 & Soc \\
J&g\_ca-gr-qc & 1.5M & \numprint{69}\% & 129.5K & 558.4K & \numprint{79} & 39 & 0.27 & 9.7e-4 & Soc \\
K&g\_ca-hep-ph & 5.3M & \numprint{70}\% & 130.7K & 2.1M & \numprint{907} & 187 & 7.12 & 1.3e-3 & Soc \\
L&g\_ca-hep-th & 1M & \numprint{68}\% & 121.6K & 373.2K & \numprint{120} & 43 & 0.52 & 8.0e-4 & Soc \\
M&g\_web-notredame & 811.6K & \numprint{68}\% & 85.2K & 288.1K & \numprint{879} & 111 & 10.9 & 1.6e-3 & HyL \\
\colorrow
N&g\_as-newman & 1.3M & \numprint{69}\% & 82.9K & 486K & \numprint{3183} & 156 & 18.41 & 1.6e-3 & Top \\
O&g\_as-routeviews & 1.5M & \numprint{69}\% & 86.4K & 555.3K & \numprint{3551} & 170 & 20.39 & 2.2e-3 & Top \\
P&g\_gnutella-25 & 722.1K & \numprint{68}\% & 103.7K & 253.2K & \numprint{360} & 62 & 2.24 & 1.7e-4 & Top \\
Q&g\_gnutella-30 & 587.7K & \numprint{67}\% & 95.3K & 202.7K & \numprint{329} & 55 & 1.92 & 1.7e-4 & Top \\
R&g\_atp-gr-qc & 277.3K & \numprint{66}\% & 86K & 90.5K & \numprint{25} & 13 & 0.07 & 2.7e-4 & Misc \\
\colorrow
S&g\_bio-proteins & 1.8M & \numprint{69}\% & 106.9K & 666.3K & \numprint{1731} & 150 & 11.94 & 6.3e-4 & Int \\

\bottomrule
\end{tabular}
\vspace{-0.5cm}
\end{table*}

\sisetup{
  detect-weight=true,
  detect-family=true
}

\renewcommand{\arraystretch}{1}
\begin{table*}[!h]
\caption{
Engineering Data Set (\EngSet{}) statistics: updates ($\delta$), insertion percentage ($\delta_+$); maximum values for non-isolated vertices ($\widehat{n}$), edges ($\widehat{m}$) and pattern counts. }\label{tab:engeneering_set}

\centering
\scriptsize

\resizebox{\linewidth}{!}{
\begin{tabular}{l|ll|cc|cc|c|c|c|c|c|c|c}
\toprule

&ID&Graph & $\delta$ & $\delta_+$ & $\widehat{n}$ & $\widehat{m}$ & $|\tri|$ & $|\tPath|$ & $|\claw|$ & $|\paw|$ & $|\fCycle|$ & $|\dia|$ & $|\fClique|$ \\
\midrule
\multirow{6}{*}{\rotatebox{90}{\DSWiki}}&A&simple\_wiki & 1.6M & \numprint{72}\% & 99.5K & 690.7K & 3.8e+6 & 5.4e+10 & 7e+11 & 1.1e+10 & 4.9e+8 & 5.1e+8 & 5.8e+6 \\
&B&nl\_wiki & 2M & \numprint{79}\% & 253K & 1.2M & 3.1e+6 & 6.3e+10 & 4.3e+11 & 7.8e+9 & 6.2e+8 & 2.7e+8 & 5.1e+6 \\
&C&pl\_wiki & 2M & \numprint{80}\% & 213.3K & 1.2M & 5.7e+6 & 9.4e+10 & 2.1e+12 & 1.3e+10 & 1.6e+9 & 6.5e+8 & 2.5e+7 \\
&D&it\_wiki & 2M & \numprint{84}\% & 279.6K & 1.4M & 3.3e+6 & 1.4e+11 & 3.4e+13 & 3.2e+10 & 3.9e+9 & 3.6e+9 & 3.6e+6 \\
&E&fr\_wiki & 2M & \numprint{80}\% & 264.5K & 1.2M & 3.8e+6 & 6.7e+10 & 5.5e+11 & 7.8e+9 & 7.7e+8 & 3.8e+8 & 1.8e+7 \\
&F&de\_wiki & 2M & \numprint{82}\% & 240.4K & 1.3M & 4e+6 & 3.9e+10 & 1.1e+11 & 3.9e+9 & 4.4e+8 & 2.8e+8 & 2.6e+7 \\

\midrule
\multirow{6}{*}{\rotatebox{90}{\DSTMP}}&F&enron & 595K & \numprint{50}\% & 87K & 298K & 1.2e+6 & 4.7e+9 & 9.7e+9 & 9.2e+8 & 9.3e+7 & 8.4e+7 & 5.5e+6 \\
&H&askubuntu & 952K & \numprint{50}\% & 4.9K & 6.1K & \numprint{74585} & 1.7e+8 & 6.8e+8 & 1.5e+7 & 8.1e+5 & 4e+5 & \numprint{10742} \\
&I&superuser & 1.5M & \numprint{50}\% & 5.5K & 7.3K & 1.3e+5 & 3e+8 & 1.1e+9 & 2.9e+7 & 1.7e+6 & 8.8e+5 & \numprint{23639} \\
&J&wiki-talk & 2M & \numprint{51}\% & 49.7K & 64.7K & 2.9e+5 & 2.4e+9 & 3.3e+12 & 1.4e+8 & 6.3e+6 & 2.8e+6 & \numprint{95480} \\
&K&youtube-growth & 2M & \numprint{100}\% & 1.2M & 2M & 3.8e+5 & 4.8e+10 & 9.7e+12 & 2.5e+9 & 7.2e+7 & 2e+7 & 1.1e+5 \\
&L&stackoverflow & 2M & \numprint{51}\% & 16.1K & 50.6K & 5.1e+5 & 2.4e+9 & 2.5e+9 & 1.5e+8 & 1.2e+7 & 3.7e+6 & \numprint{77043} \\
\midrule
\multirow{6}{*}{\rotatebox{90}{\DScKron}}&
A&c\_answers & 2M & \numprint{55}\% & 74.2K & 185.4K & \numprint{46974} & 1.8e+9 & 2.2e+9 & 1.9e+7 & 1.2e+6 & 1e+5 & \numprint{671} \\
&B&c\_delicious & 2M & \numprint{56}\% & 93.7K & 225.9K & \numprint{20074} & 7.6e+8 & 5.6e+8 & 4.8e+6 & 3.5e+5 & \numprint{22573} & \numprint{134} \\
&I&c\_ca-dblp & 2M & \numprint{65}\% & 127.1K & 597.3K & \numprint{20519} & 9.7e+8 & 4.4e+8 & 3.3e+6 & 2.9e+5 & \numprint{11323} & \numprint{62} \\
&M&c\_web-notred. & 2M & \numprint{57}\% & 85.4K & 288.7K & \numprint{95243} & 4.1e+9 & 5.7e+9 & 5.8e+7 & 3.6e+6 & 3.7e+5 & \numprint{2670} \\
&O&c\_as-routeviews & 2M & \numprint{64}\% & 86.5K & 556.9K & 5.6e+5 & 3.3e+10 & 1.1e+11 & 1e+9 & 5.6e+7 & 7.8e+6 & \numprint{51752} \\
&R&c\_atp-gr-qc & 1.7M & \numprint{53}\% & 86.4K & 90.9K & \numprint{509} & 1.4e+7 & 5.9e+6 & \numprint{8884} & \numprint{699} & \numprint{6} & \numprint{0} \\

\midrule
\multirow{6}{*}{\rotatebox{90}{\DSgKron}}&
C&g\_epinions & 1.1M & \numprint{68}\% & 85.6K & 413K & 2e+5 & 6.1e+9 & 1.4e+10 & 1.6e+8 & 9.5e+6 & 1.5e+6 & \numprint{13510} \\
&B&g\_flickr & 853K & \numprint{68}\% & 77.8K & 304.7K & 1e+5 & 2.7e+9 & 5.2e+9 & 6.6e+7 & 3.9e+6 & 6e+5 & \numprint{5703} \\
&F&g\_blog-nat06all & 2M & \numprint{70}\% & 98.4K & 798K & 7.1e+5 & 2.4e+10 & 3.2e+10 & 6.9e+8 & 5.4e+7 & 8e+6 & \numprint{75866} \\
&G&g\_cit-hep-ph & 2M & \numprint{70}\% & 117.4K & 796.7K & 1.1e+5 & 3.9e+9 & 2.1e+9 & 2.9e+7 & 3.2e+6 & 1.8e+5 & \numprint{1191} \\
&N&g\_as-newman & 1.3M & \numprint{69}\% & 82.9K & 486K & 2.7e+5 & 1.2e+10 & 4.2e+10 & 3e+8 & 1.8e+7 & 2.1e+6 & \numprint{11730} \\
&S&g\_bio-proteins & 1.8M & \numprint{69}\% & 106.9K & 666.3K & \numprint{77387} & 6.2e+9 & 1.3e+10 & 4.1e+7 & 3.6e+6 & \numprint{94856} & \numprint{199} \\
\bottomrule
\end{tabular}
}
\end{table*}

\clearpage
\subsection{Tuning the Rebalancing Frequency}
\label{sec:exp:hhh:recomp_factor}

To study the effects of the \emph{rebalance factor} $\RebalanceFactor$, we consider \hhhohe{} 
with 
different $\RebalanceFactor\in\{1.5,2,3,5,8,13,21\}$
for
\origepstab{.2},
\softepstab{}{.2},
\lazyepstab{.2}, 
and 
\lateepstab{.2}.
For \lazyepstab{}, we additionally include $\RebalanceFactor\in\{1,1.25\}$ since for \lazyepstab{}
updating the threshold more frequently 
reflects the evolution of the graph more accurately.

\begin{figure}[t]

\includegraphics[width=0.49\textwidth]{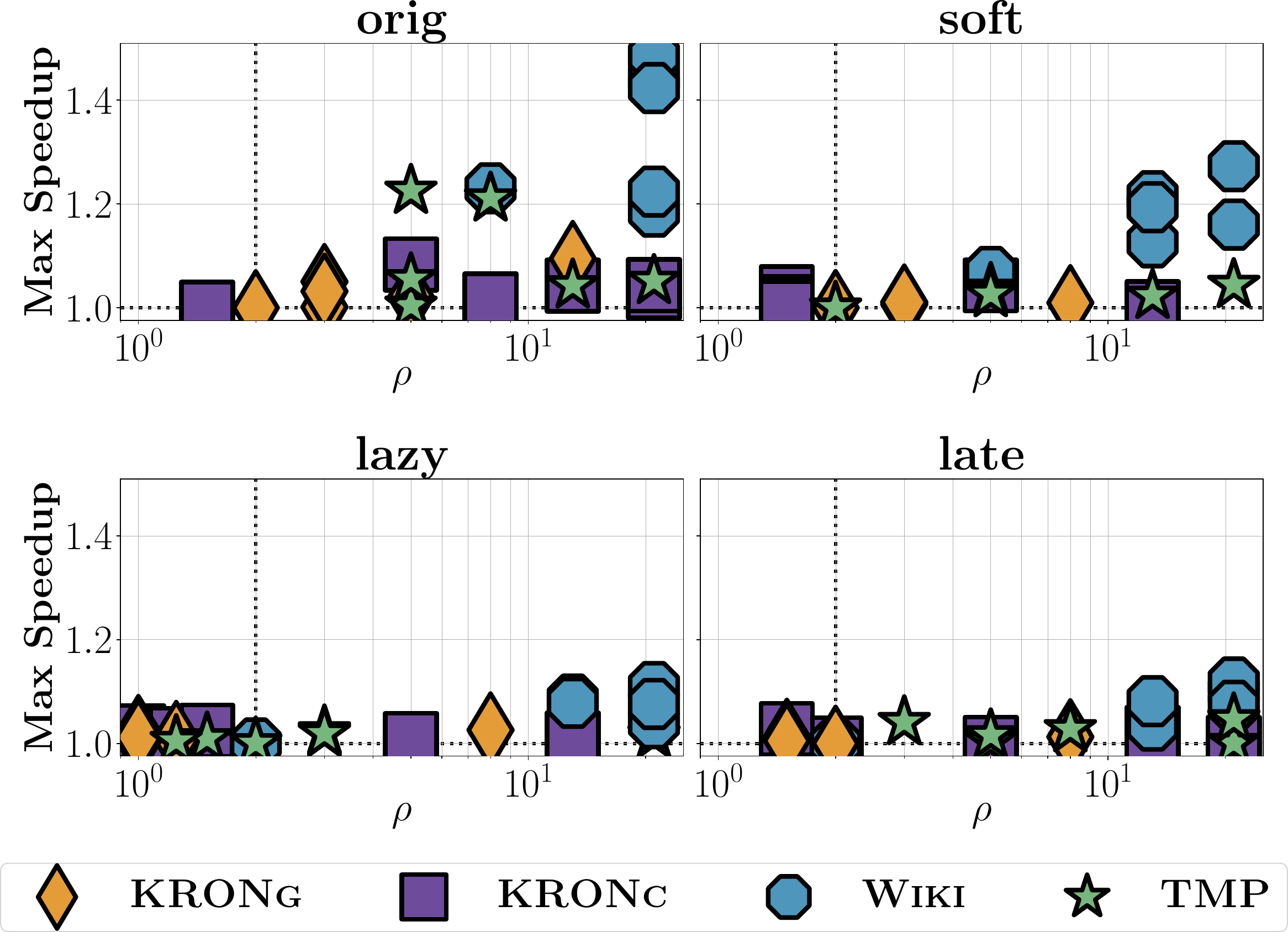}
\caption{Maximal speedup over using rebalance factor $\RebalanceFactor =2$ in different \epsT{} when using \hhhohe{} to count all patterns. Each graph corresponds to a single marker, where the horizontal position marks the recomputation factor achieving the maximal speedup and the vertical position shows the speedup over using  $\RebalanceFactor=2$.
} 
\label{exp:recomp_factor_speedup}
\end{figure}

Figure \ref{exp:recomp_factor_speedup} shows that varying $\RebalanceFactor$ can yield speedups of up to
$1.5\times$ (\origepstab{}),
$1.3\times$ (\softepstab{}), and
$1.2\times$ (\lazyepstab{}),
$1.1\times$ (\lateepstab{}{})
over the originally described $\RebalanceFactor = 2$ on individual graphs.
However, the best factor seems highly graph-specific.
In the geometric mean, $\RebalanceFactor=2$ was always within \SI{1}{\percent} for \lazyepstab{}, \SI{2}{\percent} for \lateepstab{}, \SI{3}{\percent} for \softepstab{} and \SI{9}{\percent} for \origepstab{} w.r.t.\ other values of $\RebalanceFactor$.
We therefore fix $\RebalanceFactor=2$ as a good compromise.

\begin{figure}[t]

\setlength{\belowcaptionskip}{-17pt}
\setlength{\abovecaptionskip}{2pt}
\includegraphics[width=0.49\textwidth]{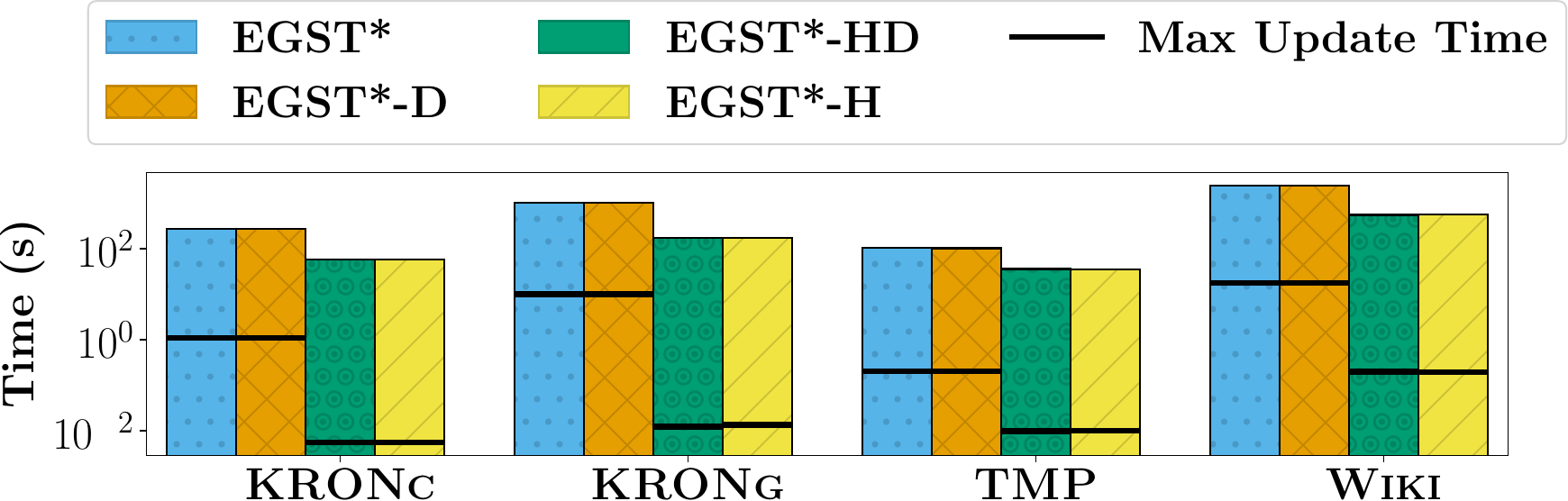}
\caption{
Running and maximal update time (log scale) of different \egst{} variants with \hInd{} using dataset \EngSet{} and counting all subgraphs simultaneously.
} 
\label{exp:egst_vari}
\end{figure}
\subsection{Proposed Optimizations for the \egst{} Algorithm}
\label{sec:egst:opt_versions:exp}
Since allowing any degree for anchors in the \Sc{3} auxiliary structure in \egst{} can incur high maintenance cost, we introduced \hegst{}, which uses \poptref{high-anchors-egst}.
Additionally, we consider \degst{}, which uses \poptref{direct-swap} to improve the partition swap of vertices. 
Figure \ref{exp:egst_vari} shows running and maximum update times for all versions when counting all subgraphs simultaneously on dataset \EngSet{}.

We compare the four combinations on \EngSet{}.
As expected, \hegst{} demonstrates clear improvements, achieving a speedup of
$4.3\times$
in terms of total runtime.
The maximum update time improves by
$130.1\times$.
The smaller improvement on \DSTMP{} graphs arises due to their sparsity, which  limits the number of found auxiliary structures, thereby reducing the impact of the optimization.
The larger gains appear on heavily growing graphs, where the number of auxiliary structures increases over time.
This results in a higher number of expensive capacity increases in hash maps, amplifying the effect of reducing the number of structures.
For \degst{}, we observe improvements of
$1.4\%$
in total runtime compared to \egst{}. This is because, in most cases, applying the edge update routine twice is more expensive than scanning the 2-hop neighborhood.
Differences in maximum update time are small (within $1\%$),
without a clear winner.
When combining both methods into \hdegst{}, we achieve the best overall results. \hdegst{} improves the maximum update time by
$132.9\times$
and the total runtime by
$4.3\times$
over \egst{}.

\clearpage

\subsection{Counting Single Patterns}
\label{sec:single_patterns}

\begin{figure}[b!]

\setlength{\belowcaptionskip}{-17pt}
\setlength{\abovecaptionskip}{-0pt}

\begin{tikzpicture}
  \node[anchor=south west, inner sep=0] (image) {\includegraphics[width=\linewidth]{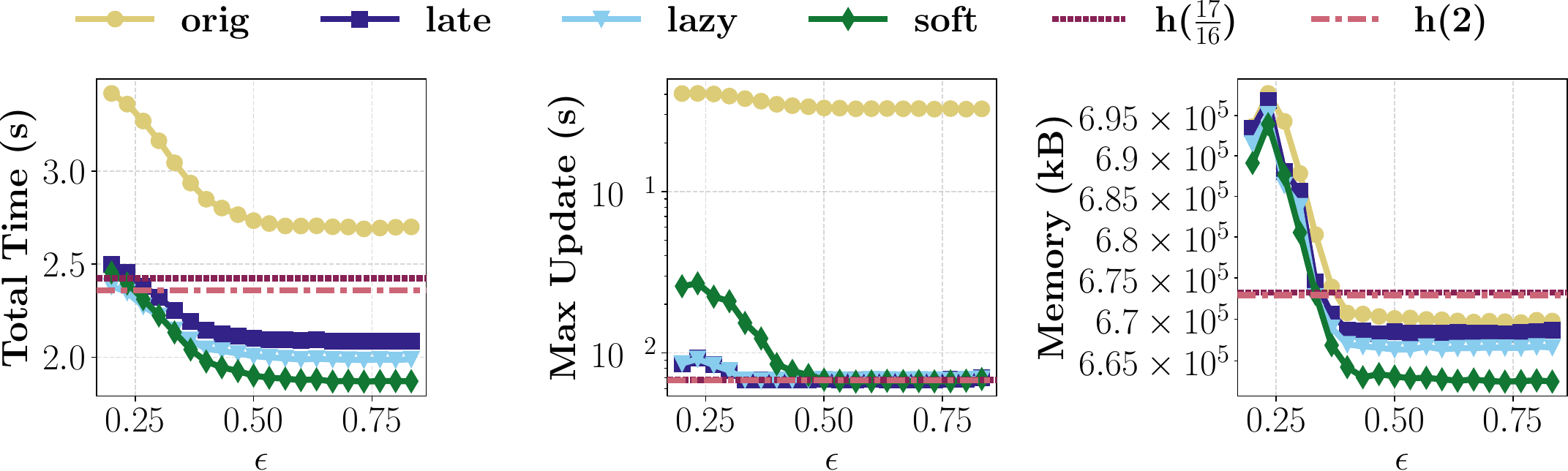}};
  \node[anchor=north west, yshift=-70pt] at (image.north west) {\small (a) Triangles \tri};
\end{tikzpicture}
\begin{tikzpicture}
  \node[anchor=south west, inner sep=0] (image) {\includegraphics[width=\linewidth]{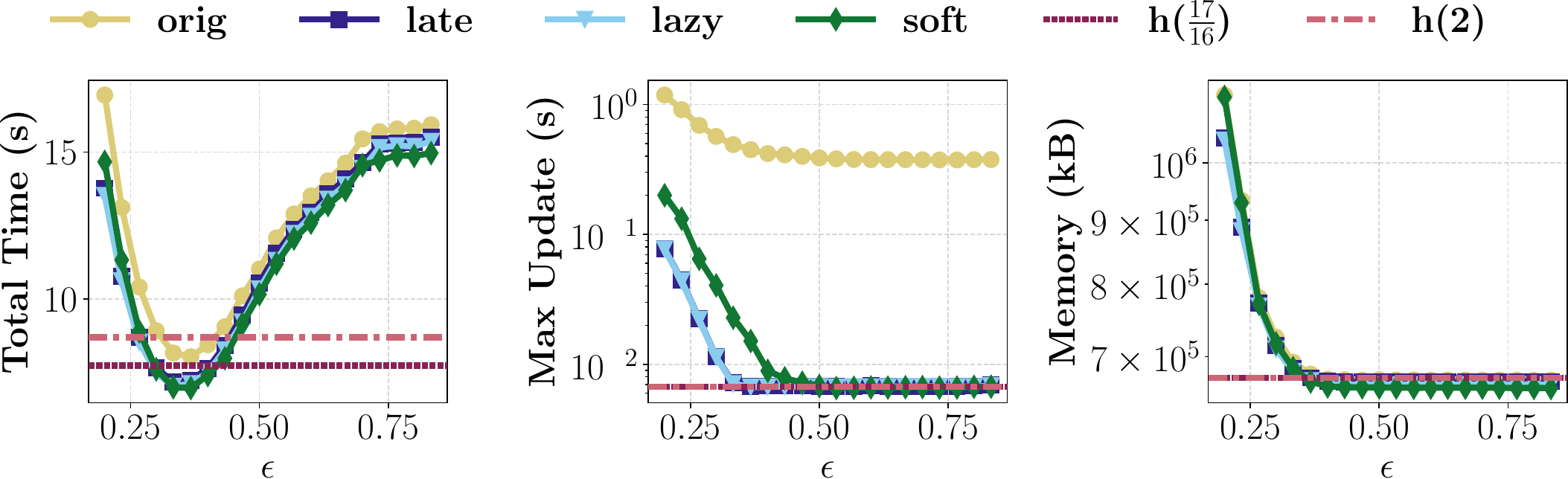}};
  \node[anchor=north west, yshift=-70pt] at (image.north west) {\small (b) 3-Paths \tPath};
\end{tikzpicture}
\begin{tikzpicture}
  \node[anchor=south west, inner sep=0] (image) {\includegraphics[width=\linewidth]{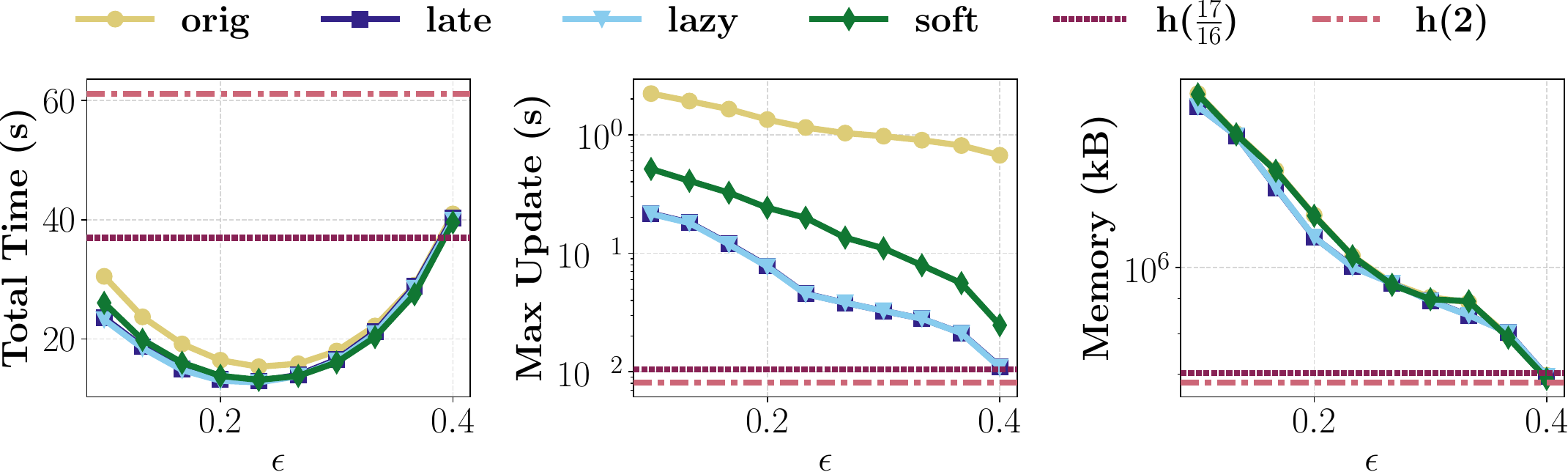}};
  \node[anchor=north west, yshift=-70pt] at (image.north west) {\small (c) Paws \paw};
\end{tikzpicture}
\begin{tikzpicture}
  \node[anchor=south west, inner sep=0] (image) {\includegraphics[width=\linewidth]{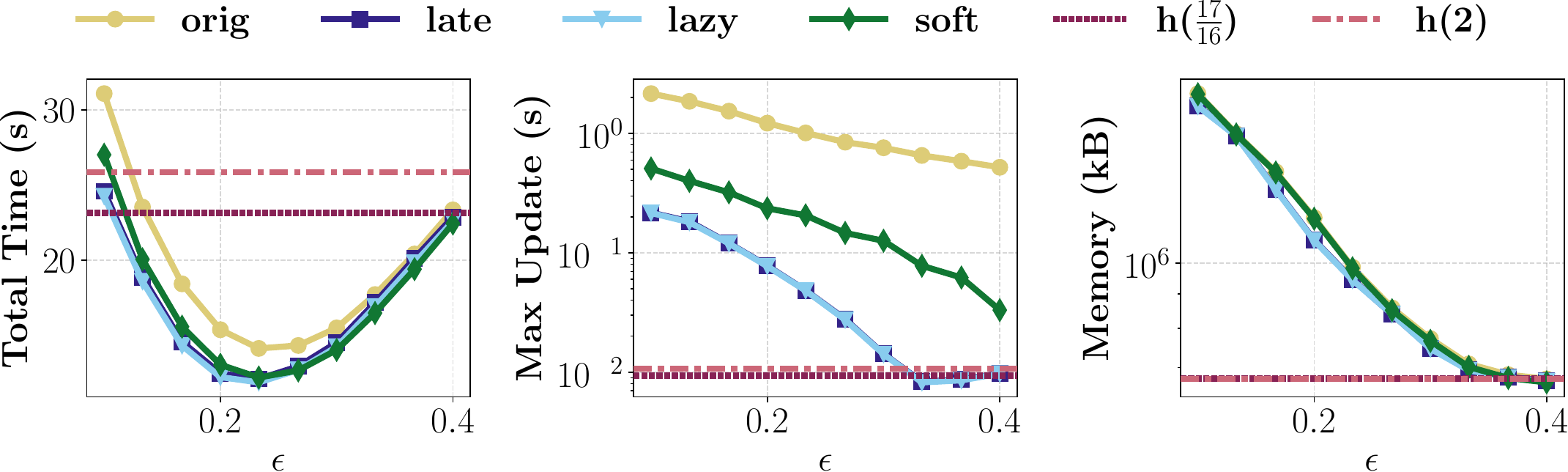}};
  \node[anchor=north west, yshift=-70pt] at (image.north west) {\small (d) 4-Cycles \fCycle};
\end{tikzpicture}
\begin{tikzpicture}
  \node[anchor=south west, inner sep=0] (image) {\includegraphics[width=\linewidth]{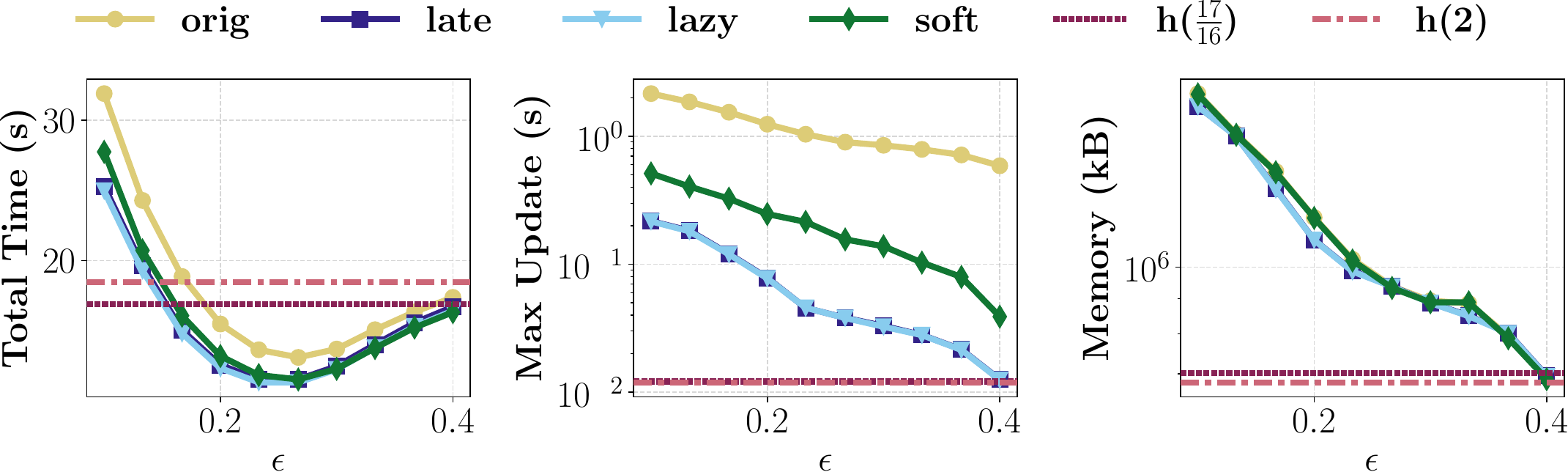}};
  \node[anchor=north west, yshift=-70pt] at (image.north west) {\small (e) Diamonds \dia};
\end{tikzpicture}
\begin{tikzpicture}
  \node[anchor=south west, inner sep=0] (image) {\includegraphics[width=\linewidth]{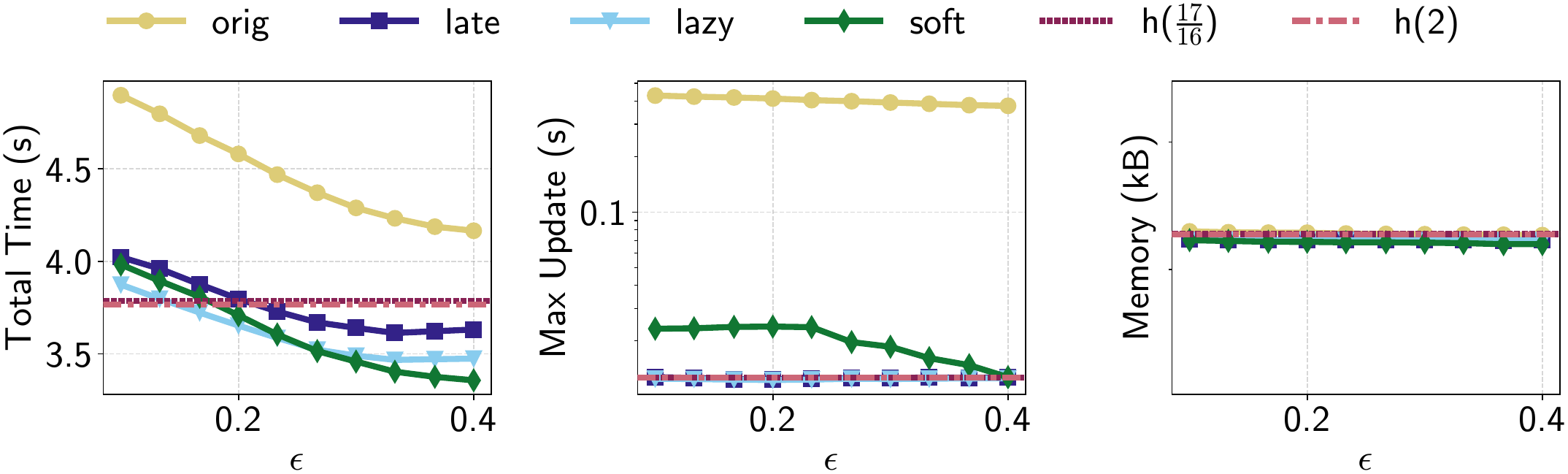}};
  \node[anchor=north west, yshift=-70pt] at (image.north west) {\small (f) 4-Cliques \fClique};
\end{tikzpicture}

\caption{
Performance of \hhhoh{} 
for different vertex partitions
w.r.t.\ $\varepsilon$ 
on \EngSet{} data
counting individual patterns. (We did not study claws here as these have constant update time)
} 
\vspace{1.8cm}
\label{exp:hhh_eps_single}
\end{figure}

\begin{figure}[b!]

\setlength{\belowcaptionskip}{-17pt}
\setlength{\abovecaptionskip}{-0pt}

\begin{tikzpicture}
  \node[anchor=south west, inner sep=0] (image) {\includegraphics[width=\linewidth]{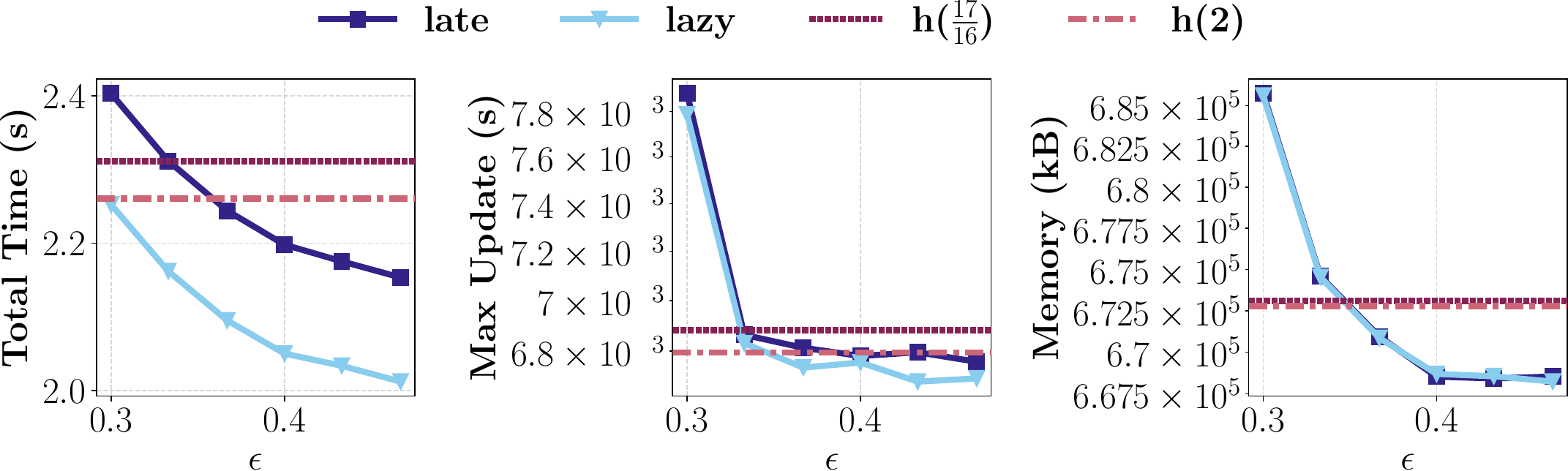}};
  \node[anchor=north west, yshift=-70pt] at (image.north west) {\small (a) Triangles \tri};
\end{tikzpicture}
\begin{tikzpicture}
  \node[anchor=south west, inner sep=0] (image) {\includegraphics[width=\linewidth]{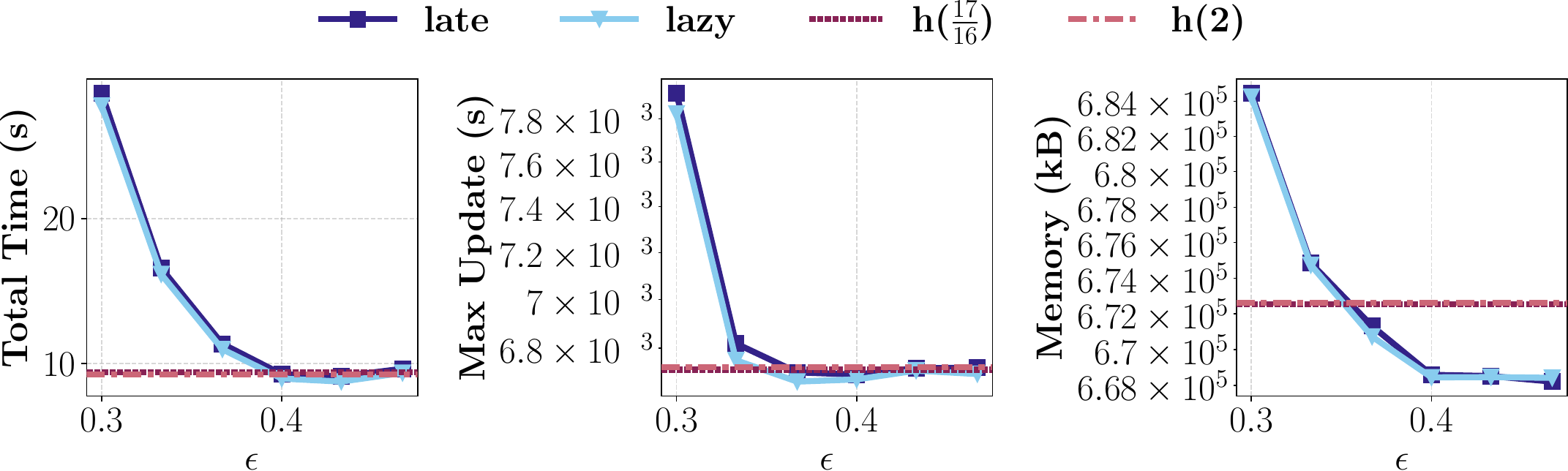}};
  \node[anchor=north west, yshift=-70pt] at (image.north west) {\small (b) 3-Paths \tPath};
\end{tikzpicture}
\begin{tikzpicture}
  \node[anchor=south west, inner sep=0] (image) {\includegraphics[width=\linewidth]{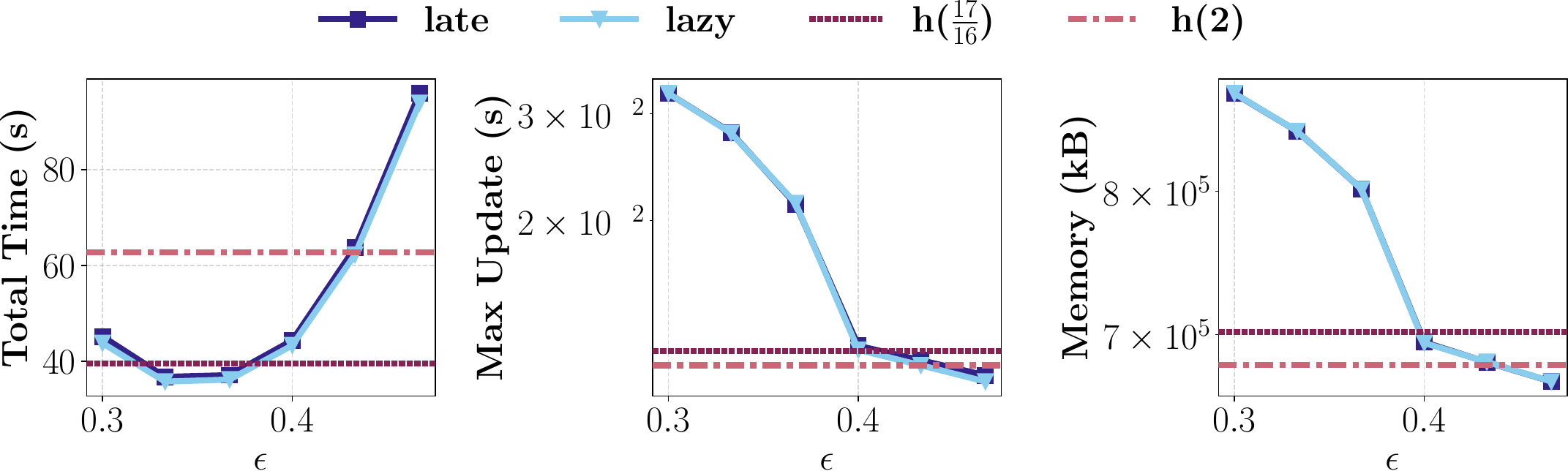}};
  \node[anchor=north west, yshift=-70pt] at (image.north west) {\small (c) Paws \paw};
\end{tikzpicture}
\begin{tikzpicture}
  \node[anchor=south west, inner sep=0] (image) {\includegraphics[width=\linewidth]{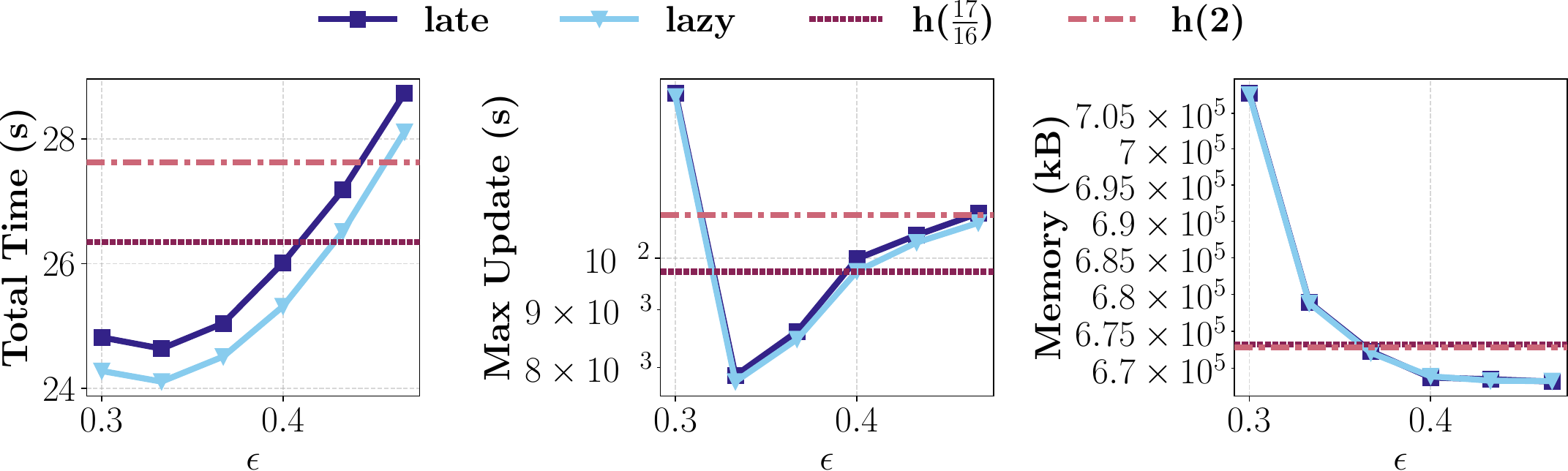}};
  \node[anchor=north west, yshift=-70pt] at (image.north west) {\small (d) 4-Cycles \fCycle};
\end{tikzpicture}
\begin{tikzpicture}
  \node[anchor=south west, inner sep=0] (image) {\includegraphics[width=\linewidth]{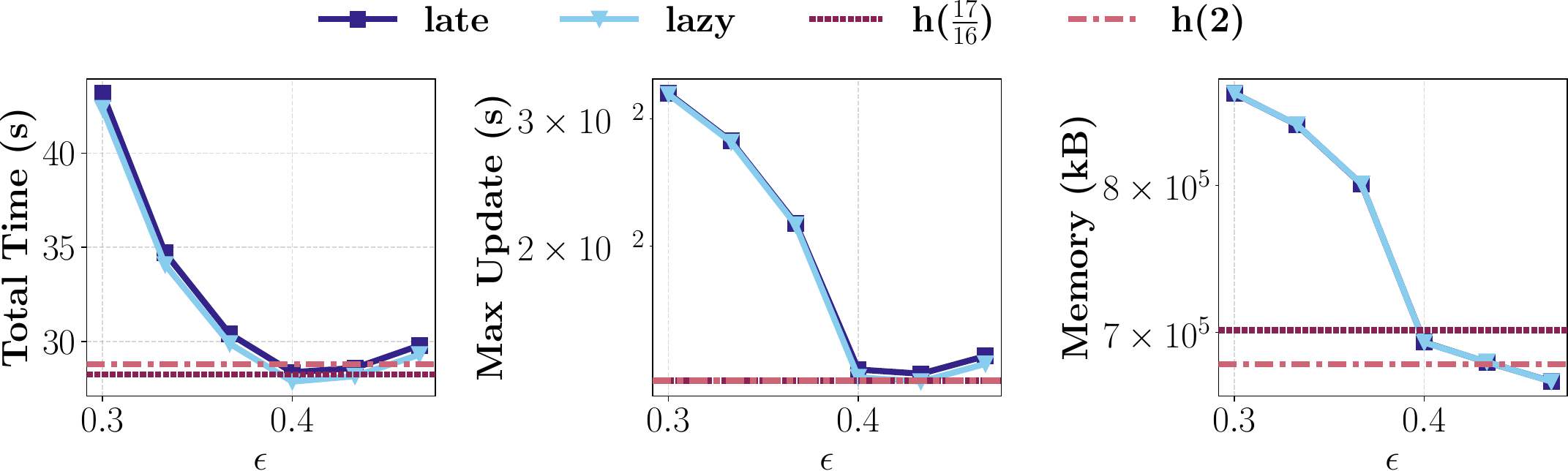}};
  \node[anchor=north west, yshift=-70pt] at (image.north west) {\small (e) Diamonds \dia};
\end{tikzpicture}
\begin{tikzpicture}
  \node[anchor=south west, inner sep=0] (image) {\includegraphics[width=\linewidth]{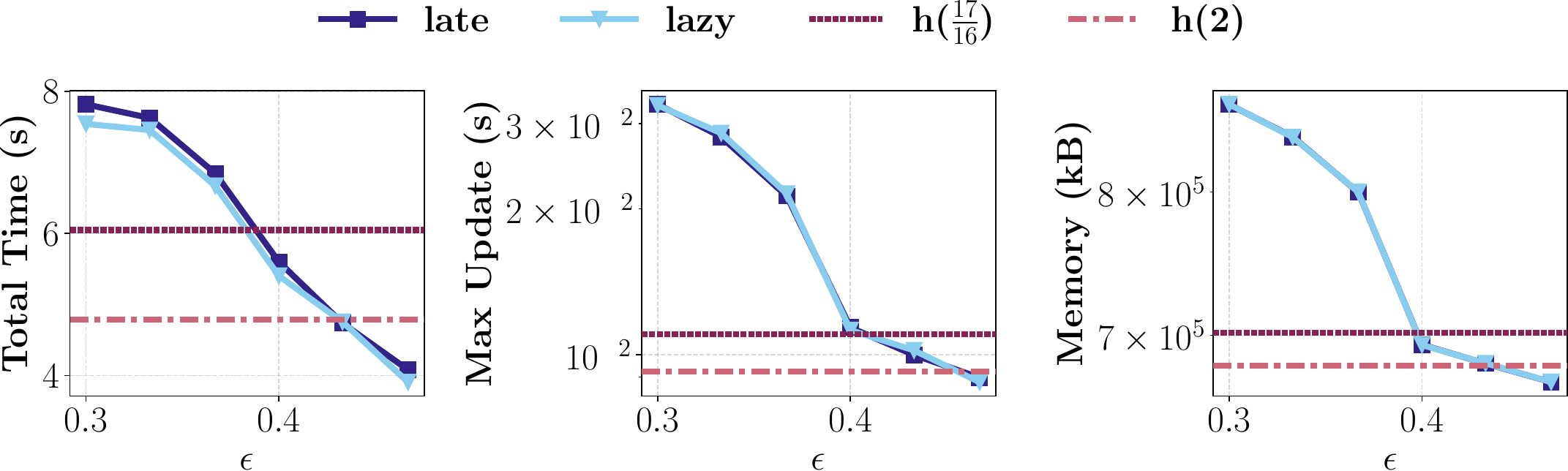}};
  \node[anchor=north west, yshift=-70pt] at (image.north west) {\small (f) 4-Cliques \fClique};
\end{tikzpicture}

\caption{
Performance of \hdegst{} 
for different vertex partitions
w.r.t.\ $\varepsilon$ 
on \EngSet{} data
counting individual patterns. (We did not study claws here as these have constant update time).
} 
\label{exp:egst_eps:single}
\end{figure}

\end{document}